%% file: main.tex
\documentclass{article}

\usepackage[final]{neurips_2024}

\usepackage{multirow}
\usepackage{tikz}
\usepackage[utf8]{inputenc} % allow utf-8 input
\usepackage[T1]{fontenc}    % use 8-bit T1 fonts
\usepackage{hyperref}       % hyperlinks
\usepackage{url}            % simple URL typesetting
\usepackage{booktabs}       % professional-quality tables
\usepackage{amsfonts}       % blackboard math symbols
\usepackage{nicefrac}       % compact symbols for 1/2, etc.
\usepackage{microtype}      % microtypography
\usepackage{xcolor}         % colors
\usepackage{amsmath}
\usepackage{amsthm,enumerate,amssymb}
\usepackage{thm-restate}
\usepackage{thmtools} 
\usepackage{mathtools}
\usepackage{wrapfig}
\usepackage{amsmath}
\usepackage{pifont}

\usepackage{comment}
\usepackage{graphics}
\usepackage{graphicx}
\usepackage{appendix}
\usepackage{algorithm}
\usepackage{enumitem}
\usepackage[noend]{algpseudocode}
\algdef{SE}[DOWHILE]{Do}{doWhile}{\algorithmicdo}[1]{\algorithmicwhile\ #1}%
\usepackage{bbm}
\algnewcommand{\IfThen}[2]{% \IfThenElse{<if>}{<then>}{<else>}
	\State \algorithmicif\ #1\ \algorithmicthen\ #2}

\usepackage{tikzscale} % To use \includegraphics with TikZ code

\usepackage[short, nocomma]{optidef} % optimization problems	

\newtheorem{definition}{Definition}

\newcommand{\sset}{\mathcal{S}}
\newcommand{\preg}{\mathcal{P}}

\newcommand{\A}{\mathcal{A}}

\newcommand{\X}{\mathcal{X}}

\newcommand{\Hsquare}{%
	\text{\fboxsep=-.2pt\fbox{\rule{0pt}{1.2ex}\rule{1.2ex}{0pt}}}%
}

\usepackage[textsize=tiny]{todonotes}

\newcommand*\circled[1]{\tikz[baseline=(char.base)]{\node[shape=circle,draw,inner sep=1pt] (char) {\scriptsize #1};}}

\algblockdefx[Execute]{Execute}{EndExecute}[1]{\textbf{until} #1 \textbf{run}}{\textbf{end execute}}
\makeatletter
\ifthenelse{\equal{\ALG@noend}{t}}%
{\algtext*{EndExecute}}
{}%
\makeatother

\algblockdefx[IfInline]{IfInline}{EndIfInline}[1]{\textbf{if} #1 \textbf{then} }{\textbf{end if}}
\makeatletter
\ifthenelse{\equal{\ALG@noend}{t}}%
{\algtext*{EndIfInline}}
{}%
\makeatother

\DeclareMathOperator*{\argmax}{arg\,max}
\DeclareMathOperator*{\argmin}{arg\,min}

\definecolor{mygreen}{rgb}{0.0, 0.5, 0.0}
\definecolor{myorange}{rgb}{0.55, 0.62, 1}
\newcommand{\nb}[3]{{\colorbox{#2}{\bfseries\sffamily\scriptsize\textcolor{white}{#1}}}{\textcolor{#2}{\sf\small\textsf{#3}}}}

\newcommand{\bollo}[1]{\nb{Bollo}{red}{#1}}

\allowdisplaybreaks

\title{Online Bayesian Persuasion Without a Clue}

\author{
	Francesco Bacchiocchi \\
	Politecnico di Milano\\
	\texttt{francesco.bacchiocchi@polimi.it} \\
	\And
	Matteo Bollini \\
	Politecnico di Milano\\
	\texttt{matteo.bollini@polimi.it} \\
	\AND
	Matteo Castiglioni \\
	Politecnico di Milano\\
	\texttt{matteo.castiglioni@polimi.it} \\
	\And
	Alberto Marchesi \\
	Politecnico di Milano\\
	\texttt{alberto.marchesi@polimi.it} \\
	\And
	Nicola Gatti \\
	Politecnico di Milano\\
	\texttt{nicola.gatti@polimi.it}  \\}
\begin{document}
\maketitle

\input{abstract}

\input{introduction}

\input{related_works}
\input{preliminaries}

\input{two_signlas}

\input{regret}
\input{find_partition}

\input{find_ss}
\input{lower_bounds_regret}
\input{sample_complexity}
\bibliographystyle{plainnat}
\bibliography{references}
\newpage
\appendix
\input{appendix_additional_related_works}
\input{appendix_additional_preliminaries}

\input{appendix_find_search_space}
\input{appendix_regret_th}

\input{appendix_find_signaling_scheme}
\input{appendix_find_partition}

\input{appendix_lb_regret}

\input{appendix_ub_sc}

\input{appendix_lb_sl}
%\input{appendix_h_representation}
%\input{appendix_extra_lemmas}
\newpage
\input{checklist}

\end{document}

%% file: abstract.tex
\begin{abstract}
	We study \emph{online Bayesian persuasion} problems in which an informed sender repeatedly faces a receiver with the goal of influencing their behavior through the provision of payoff-relevant information.
	Previous works assume that the sender has \emph{knowledge about} either the prior distribution over states of nature or receiver's utilities, or both.
	%
	% Previous works modeled sender's uncertainty about the receiver by assuming that the latter is characterized by a \emph{type} adversarially selected from a finite and \emph{known} set.
	%
	We relax such unrealistic assumptions by considering settings in which the sender does \emph{not know anything} about the prior and the~receiver.
	%
	% Moreover, we also relax the stringent assumption that the sender knows the prior distribution over states of nature, made by previous~works.
	%
	We design an algorithm that achieves sublinear---in the number of rounds---regret with respect to an optimal signaling scheme, and we also provide a collection of lower bounds showing that the guarantees of such an algorithm are tight.
	Our algorithm works by searching a suitable space of signaling schemes in order to learn receiver's best responses.
	To do this, we leverage a non-standard representation of signaling schemes that allows to cleverly overcome the challenge of \emph{not} knowing anything about the prior over states of nature and receiver's utilities.
	Finally, our results also allow to derive lower/upper bounds on the \emph{sample complexity} of learning signaling schemes in a related Bayesian persuasion PAC-learning problem. 
\end{abstract}

%% file: introduction.tex
\section{Introduction}

\emph{Bayesian persuasion} has been introduced by~\citet{kamenica2011bayesian} to model how strategically disclosing information to decision makers influences their behavior.  
% how individuals take decisions based on information received from others, and how this can be used to influence their behavior.
%
Over the last years, it has received a terrific attention in several fields of science, since it is particularly useful for understanding strategic interactions involving individuals with different levels of information, which are ubiquitous in the real world.
%
%As a consequence, Bayesian persuasion has been applied in several settings, such as, \emph{e.g.}, online advertising~\citep{emek2014signaling,badanidiyuru2018targeting,ijcai2022p6,Shipra2023},
As a consequence, Bayesian persuasion has been applied in several settings, such as online advertising~\citep{emek2014signaling,badanidiyuru2018targeting,ijcai2022p6,Shipra2023},
voting~\citep{alonso2016persuading,castiglioni2019persuading,semipublic}, traffic routing~\citep{vasserman2015implementing,bhaskar2016hardness,castiglioni2020signaling}, recommendation systems~\citep{cohen2019optimal,mansour2022bayesian},
%\bollo{\citet{mansour2022bayesian} non mi sembra riguardi i recommender systems: la sua ``recommendaion'' è l'azione del diresct signaling scheme}
security~\citep{rabinovich2015information,xu2016signaling}, e-commerce~\citep{bro2012send,castiglioni2022signaling}
%
% marketing~\citep{babichenko2017algorithmic,candogan2019persuasion,romano2021online},
%
medical research~\citep{kolotilin2015experimental}, 
and financial regulation~\citep{goldstein2018stress}.

In its simplest form, Bayesian persuasion involves a \emph{sender} observing some information about the world, called \emph{state of nature}, and a \emph{receiver} who has to take an action.
Agents' utilities are misaligned, but they both depend on the state of nature and receiver's action.
Thus, sender's goal is to devise a mechanism to (partially) disclose information to the receiver, so as to induce them to take a favorable action.
This is accomplished by committing upfront to a \emph{signaling scheme}, encoding a randomized policy that defines how to send informative signals to the receiver based on the observed state.

Classical Bayesian persuasion models (see, \emph{e.g.},~\citep{dughmi2016algorithmic,dughmi2017algorithmicExternalities,xu2020tractability}) rely on rather stringent assumptions that considerably limit their applicability in practice. 
% limit their effectiveness when being applied in practice.
%
Specifically, they assume that the sender perfectly knows the surrounding environment, including receiver's utilities and the probability distribution from which the state of nature is drawn, called \emph{prior}.
This has motivated a recent shift of attention towards Bayesian persuasion models that incorporate concepts and ideas from \emph{online learning}, with the goal of relaxing some of such limiting assumptions.
However, existing works only partially fulfill this goal, as they still assume some knowledge of either the prior (see, \emph{e.g.},~\citep{Castiglioni2020online,Castiglioni2021MultiReceiver,castiglioni2023regret,Babichenko2022,Bernasconi2023Optimal}) or receiver's utilities (see, \emph{e.g.},~\citep{Zu2021,Bernasconi2022Sequential,jibang2022}).

\subsection{Original contributions}

We address---for the first time to the best of our knowledge---Bayesian persuasion settings where the sender \emph{has no clue} about the surrounding environment.
%, as they do \emph{not} know anything about the prior distribution over states of nature and receiver's utilities.
%
In particular, we study the online learning problem faced by a sender who repeatedly interacts with a receiver over multiple rounds, \emph{without} knowing anything about both the prior distribution over states of nature and receiver's utilities.
At each round, the sender commits to a signaling scheme, and, then, they observe a state realization and send a signal to the receiver based on that.
After each round, the sender gets \emph{partial feedback}, namely, they only observe the best-response action played by the receiver in that round.
In such a setting, the goal of the sender is to minimize their \emph{regret}, which measures how much utility they lose with respect to committing to an optimal (\emph{i.e.}, utility-maximizing) signaling scheme in every round.

We provide a learning algorithm that achieves regret of the order of $\widetilde{\mathcal{O}}(\sqrt{T})$, where $T$ is the number of rounds.
%in the \emph{partial feedback} setting, namely, when only the action played by the receiver is revealed after each round.
%
We also provide lower bounds showing that the regret guarantees attained by our algorithm are tight in $T$ and in the parameters characterizing the Bayesian persuasion instance, \emph{i.e.}, the number of states of nature $d$ and that of receiver's actions $n$.
Our algorithm implements a sophisticated \emph{explore-then-commit} scheme, with exploration being performed in a suitable space of signaling schemes so as to learn receiver's best responses \emph{exactly}.
This is crucial to attain tight regret guarantees, and it is made possible by employing a non-standard representation of signaling schemes, which allows to cleverly overcome the challenging lack of knowledge about both the prior and receiver's utilities.

Our results also allow us to derive lower/upper bounds on the \emph{sample complexity} of learning signaling schemes in a related Bayesian persuasion PAC-learning problem, where the goal is to find, with high probability, an approximately-optimal signaling scheme in the minimum possible number of rounds.
%while using the minimum possible number of rounds. 

%\bac{see also follow-up works, si capisce palesemetne che sei autore di questo paper. Soprattutto il paper tree-form, qui non centra li non si conosce il prior}

%% file: related_works.tex
\subsection{Related works}\label{sec:related_works}

\citet{Castiglioni2020online} were the first to introduce \emph{online learning} problems in Bayesian persuasion scenarios, with the goal of relaxing sender's knowledge about receiver's utilities (see also follow-up works~\citep{Castiglioni2021MultiReceiver,castiglioni2023regret,Bernasconi2023Optimal}).
In their setting, sender's uncertainty is modeled by means of an \emph{adversary} selecting a receiver's \emph{type} at each round, with types encoding information about receiver's utilities.
However, in such a setting, the sender still needs knowledge about the finite set of possible receiver's types and their associated utilities, as well as about the prior.

A parallel research line has focused on relaxing sender's knowledge about the prior.
\citet{Zu2021} study online learning in a repeated version of Bayesian persuasion.
Differently from this paper, they consider the sender's learning problem of issuing \emph{persuasive} action recommendations (corresponding to signals in their case), where persuasiveness is about correctly incentivizing the receiver to actually follow such recommendations.   
They provide an algorithm that attains sublinear regret while being persuasive at every round with high probability, despite having \emph{no} knowledge of the prior.
\citet{jibang2022,gan2023sequential,bacchiocchi2024markov} achieve similar results for Bayesian persuasion in episodic Markov decision processes, while~\citet{Bernasconi2022Sequential} in non-Markovian environments.
%by relaxing the persuasiveness requirement.
%
All these works crucially differ from ours, since they strongly rely on the assumption that receiver's utilities are known to the sender, which is needed in order to meet persuasiveness requirements.
As a result, the techniques employed in such works are fundamentally different from ours as well.

Finally, learning receiver's best responses exactly (a fundamental component of our algorithm) is related to learning in Stackelberg games~\citep{letchford2009learning,Peng2019,bacchiocchi2024sample}.
For more details on these works and other related works, we refer the reader to Appendix~\ref{appendix:additional_rleated}.

%% file: preliminaries.tex
\section{Preliminaries}

In Section~\ref{sec:prelimin_persuasion}, we introduce all the needed ingredients of the classical \emph{Bayesian persuasion} model by~\citet{kamenica2011bayesian}, while, in the following Section~\ref{sec:prelimin_learning}, we formally define the Bayesian persuasion setting faced in the rest of the paper and its related \emph{online learning} problem.
%
% \emph{online learning} Bayesian persuasion problem faced in the rest of the paper.

\subsection{Bayesian persuasion}\label{sec:prelimin_persuasion}

% An informed \emph{sender} aims at (partially) disclosing information to a self-interested \emph{receiver}, in order to influence their behavior.
%
A Bayesian persuasion instance is characterized by a finite set $\Theta\coloneqq \{\theta_i\}_{i=1}^d$ of $d$ states of nature and a finite set $\A \coloneqq  \{a_i\}_{i=1}^n$ of $n$ receiver's actions.
Agents' payoffs are encoded by utility functions $u, u^\text{s}: \Theta \times \A \to [0,1]$, with $u_\theta(a) \coloneqq u(\theta,a)$, respectively $u^\text{s}_\theta(a) \coloneqq u^\text{s}(\theta,a)$, denoting the payoff of the receiver, respectively the sender, when action $a \in \A$ is played~in~state~$\theta \in \Theta$.
The sender observes a state of nature drawn from a commonly-known \emph{prior} probability distribution $\mu \in\textnormal{int}(\Delta_\Theta)$,\footnote{Given a finite set $X$, we denote by $\Delta_X$ the set of all the probability distributions over $X$.} with $\mu_\theta \in (0,1]$ denoting the probability of~$\theta \in \Theta$.
To disclose information about the realized state, the sender can publicly commit upfront to a \emph{signaling scheme} $\phi:\Theta \to \Delta_\sset$, which defines a randomized mapping from states of nature to signals being sent to the receiver, for a finite set $\sset$ of signals.
For ease of notation, we let $\phi_\theta \coloneqq \phi(\theta)$ be the probability distribution over signals prescribed by $\phi$ when the the sate of nature is $\theta \in \Theta$, with $\phi_\theta(s) \in [0,1]$ denoting the probability of sending signal $s \in \sset$.

The sender-receiver interaction goes as follows: (1) the sender commits to a signaling scheme $\phi$; (2) the sender observes a state of nature $\theta \sim \mu$ and sends a signal $s \sim \phi_\theta$ to the receiver; (3) the receiver updates their belief over states of nature according to {\em Bayes rule}; and (4) the receiver plays a best-response action $a \in \A$, with sender and receiver getting payoffs $u_\theta(a)$ and $u^\text{s}_\theta(a)$, respectively.
Specifically, an action is a \emph{best response} for the receiver if it maximizes their expected utility given the belief computed in step~(3) of the interaction.
Formally, given a signaling scheme $\phi:\Theta \to \Delta_\sset$ and a signal $s \in \sset$, we let $\A^\phi(s) \subseteq \A$ be the set of receivers' best-response actions, where:
\begin{equation}\label{eq:br_set}
	\A^\phi(s)  \coloneqq \left\{ a_i \in \A \mid \sum_{\theta \in \Theta} \mu_\theta \phi_\theta(s) u_\theta(a_i) \geq \sum_{\theta \in \Theta} \mu_\theta \phi_\theta(s) u_\theta(a_j) \quad \forall a_j \in \A  \right\}.
\end{equation}
As customary in the literature on Bayesian persuasion (see, \emph{e.g.},~\citep{dughmi2016algorithmic}), we assume that, when the receiver has multiple best responses available, they break ties in favor of the sender.
In particular, we let $a^\phi(s)$ be the best response that is actually played by the receiver when observing signal $s \in \sset$ under signaling scheme $\phi$, with $a^\phi(s) \in \argmax_{a \in \A^\phi(s)} \sum_{\theta \in \Theta} \mu_\theta \phi_\theta(s) u^\text{s}_\theta(a)$.

The goal of the sender is to commit to an \emph{optimal} signaling scheme, namely, a $\phi:\Theta \to \Delta_\sset$ that maximizes sender's expected utility, defined as $u^\text{s}(\phi)\coloneqq \sum_{s \in \sset} \sum_{\theta \in \Theta} \mu_{\theta} \phi_{\theta}(s) u^\text{s}_{\theta}(a^\phi(s))$.
In the following, we let $\textnormal{OPT} \coloneqq \max_{\phi} u^\text{s}(\phi)$ be the optimal value of sender's expected utility.
Moreover, given an additive error $\gamma \in (0,1)$, we say that a signaling scheme $\phi$ is $\gamma$-optimal if $u^\text{s}(\phi) \geq \text{OPT} - \gamma$.

\subsection{Learning in Bayesian persuasion}\label{sec:prelimin_learning}

We study settings in which the sender \emph{repeatedly} interacts with the receiver over multiple rounds, with each round involving a one-shot Bayesian persuasion interaction (as described in Section~\ref{sec:prelimin_persuasion}).
We assume that the sender has \emph{no} knowledge about both the prior $\mu$ and receiver's utility $u$, and that the only feedback they get after each round is the best-response action played by the receiver.
%
%The sender has \emph{no} knowledge about the prior distribution $\mu$, as well as about receiver's utility function $u$ and action set $\A$.
%%
%The only knowledge available to the sender is the set $\Theta$.\alb{Conosciamo $\A$.}

At each round $t \in [T]$,\footnote{We denote by $[n] \coloneqq \left\{ 1, \ldots, n \right\}$ the set of the first $n \in \mathbb{N}$ natural numbers.} the sender commits to a signaling scheme $\phi_t : \Theta \to \Delta_\sset$ and observes a state of nature $\theta^t \sim \mu$.
Then, they draw a signal  $s^t \sim \phi_{t, \theta^t}$ and send it to the receiver, who plays a best-response action $a^t \coloneqq a^{\phi_t}(s^t)$.
Finally, the sender gets payoff $u^\text{s}_t \coloneqq u_{\theta^t}^\text{s} (a^t)$ and observes a \emph{feedback} consisting in the action $a^t$ played by the receiver.
%
%gets a \emph{feedback} consisting in their payoff $u^\text{s}_t \coloneqq u_{\theta^t}^\text{s} (a^t)$ and the best-response action $a^t$ played by the receiver. \bac{cambiare con : Finally, the sender gets a \emph{feedback} consisting the best-response action $a^t$ played by the receiver and achieves utility $u^\text{s}_t \coloneqq u_{\theta^t}^\text{s} (a^t)$ and .}
%
The goal of the sender is to learn how to maximize their expected utility while repeatedly interacting with the receiver.
When the sender commits to a sequence $\{ \phi_t \}_{t \in [T]}$ of signaling schemes, their performance over the $T$ rounds is measured  by means of the following notion of \emph{cumulative (Stackelberg) regret}:
%
%we measure sender's performance over the $T$ rounds, when they commit to a sequence $\{ \phi_t \}_{t \in [T]}$ of signaling schemes, by means of the following notion of \emph{cumulative (Stackelberg) regret}: 
\[
	R_T( \{ \phi_t \}_{t \in [T]}) \coloneqq T \cdot \text{OPT} - \mathbb{E}\left[\sum_{t=1}^T u^\text{s} ( \phi_t )\right],
\]
%where the expectation is taken over the probability measure induced by the algorithm randomness and the stochasticity of the environment.\alb{check}
where the expectation is with respect to the randomness of the algorithm.
In the following, for ease of notation, we omit the dependency on $\{ \phi_t \}_{t \in [T]}$ from the cumulative regret, by simply writing~$R_T$.
Then, our goal is to design \emph{no-regret} learning algorithms for the sender, which prescribe a sequence of signaling schemes $\phi_t$ that results in the regret $R_T$ growing sublinearly in $T$, namely $R_T = o(T)$.

%% file: two_signlas.tex
\section{Warm-up: A single signal is all you need}\label{sec:slices}

In order to design our learning algorithm in Section~\ref{sec:algo_noregret}, we exploit a non-standard representation of signaling schemes, which we introduce in this section.
Adopting such a representation is fundamental to be able to learn receiver's best responses \emph{without} any knowledge of both the prior $\mu$ and receiver's utility function $u$.  
The crucial observation that motivates its adoption is that receiver's best responses $a^\phi(s)$ only depend on the components of $\phi$ associated with $s \in \sset$, namely $\phi_\theta(s)$ for $\theta \in \Theta$.
Thus, in order to learn them, it is sufficient to learn how $a^\phi(s)$ varies as a function of such components.

The signaling scheme representation introduced in this section revolves around the concept of \emph{slice}.
\begin{definition}[Slice]
	Given a signaling scheme $\phi : \Theta \to \Delta_\sset$, the \emph{slice of $\phi$ with respect to signal $s \in \sset$} is the $d$-dimensional vector $x \in [0,1]^d$ with components $x_\theta \coloneqq \phi_\theta(s)$ for $\theta \in \Theta$.
\end{definition}
%
% Such a vector belongs to the hypercube~$[0,1]^d$.
%
% NORMALIZED SLICES
%
% Moreover, we define the \emph{normalized slice of $\phi$ with respect to signal $s \in \sset$} as the $d$-dimensional vector whose components are defined as $ {\phi_\theta(s)} / \Lambda$ for $\theta \in \Theta$, where $\Lambda \coloneqq \sum_{\theta' \in \Theta} \phi_{\theta'}(s)$.
%
In the following, we denote by $\X^{\Hsquare} \coloneqq [0,1]^d$ the set of \emph{all} the possible slices of signaling schemes.
%, and, given any slice $x \in \X^{\Hsquare}$, we denote by $x_\theta$ the component corresponding to state of nature $\theta \in \Theta$.
%
Moreover, we let $\X^{\triangle}$ be the set of \emph{normalized} slices, which is simply obtained by restricting slices to lie in the $(d-1)$-dimensional simplex.
%
% Notice that normalized slices are simply obtained by projecting (unnormalized) ones onto the $d$-dimensional simplex.
%
Thus, it holds that $\X^{\triangle} \coloneqq \left\{ x \in [0,1]^d \mid \sum_{\theta \in \Theta} x_\theta = 1 \right\} $.

\begin{wrapfigure}[14]{R}{0.44\textwidth}
	\vspace{-.8cm}
	\begin{minipage}{0.44\textwidth}
		\begin{figure}[H]
			\centering
			%\includegraphics[width=0.75\linewidth]{Tikz_pictures/slicing_signaling_schemes.tikz} % Non so perchè ma non funziona
			%			\resizebox{0.75\linewidth}{!}{\input{Tikz_pictures/slicing_signaling_schemes.tikz}\unskip}
			%\vspace{-.2cm}
			\resizebox{0.65\linewidth}{!}{\input{Tikz_pictures/slicing_signaling_schemes_2.tikz}}
%			\resizebox{0.65\linewidth}{!}{\input{Tikz_pictures/slicing_signaling_schemes.tikz}}
			%\vspace{-.3cm}
			\caption{Representation of sets $\mathcal{X}^\square(a_i)$ and $\mathcal{X}^\triangle(a_i)$ for an instance with $d=2$ states of nature and $n=3$ receivers' actions.}
			\label{fig:cones}
		\end{figure}
	\end{minipage}
\end{wrapfigure}
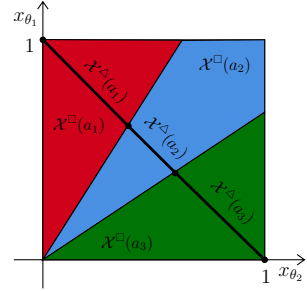

Next, we show that receiver's actions induce particular coverings of the sets $\X^{\Hsquare}$ and $\X^\triangle$, which also depend on both the prior $\mu$ and receiver's utility $u$.
First, we introduce $\mathcal{H}_{ij} \subseteq \mathbb{R}^d$ to denote the halfspace of slices under which action $a_i \in A$ is (weakly) better than action $a_j \in A$ for the receiver.
%
%Formally:
%
\[
	\mathcal{H}_{ij} \hspace{-0.5mm} \coloneqq \hspace{-0.5mm}   \left\{ \hspace{-0.2mm} \hspace{-0.3mm}  x \in \mathbb{R}^{d}  \mid  \sum_{\theta \in \Theta}  x_\theta \mu_\theta \big( u_\theta(a_i) \hspace{-0.3mm}  - \hspace{-0.3mm} u_\theta(a_j) \big)  \hspace{-0.4mm} \geq  \hspace{-0.4mm} 0 \hspace{-0.2mm} \right\} \hspace{-0.6mm}.
\]
Moreover, we denote by $H_{ij} \coloneqq \partial \mathcal{H}_{ij}$ the hyperplane constituting the boundary of the halfspace $\mathcal{H}_{ij}$, which we call the \emph{separating hyperplane} between actions $a_i$ and $a_j$.\footnote{We let $\partial \mathcal{H}$ be the boundary hyperplane of halfspace $\mathcal{H} \subseteq \mathbb{R}^d$.	Notice that $H_{ij}$ and $H_{ji}$ actually refer to the same hyperplane. In this paper, we use both names for ease of presentation.}
Then, for every $a_i \in \A$, we introduce the polytopes $\X^{\Hsquare} (a_i) \subseteq\X^{\Hsquare} $ and $\X^\triangle(a_i) \subseteq \X^\triangle$:
\[
	\X^{\Hsquare} (a_i) \coloneqq \X^{\Hsquare} \cap \Bigg(  \bigcap_{a_j \in \A: a_i \neq a_j} \mathcal{H}_{ij} \Bigg) \quad \text{and} \quad  \X^\triangle(a_i) \coloneqq \X^\triangle \cap \Bigg(  \bigcap_{a_j \in \A: a_i \neq a_j} \mathcal{H}_{ij} \Bigg).
\]
Clearly, the sets $\X^{\Hsquare}(a_i)$, respectively $\X^\triangle(a_i)$, define a cover of $\X^{\Hsquare}$, respectively $\X^\triangle$.\footnote{In this paper, given a polytope $\mathcal{P} \subseteq \mathbb{R}^D$, we let $V(\mathcal{P})$ be its set of vertices, while we denote by $\text{vol}_{m}(\mathcal{P})$ its Lebesgue measure in $m$ dimensions. For ease of notation, whenever $m = D-1$, we simply write $\text{vol}(\mathcal{P})$. Moreover, we let $\text{int}(\mathcal{P})$ be the interior of $\mathcal{P}$ relative to a subspace that fully contains $\mathcal{P}$ and has minimum dimension. In the case of polytopes $\X^\triangle(a_i)$, the $(d-1)$-dimensional simplex is one of such subspaces.}
Intuitively, the set $\X^{\Hsquare}(a_i)$ encompasses all the slices under which action $a_i$ is a best response for the receiver.
The set $\X^\triangle(a_i)$ has the same interpretation, but for normalized slices.
Specifically, if $x \in \X^{\Hsquare}(a_i)$ is a slice of $\phi$ with respect to $s \in \sset$, then $a_i \in \A^\phi(s)$.
Notice that a slice $x \in \X^{\Hsquare}$ may belong to more than one polytope $\X^{\Hsquare}(a_i)$, when it is the case that $|\A^\phi(s)| > 1$ for signaling schemes $\phi$ having $x$ as slice with respect to $s \in \sset$.\footnote{Let us remark that, in this paper, we assume that there are \emph{no} two equivalent receiver's actions $a_i \neq a_j \in \mathcal{A}$ such that $u_\theta(a_i)=u_\theta(a_j)$ for all $\theta \in \Theta$. Thus, a slice can belong to more than one polytope only if it lies on some separating hyperplane. This assumption is w.l.o.g.~since it is possible to account for equivalent actions by introducing at most $n$ additional separating hyperplanes to consider tie breaking in favor of the sender.}
In order to denote the best-response action actually played by the receiver under a slice $x \in \X^{\Hsquare}$, we introduce the symbol $a(x)$, where $a(x) \coloneqq a^\phi(s)$ for any $\phi$ having $x$ as slice with respect to $s \in \sset$.
%
%Specifically, given any $x \in \X^{\Hsquare}(a_i)$, if a signaling scheme $\phi: \Theta \to \Delta_\sset$ is such that $\phi_\theta(s) = x_\theta$ for all $\theta \in \Theta$ for some signal $s \in \sset$, then the set $A^\phi(s)$ of receiver's best-response actions contains action $a_i$.
%
Figure~\ref{fig:cones} depicts an example of the polytopes $\X^{\Hsquare}(a_i)$ and $\X^\triangle(a_i)$ in order to help the reader to grasp more intuition about them.

\begin{comment}

\footnote{In the following, given a receiver's action $a_i \in A$, we denote with $V(\mathcal{P}(a_i))$ the set of vertices of $\mathcal{P}(a_i)$. We notice that each vertex of $\X(a_i)$ lies at the intersection of $d-1$ hyperplanes selected from the separating hyperplanes $H_{ij}$ and boundary planes, defined as $H_i \coloneqq \{ x \in \mathbb{R}^d \mid x_{\theta_i} = 0 \}$~for~$\theta_i \in \Theta$. We denote by $\text{vol}_{m}(\mathcal{P})$ the Lebesgue measure in $m$ dimensions of a polytope $\mathcal{P} \subseteq \mathbb{R}^D$. For ease of notation, whenever $D = d$ and $m = d-1$, we simply write $\text{vol}(\mathcal{P})$. Moreover, we let $\text{int}(\mathcal{P})$ be the interior of $\mathcal{P}$ relative to a subspace that fully contains $\mathcal{P}$ and has minimum dimension. In the case of a best-response region $\mathcal{P}(a_i)$, the $(d-1)$-dimensional simplex is one of such subspaces.}

\end{comment}

%\alb{Il pezzo qua sotto non va più bene.}\bollo{Ora dovrebbe essere giusto}

A crucial fact exploited by the learning algorithm developed in Section~\ref{sec:algo_noregret} is that knowing the separating hyperplanes $H_{ij}$ defining the polytopes $\X^\triangle(a_i)$ of normalized slices is \emph{sufficient} to determine an optimal signaling scheme.
Indeed, the polytopes $\X^\Hsquare(a_i)$ of unnormalized slices can be easily reconstructed by simply removing the normalization constraint $\sum_{\theta \in \Theta} x_\theta = 1$ from $\X^\triangle(a_i)$.
Furthermore, as we show in Section~\ref{sec:algo_noregret} (see the proof of Lemma~\ref{lem:find_signaling} in particular), there alway exits an optimal signaling scheme using at most one slice $x^a \in \X^\Hsquare(a)$ for each receiver's action $a \in \mathcal{A}$. 
%
% there always exists an optimal signaling scheme $\phi^\star : \Theta \rightarrow \Delta_{\sset}$ that is \emph{direct} and \emph{persuasive}, which means that $\sset = \mathcal{A}$ and $a \in \mathcal{A}^{\phi^\star}(a)$ for every $a \in \mathcal{A}$.
% This implies that $\phi^\star$ employs only one slice $x^a$ for each action $a \in \mathcal{A}$, with $x^a \in \X^\Hsquare(a)$.
%
%In particular, as we show in the following Section~\ref{sec:algo_noregret} (see Lemma~\ref{lem:espoptsol} in particular), an optimal signaling scheme can be computed by suitably combining normalized slices representing vertices of polytopes $\X^\triangle(a_i)$. 
%Intuitively, this follows from the fact that each $\X^{\Hsquare}(a_i)$ is obtained by intersecting a ``cone-shaped'' space with the hypercube $[0,1]^d$, since all the separating hyperplanes pass through the origin in $\mathbb{R}^d$.
%%
%Thus, if one takes any slice $x \in \X^{\Hsquare}(a_i)$ and ``scales'' it by some scalar value $\alpha$, then the resulting slice is still in $\X^{\Hsquare}(a_i)$.
%
% Figure~\ref{fig:cones} depicts an example of the polytopes $\X^{\Hsquare}(a_i)$ and $\X^\triangle(a_i)$ to help the reader to grasp more intuition on them.

We conclude the section with some remarks that help to better clarify why we need to work with signaling scheme slices in order to design our learning algorithm in Section~\ref{sec:algo_noregret}.

\paragraph{Why we need slices for learning}
	The coefficients of separating hyperplanes $H_{ij}$ are products between prior probabilities $\mu_\theta$ and receiver's utility differences $u_\theta (a_i) - u_\theta(a_j)$.
	In order to design a no-regret learning algorithm, it is fundamental that such coefficients are learned exactly, since even an arbitrarily small approximation error may result in ``missing'' some receiver's best responses, and this may potentially lead to a large loss in sender's expected utility (and, in its turn, to large regret).
	As a result, any na\"ive approach that learns the prior and receiver's payoffs separately is deemed to fail, as it would inevitably result in approximate separating hyperplanes being learned.
	Operating in the space of signaling scheme slices crucially allows us to learn separating hyperplanes exactly.
	As we show in Section~\ref{sec:algo_noregret}, it makes it possible to directly learn the coefficients of separating hyperplanes without splitting them into products of prior probabilities and receiver's utility differences. 
	%
	% without splitting them into products of prior probabilities and receiver's utility differences. 

\paragraph{Why we need normalized slices}
	One may wonder why we cannot work with (unnormalized) slices in $\X^{\Hsquare}$, rather than with normalized ones in $\X^{\triangle}$.
	Indeed, as we show in Section~\ref{sec:algo_noregret}, our procedure to learn separating hyperplanes crucially relies on the fact that we can restrict the search space to a suitable subset of the $(d-1)$-dimensional simplex.
	This makes it possible to avoid always employing a number of rounds exponential in $d$, which would lead to non-tight regret guarantees.

\paragraph{Why \emph{not} working with posteriors}
	In Bayesian persuasion, it is oftentimes useful to work in the space of posterior distributions induced by sender's signals (see, \emph{e.g.},~\citep{Castiglioni2020online}).
	These are the beliefs computed by the receiver according to Bayes rule at step~(3) of the interaction.
	Notice that posteriors do \emph{not} only depend on the signaling scheme $\phi$ and the sent signal $s \in \sset$, but also on the prior distribution $\mu$.
	Indeed, the same signaling scheme may induce different posteriors for different prior distributions.
	Thus, since in our setting the sender has no knowledge of $\mu$, we cannot employ posteriors.
	Looking at signaling scheme slices crucially allows us to overcome the lack of knowledge of the prior.
	Indeed, one way of thinking of them is as ``prior-free'' posterior distributions.

%% file: Tikz_pictures/slicing_signaling_schemes_2.tikz
\tikzset{every picture/.style={line width=0.75pt}} %set default line width to 0.75pt        

\begin{tikzpicture}[x=0.75pt,y=0.75pt,yscale=-1,xscale=1]
	%uncomment if require: \path (0,317); %set diagram left start at 0, and has height of 317
	
	%Shape: Axis 2D [id:dp916494551797479] 
	\draw  (92.81,279.18) -- (368.81,279.18)(120.41,34) -- (120.41,306.42) (361.81,274.18) -- (368.81,279.18) -- (361.81,284.18) (115.41,41) -- (120.41,34) -- (125.41,41)  ;
	%Shape: Rectangle [id:dp7450157736062506] 
	\draw   (120.41,70.15) -- (331.04,70.15) -- (331.04,279.18) -- (120.41,279.18) -- cycle ;
	%Shape: Polygon [id:ds6310926340028197] 
	\draw  [fill={rgb, 255:red, 208; green, 2; blue, 27 }  ,fill opacity=0.25 ] (120.41,69.94) -- (252.5,70.59) -- (120.41,278.96) -- cycle ;
	%Shape: Polygon [id:ds1664744776076852] 
	\draw  [fill={rgb, 255:red, 74; green, 144; blue, 226 }  ,fill opacity=0.39 ] (252.5,70.8) -- (331.04,70.15) -- (331.04,139) -- (120.41,279.18) -- cycle ;
	%Shape: Polygon [id:ds18099738012592204] 
	\draw  [fill={rgb, 255:black, 65; green, 117; blue, 5 }  ,fill opacity=0.29 ] (331.04,139) -- (331.04,279.18) -- (120.41,279.18) -- cycle ;
	%Straight Lines [id:da9165062179934965] 
	\draw [color={rgb, 255:black, 65; black, 117; black, 5 }  ,draw opacity=1 ][line width=2]    (245.52,195.78) -- (331.04,279.18) ;
	%Straight Lines [id:da13813351218938197] 
	\draw [color={rgb, 255:black, 0; black, 55; black, 122 }  ,draw opacity=1 ][line width=2]    (200.8,151.4) -- (245.52,195.78) ;
	%Shape: Boxed Line [id:dp7208203916581035] 
	\draw [color={rgb, 255:black, 163; black, 2; black, 24 }  ,draw opacity=1 ][line width=2]    (119.66,69.42) -- (200.8,151.4) ;
	%Shape: Ellipse [id:dp9748504957606339] 
	\draw  [fill={rgb, 255:red, 0; green, 0; blue, 0 }  ,fill opacity=1 ] (117.78,70.15) .. controls (117.78,71.6) and (118.96,72.77) .. (120.41,72.77) .. controls (121.86,72.77) and (123.04,71.6) .. (123.04,70.15) .. controls (123.04,68.71) and (121.86,67.54) .. (120.41,67.54) .. controls (118.96,67.54) and (117.78,68.71) .. (117.78,70.15) -- cycle ;
	%Shape: Ellipse [id:dp47316308877969204] 
	\draw  [fill={rgb, 255:red, 0; green, 0; blue, 0 }  ,fill opacity=1 ] (328.4,279.18) .. controls (328.4,280.62) and (329.58,281.79) .. (331.04,281.79) .. controls (332.49,281.79) and (333.67,280.62) .. (333.67,279.18) .. controls (333.67,277.74) and (332.49,276.57) .. (331.04,276.57) .. controls (329.58,276.57) and (328.4,277.74) .. (328.4,279.18) -- cycle ;
	%Shape: Ellipse [id:dp36932299576925565] 
	\draw  [fill={rgb, 255:red, 0; green, 0; blue, 0 }  ,fill opacity=1 ] (198.91,151.92) .. controls (198.91,153.37) and (200.09,154.54) .. (201.55,154.54) .. controls (203,154.54) and (204.18,153.37) .. (204.18,151.92) .. controls (204.18,150.48) and (203,149.31) .. (201.55,149.31) .. controls (200.09,149.31) and (198.91,150.48) .. (198.91,151.92) -- cycle ;
	%Shape: Ellipse [id:dp05564215223753255] 
	\draw  [fill={rgb, 255:red, 0; green, 0; blue, 0 }  ,fill opacity=1 ] (243.42,196.83) .. controls (243.42,198.27) and (244.6,199.44) .. (246.05,199.44) .. controls (247.51,199.44) and (248.69,198.27) .. (248.69,196.83) .. controls (248.69,195.38) and (247.51,194.21) .. (246.05,194.21) .. controls (244.6,194.21) and (243.42,195.38) .. (243.42,196.83) -- cycle ;
	
	% Text Node
	\draw (128.82,140.65) node [anchor=north west][inner sep=0.75pt]  [font=\large] [align=left] {$\displaystyle \mathcal{X}^{\square }( a_{1})$};
	% Text Node
	\draw (267.09,83.33) node [anchor=north west][inner sep=0.75pt]  [font=\large] [align=left] {$\displaystyle \mathcal{X}^{\square }( a_{2})$};
	% Text Node
	\draw (174.37,252.61) node [anchor=north west][inner sep=0.75pt]  [font=\large] [align=left] {$\displaystyle \mathcal{X}^{\square}( a_{3})$};
	% Text Node
	\draw (91.07,42.36) node [anchor=north west][inner sep=0.75pt]   [align=left] {\Large $\displaystyle x_{\theta _{1}}$};
	% Text Node
	\draw (342.67,284.84) node [anchor=north west][inner sep=0.75pt]   [align=left] {\Large $\displaystyle x_{\theta _{2}}$};
	% Text Node
	\draw (326.08,287.5) node [anchor=north west][inner sep=0.75pt]   [align=left] {\Large $\displaystyle 1$};
	% Text Node
	\draw (101.76,67.5) node [anchor=north west][inner sep=0.75pt]   [align=left] {\Large $\displaystyle 1$};
	% Text Node
	\draw (168.88,82.5) node [anchor=north west][inner sep=0.75pt]  [font=\large,rotate=-45] [align=left] {$\displaystyle \mathcal{X}^{\triangle }( a_{1})$};
	% Text Node
	\draw (225.96,135.88) node [anchor=north west][inner sep=0.75pt]  [font=\large,rotate=-45] [align=left] {$\displaystyle \mathcal{X}^{\triangle }( a_{2})$};
	% Text Node
	\draw (288.9,197.4) node [anchor=north west][inner sep=0.75pt]  [font=\large,rotate=-45] [align=left] {$\displaystyle \mathcal{X}^{\triangle }( a_{3})$};

\end{tikzpicture}

%% file: regret.tex
\section{Learning to persuade without a clue}\label{sec:algo_noregret}

In this section, we design our no-regret algorithm (Algorithm~\ref{alg:main_algorithm}).
%
% In the following, we assume that there are no two equivalent actions $a \neq a' \in \mathcal{A}$ such that $u_\theta(a)=u_\theta(a')$ for every $\theta \in \Theta$.\footnote{This assumption is without loss of generality, as one can account for equivalent actions by introducing at most $n$ additional hyperplanes to consider tie-breaking in favor of the sender.}
%
We adopt a sophisticated \emph{explore-then-commit} approach that exploits the signaling scheme representation based on slices introduced in Section~\ref{sec:slices}.
Specifically, our algorithm works by first exploring the space $\X \coloneqq \X^\triangle$ of normalized slices in order to learn satisfactory ``approximations'' of the polytopes $\X(a_i) \coloneqq \X^\triangle(a_i)$.\footnote{In the rest of the paper, for ease of notation, we write $\X$ and $\X(a_i)$ instead of $\X^\triangle$ and $\X^\triangle(a_i)$.}
Then, it exploits them in order to compute suitable approximately-optimal signaling schemes to be employed in the remaining rounds.
Effectively implementing this approach raises considerable challenges.

The \textbf{first} challenge is that the algorithm cannot directly ``query'' a slice $x \in \X$ to know action $a(x)$, as it can only commit to fully-specified signaling schemes.
Indeed, even if the algorithm commits to a signaling scheme including the selected slice $x \in \X$, the probability that the signal associated with $x$ is sent depends on the (unknown) prior, as it is equal to $\sum_{\theta \in \Theta} \mu_\theta x_\theta$.
This probability can be arbitrarily small.
Thus, in order to observe $a(x)$, the algorithm may need to commit to the signaling scheme for an unreasonably large number of rounds.
To circumvent this issue, we show that it is possible to focus on a subset $\X_\epsilon \subseteq \X$ of normalized slices ``inducible'' with at least a suitably-defined probability $\epsilon \in (0,1)$.
Such a set $\X_\epsilon$ is built by the algorithm in its first phase.

%\alb{Il pezzo qua sotto non va più bene.}\bollo{Ora dovrebbe andare}

%The \textbf{second} challenge that we face is learning the vertices of polytopes $\X_\epsilon(a_i) \coloneqq \X(a_i) \cap \X_\epsilon$ \emph{exactly}.
The \textbf{second} challenge that we face is learning the polytopes $\X_\epsilon(a_i) \coloneqq \X(a_i) \cap \X_\epsilon$.
This is done by means of a technically-involved procedure that learns the separating hyperplanes $H_{ij}$ needed to identify them.
This procedure is an adaptation to Bayesian persuasion settings of an algorithm recently introduced for a similar problem in Stackelberg games~\citep{bacchiocchi2024sample}.

Finally, the \textbf{third} challenge is how to use the polytopes $\X_\epsilon(a_i)$ to compute suitable approximately-optimal signaling schemes to commit to after exploration.
We show that this can be done by solving an LP, which, provided that the set $\X_\epsilon$ is carefully constructed, gives signaling schemes with sender's expected utility sufficiently close to that of an optimal signaling scheme.

%\alb{Volendo in questa pagina c'è un sacco di spazio.}

\begin{wrapfigure}[17]{R}{0.5\textwidth}
	\vspace{-0.65cm}
	\begin{minipage}{0.5\textwidth}
		\begin{algorithm}[H]
			\caption{\texttt{Learn-to-Persuade-w/o-Clue}}\label{alg:main_algorithm}
			\begin{algorithmic}[1]
				\Require $T \in \mathbb{N}$
				\State $\delta \gets \nicefrac 1 T $, $\zeta \gets \nicefrac 1 T$
				\State $\epsilon \gets  \left\lceil \frac{\sqrt{B n} d^4 }{ \sqrt T } \right\rceil$
				\State $T_1\gets \left\lceil \frac{12}{ \epsilon}\log\left(\frac{2d}{\delta} \right) \right\rceil$
				\State $t \gets 1$
				\State $\mathcal{X}_\epsilon  \gets \texttt{Build-Search-Space}(T_1, \epsilon)$
				\State $ \mathcal{R}_\epsilon \gets \texttt{Find-Polytopes}(\mathcal{X}_\epsilon, \zeta)$
				%\State $ \mathcal{Y} \gets \texttt{Find-Slices}(\mathcal{X}_\epsilon, \zeta)$ 
				\While{$t \le T$}
				\State $\phi_t \gets \texttt{Compute-Signaling} (\mathcal{R}_\epsilon, \X_\epsilon, \widehat{\mu}_t ) $ 
				%\State $\phi_t \gets \texttt{Compute-Signaling} (\mathcal{Y}, \X_\epsilon, \widehat{\mu}_t ) $ 
				\State Commit to $\phi_t$, observe $\theta^t$, and send $s^t$
				\State Observe feedback $a^t$ and receive $u^\text{s}_t$
				\State Compute prior estimate $\widehat{\mu}_{t+1}$
				\State $t \gets t + 1$
				\EndWhile
			\end{algorithmic}
		\end{algorithm}
\end{minipage}
\end{wrapfigure}

The pseudocode of our no-regret learning algorithm is provided in Algorithm~\ref{alg:main_algorithm}.
In the pseudocode, we assume that all the sub-procedures have access to the current round counter $t$, all the observed states of nature $\theta^t$, and all the feedbacks $a^t$, $u^\text{s}_t$ received by the sender.\footnote{Notice that, in Algorithm~\ref{alg:main_algorithm}, the sub-procedures \texttt{Build-Search-Space} and \texttt{Find-Polytopes} perform some rounds of interaction, and, thus, they update the current round counter $t$. For ease of presentation, we assume that, whenever $t > T$, their execution is immediately stopped (as well as the execution of Algorithm~\ref{alg:main_algorithm}).}
Moreover, in Algorithm~\ref{alg:main_algorithm} and its sub-procedures, we use $\widehat{\mu}_t \in \Delta_\Theta$ to denote the \emph{prior estimate} at any round $t > 1$, which is a vector with components defined as $\widehat{\mu}_{t,\theta} \coloneqq {N_{t,\theta}}/({t-1})$ for $\theta \in \Theta$, where $N_{t,\theta} \in \mathbb{N}$ denotes the number of times that state of nature $\theta$ is observed up to round $t$ (excluded).
Algorithm~\ref{alg:main_algorithm} can be conceptually divided into three phases.
In \emph{phase 1}, the algorithm employs the $\texttt{Build-Search-Space}$ procedure (Algorithm~\ref{alg:estimate_prior}) to build a suitable subset $\X_\epsilon \subseteq \X$ of ``inducible'' normalized slices.
Then, in \emph{phase 2}, the algorithm employs the $ \texttt{Find-Polytopes}$ procedure (see Algorithm~\ref{alg:find_partition} in Appendix~\ref{appendix:find_partition}) to find a collection of polytopes $\mathcal{R} \coloneqq \{\mathcal{R}_\epsilon(a)\}_{a \in \mathcal{A}}$, where each $\mathcal{R}_\epsilon(a)$ is either $\mathcal{X}_\epsilon(a)$ or a suitable subset of $\mathcal{X}_\epsilon(a)$ that is sufficient for achieving the desired goals (see Section~\ref{sec:find_partition}). 
%
% a subset sufficient to reconstruct an approximately-optimal signaling scheme (see Section~\ref{sec:find_partition}).
%
% to find the set $\mathcal{V}$ of all vertices of the polytopes $\X_\epsilon(a_i)$.
%
Finally, \emph{phase 3} uses the remaining rounds to exploit the knowledge acquired in the preceding two phases.
Specifically, at each $t$, this phase employs the \texttt{Compute-Signaling} procedure (Algorithm~\ref{alg:find_signaling_scheme}) to compute an approximately-optimal signaling scheme, by using $\mathcal{R}_\epsilon$, the set $\X_\epsilon$, and the current prior estimate $\widehat{\mu}_t$.

In the rest of this section, we describe in detail the three phases of Algorithm~\ref{alg:main_algorithm}, bounding the regret attained by each of them.
This allows us to prove the following main result about Algorithm~\ref{alg:main_algorithm}.\footnote{Notice that the regret attained by Algorithm~\ref{alg:find_partition} (stated in Theorem~\ref{thm:no_regret_algo}) depends on $B$, which is the bit-complexity of numbers $\mu_\theta u_\theta(a_i)$, \emph{i.e.}, the number of bits required to represent them. We refer the reader to Appendix~\ref{appendix:find_partition} for more details about how we manage the bit-complexity of numbers in this paper.}
%
% By separately bounding the regret suffered in each phase of Algorithm~\ref{alg:main_algorithm}, it is possible to show that its cumulative regret can be upper bounded as follows.
%
\begin{restatable}{theorem}{NoRegretThm}\label{thm:no_regret_algo}
	The regret attained by Algorithm~\ref{alg:main_algorithm} is $R_T \le \widetilde{\mathcal{O}}\big(   \binom{d+n}{d} n^{\nicefrac{3}{2}} d^3 \sqrt {BT} \big)$.
	%
	% Algorithm~\ref{alg:main_algorithm}  suffers an expected regret bounded as:  $$R_T \le \widetilde{\mathcal{O}}\left(   \binom{d+n}{d} {n} d^3 \sqrt {BT} \right).$$
\end{restatable}
%
% \bac{B ancora non si sa cosa sia qui. Proposta. The parameter $B$ refers to precision to represent the quantities characterizing the instance of the problem, we refer the reread to Appendix~\dots  for a more detailed discussion. }
%
We observe that the regret bound in Theorem~\ref{thm:no_regret_algo} has an exponential dependence on the number of states of nature $d$ and the number of receiver's actions $n$, due to the binomial coefficient.
Indeed, when $d$, respectively $n$, is constant, the regret bound is of the order of $\widetilde{\mathcal{O}}( n^{d}\sqrt{T})$, respectively $\widetilde{\mathcal{O}}( d^{n}\sqrt{T})$.
Such a dependence is tight, as shown by the lower bounds that we provide in Section~\ref{sec:lower_bounds_regret}.

\subsection{Phase 1: \texttt{Build-Search-Space}}

Given an $\epsilon \in (0,{\nicefrac{1}{6d}})$ and a number of rounds $T_1$, the \texttt{Build-Search-Space} procedure (Algorithm~\ref{alg:estimate_prior}) computes a subset $\X_\epsilon \subseteq \X$ of normalized slices satisfying two crucial properties needed by the learning algorithm to attain the desired guarantees.
Specifically, the first property is that any slice $x \in \X_\epsilon$ can be ``induced'' with sufficiently high probability by a signaling scheme, while the second one is that, if $x \notin \X_\epsilon$, then signaling schemes ``induce'' such a slice with sufficiently small probability.
Intuitively, the first property ensures that it is possible to associate any $x \in \X_\epsilon$ with the action $a(x)$ in a number of rounds of the order of $\nicefrac{1}{\epsilon}$, while the second property is needed to bound the loss in sender's expected utility due to only considering signaling schemes with slices in $\X_\epsilon$.
%
%
%The goal of the \texttt{Find-Search-Space} procedure (Algorithm~\ref{alg:estimate_prior}) is to compute a suitable set of two signaling schemes $\mathcal{X}_\epsilon$, in a way that to each signaling scheme $x \in \mathcal{X}_\epsilon$ it is possible to associate its corresponding best response $b_x \in \mathcal{A}$ in at most $\mathcal{O}(\nicefrac{1}{\epsilon})$ rounds, with high probability. To do so, Algorithm~\ref{alg:estimate_prior} at Lines~\ref{line:compute_prior_begin}-\ref{line:compute_prior_end} employs $p\in \mathbb{N}_{+}$ rounds to compute an estimate of the prior distribution.

\begin{wrapfigure}[14]{R}{0.5\textwidth}
	\vspace{-0.85cm}
	\begin{minipage}{0.5\textwidth}
		\begin{algorithm}[H]
			\caption{\texttt{Build-Search-Space}}\label{alg:estimate_prior}
			\begin{algorithmic}[1]
				\Require $\epsilon  \in (0,{\nicefrac{1}{6d}})$, number of rounds $T_1$
				\State $\widetilde{\Theta} \gets \varnothing$
				% $T_1\gets \left\lceil \frac{12}{ \epsilon}\log\left(\frac{2d}{\delta} \right) \right\rceil$
				% and $N_\theta \gets 0 \,\, \forall \theta \in \Theta$
				\While{$t \leq  T_1$} \label{line:compute_prior_begin}
				\State Commit to any $\phi_t$, observe $\theta^t$, and send $s^t$
				\State Observe feedback $a^t$ and receive $u^\text{s}_t$
				\State $t \gets t +1$
				\EndWhile
				\State Compute prior estimate $\widehat{\mu}_{t}$
				\State $\widehat{\mu} \gets \widehat{\mu}_t$ \label{line:compute_prior_end}
				% \gets {N_{\theta}}/ {p} \,\,\ \forall \theta \in \Theta$ 
				\For{$\theta \in \Theta$} \label{line:loop_set_theta}
				\State \textbf{if} $\widehat{\mu}_{\theta} >  2\epsilon$ \textbf{then} $\widetilde{\Theta} \gets \widetilde{\Theta}\cup \{ \theta \}$\label{line:set_Theta} 	
				\EndFor
				\State $\mathcal{X}_\epsilon \gets \left\{ x \in \X \mid \sum_{\theta  \in \widetilde{\Theta}} \widehat{\mu}_{\theta}x_\theta \ge 2\epsilon \right\}$\label{line:set_xi_defintion}
				\State \textbf{return} $\mathcal{X}_\epsilon $						
			\end{algorithmic}
		\end{algorithm}
	\end{minipage}
\end{wrapfigure}

Algorithm~\ref{alg:estimate_prior} works by simply observing realized states of nature $\theta^t$ for $T_1$ rounds, while committing to any signaling scheme meanwhile.
This allows the algorithm to build a sufficiently accurate estimate $\widehat{\mu}$ of the true prior $\mu$.
Then, the algorithm uses such an estimate to build the set $\X_\epsilon$.
Specifically, it constructs $\X_\epsilon$ as the set containing all the normalized slices that are ``inducible'' with probability at least $2 \epsilon$ under the estimated prior $\widehat{\mu}$, after filtering out all the states of nature whose estimated probability is \emph{not} above the $2\epsilon$ threshold (see Lines~\ref{line:loop_set_theta}--\ref{line:set_xi_defintion} in Algorithm~\ref{alg:estimate_prior}).

%We observe that a standard approach to compute an estimate of the different components of the prior distribution would require employing the well-known Hoeffding inequality so that the final estimates are within $\epsilon$ close to the original ones with high probability. However, this would require to employ $\mathcal{O}({1}/{\epsilon^2})$ rounds, thus preventing our algorithm to achieve the desired regret guarantees.
%
% Then, at Lines~\ref{line:loop_set_theta}~-~\ref{line:set_xi_defintion}, Algorithm~\ref{alg:estimate_prior} computes a subset $\Theta' \subseteq \Theta$ of states of nature in way that it is possible to build a set of two signals singling schemes $\mathcal{X}_\epsilon \subseteq \Delta_d$ that satisfies the following lemma.

The following lemma formally establishes the two crucial properties that are guaranteed by Algorithm~\ref{alg:estimate_prior}, as informally described above.
\begin{restatable}{lemma}{PriorEstimate}\label{lem:prior_estimate} 
	Given $T_1 \coloneqq \left\lceil \frac{12}{ \epsilon}\log\left(\nicefrac{2d}{\delta} \right) \right\rceil$ and $\epsilon \in (0,\nicefrac{1}{6d})$, Algorithm~\ref{alg:estimate_prior} employs $T_1$ rounds and terminates with a set $\mathcal{X}_\epsilon \subseteq \X$ such that, with probability at least $1-\delta$: (i) $\sum_{\theta \in \Theta} \mu_\theta x_\theta \ge \epsilon$ for every slice $x \in \mathcal{X}_\epsilon$ and (ii) $\sum_{\theta \in \Theta} \mu_\theta x_\theta \le 10\epsilon$ for every slice $x \in \X \setminus \mathcal{X}_\epsilon$.
	%	
	%	$\delta \in (0,1)$, with probability at least $1-\delta$, Algorithm~\ref{alg:estimate_prior} terminates in $T_1 \coloneqq \left\lceil \frac{12}{ \epsilon}\log\left(\nicefrac{2d}{\delta} \right) \right\rceil$ rounds with a set $\mathcal{X}_\epsilon$ such that the following holds: (i) $\sum_{\theta \in \Theta} \mu_\theta x_\theta \ge \epsilon$ for every slice $x \in \mathcal{X}_\epsilon$ and (ii) $\sum_{\theta \in \Theta} \mu_\theta x_\theta \le 10\epsilon$ for every slice $x \in \X \setminus \mathcal{X}_\epsilon$.
	%
	% With probability at least $1-\delta$ in $p = \left\lceil \nicefrac{12}{ \epsilon}\log\left(\nicefrac{2d}{\delta} \right) \right\rceil$ rounds, Algorithm~\ref{alg:estimate_prior} computes a set $\Theta' \subseteq \Theta$ such that if $\mu_{\theta} \le \epsilon$, then $\theta \not \in \Theta'$, while if $\mu_{\theta} \ge 6\epsilon$ then $\theta \in \Theta'$. Furthermore, for each $\theta \in \Theta'$, it holds ${\mu}_\theta/2  \le \hat{\mu}_\theta \le 3{\mu}_\theta/2$, while every $x \in \mathcal{X}_\epsilon$ satisfies $\sum_{\theta \in \Theta} \mu_\theta x_\theta \ge \epsilon$.
	%
	% With probability at least $1-\delta$ in $p = \left\lceil \nicefrac{12}{ \epsilon}\log\left(\nicefrac{2d}{\delta} \right) \right\rceil$ rounds, Algorithm~\ref{alg:estimate_prior} computes a set $\mathcal{X}_\epsilon$ such that if $x \in \mathcal{X}_\epsilon$ then $\sum_{\theta \in \Theta} \mu_\theta x_\theta \ge \epsilon$, while  if $x \not \in \mathcal{X}_\epsilon$ then $\sum_{\theta \in \Theta} \mu_\theta x_\theta \le 10\epsilon$.
\end{restatable}
To prove Lemma~\ref{lem:prior_estimate}, we employ the multiplicative version of Chernoff bound, so as to show that it is possible to distinguish whether prior probabilities $\mu_{\theta}$ are above or below suitable thresholds in a number of rounds of the order of $\nicefrac{1}{\epsilon}$.
Specifically, we show that, after $T_1$ rounds and with probability at least $1-\delta$, the set $\widetilde{\Theta}$ in the definition of $\X_\epsilon$ does \emph{not} contain states $\theta \in \Theta$ with $\mu_\theta \leq \epsilon$, while it contains all the states with a sufficiently large $\mu_\theta$.
This immediately proves properties (i)~and~(ii) in Lemma~\ref{lem:prior_estimate}.
Notice that using a multiplicative Chernoff bound is a necessary technicality, since standard concentration inequalities would result in a number of needed rounds of the order of $\nicefrac{1}{\epsilon^2}$, leading to a suboptimal regret bound in the number of rounds $T$.
% 
% in a way that it is possible to distinguish in $\mathcal{O}(\nicefrac{1}{\epsilon})$ rounds if the components of the prior distribution are smaller than the given threshold $\epsilon>0$. Such a technical component is necessary to avoid using $\mathcal{O}(\nicefrac{1}{\epsilon^2})$ rounds to estimate the prior by employing standard concentration inequalities. Furthermore, we introduce the definition of the \emph{clean event} of Phase~1, as formally stated below.
%
%	\begin{definition}[Clean event of Phase~1]
%		We let $\mathcal{E}^1$ be the event under which all the elements $\theta \in \Theta'$ are such that $\hat{\mu}_\theta/2  \le \mu_\theta \le 3\hat{\mu}_\theta/2$, while every $x \in \mathcal{X}_\epsilon$ satisfies $\sum_{\theta \in \Theta} \mu_\theta x_\theta \ge \epsilon$. Furthermore, if $\theta \in \Theta' $, then $\mu_\theta \ge \epsilon$, while if  $\theta \not \in \Theta' $, then $\mu_\theta \le 6\epsilon$.
%	\end{definition}

For ease of presentation, we introduce the following \emph{clean event} for phase~1 of Algorithm~\ref{alg:main_algorithm}.
This encompasses all the situations in which Algorithm~\ref{alg:estimate_prior} outputs a set $\X_\epsilon$ with the desired properties.
\begin{definition}[Phase~1 clean event]\label{def:clean_1}
	$\mathcal{E}_1$ is the event in which $\X_\epsilon$ meets properties (i)--(ii) in Lemma~\ref{lem:prior_estimate}.
	%
	% such that  if $x \in \mathcal{X}_\epsilon$ then $\sum_{\theta \in \Theta} \mu_\theta x_\theta \ge \epsilon$, while  if $x \not \in \mathcal{X}_\epsilon$ then $\sum_{\theta \in \Theta} \mu_\theta x_\theta \le 10\epsilon$.
\end{definition}

%% file: find_partition.tex
\subsection{Phase 2: \texttt{Find-Polytopes}}
\label{sec:find_partition}

%\alb{Aggiornare questa parte.}\bollo{Fatto}

Given a set $\X_\epsilon \subseteq \X$ computed by the \texttt{Build-Search-Space} procedure and $\zeta \in (0,1)$ as inputs, the $\texttt{Find-Polytopes}$ procedure (Algorithm~\ref{alg:find_partition} in Appendix~\ref{appendix:find_partition}) computes a collection $\mathcal{R}_\epsilon \coloneqq \{\mathcal{R}_\epsilon(a)\}_{a \in \mathcal{A}}$ of polytopes enjoying suitable properties sufficient to achieve the desired goals.
%, as described in the following. 
%
% correctly terminates in a suitably-bounded number of rounds with probability at least $1-\zeta$ (see Lemma~\ref{lem:final_partition}).

Ideally, we would like $\mathcal{R}_\epsilon(a) = \mathcal{X}_\epsilon(a)$ for every $a \in \mathcal{A}$.
% this procedure to compute the polytopes $\mathcal{X}_\epsilon(a)$ for every action $a \in \mathcal{A}$.
%
However, it is \emph{not} possible to completely identify the polytopes $\mathcal{X}_\epsilon(a)$ with $\text{vol}(\mathcal{X}_\epsilon(a))=0$.
Indeed, if $\text{vol}(\mathcal{X}_\epsilon(a_i))=0$, then $\mathcal{X}_\epsilon(a_i) \subseteq \mathcal{X}_\epsilon(a_j)$ for some other polytope $\mathcal{X}_\epsilon(a_j)$ with positive volume.
%
%, since $\mathcal{X}_\epsilon(a_i) \subseteq \mathcal{X}_\epsilon(a_j)$ for some other polytope $\mathcal{X}_\epsilon(a_j)$ with positive volume.
%
Thus, due to receiver's tie breaking, it could be impossible to identify the whole polytope $\mathcal{X}_\epsilon(a_i)$.
%
%may be the case that the polytope $\mathcal{X}_\epsilon(a_i)$ cannot be identified.
%resulting from the sender's utility function, the algorithm may be unable to identify the whole region $\mathcal{X}_\epsilon(a_j)$.
%
% Hence, Algorithm~\ref{alg:find_partition} can recover only a subset of $\mathcal{X}_\epsilon(a_j)$.
%
As a result, the $\texttt{Find-Polytopes}$ procedure can output polytopes $\mathcal{R}_\epsilon(a) = \mathcal{X}_\epsilon(a)$ only if $\text{vol}(\mathcal{X}_\epsilon(a)) > 0$.
However, we show that, if $\text{vol}(\mathcal{X}_\epsilon(a)) = 0$, it is sufficient to guarantee that the polytope $\mathcal{R}_\epsilon(a)$ contains a suitable subset $\mathcal{V}_\epsilon(a) $ of the vertices of $\mathcal{X}_\epsilon(a)$; specifically, those in which the best response actually played by the receiver is~$a$.
For every $a \in \mathcal{A}$, such a set is formally defined as $\mathcal{V}_\epsilon(a) \coloneqq \{x \in V(\mathcal{X}_\epsilon(a)) \mid a(x)=a\}.$
Thus, we design $\texttt{Find-Polytopes}$ so that it returns a collection $\mathcal{R}_\epsilon \coloneqq \{\mathcal{R}_\epsilon(a)\}_{a \in \mathcal{A}}$ of polytopes such that:
\begin{itemize}[noitemsep,nolistsep]
	\item[(i)] if it holds $\text{vol}(\mathcal{X}_\epsilon(a)) > 0$, then $\mathcal{R}_\epsilon(a) = \mathcal{X}_\epsilon(a)$, while
	\item[(ii)] if $\text{vol}(\mathcal{X}_\epsilon(a)) = 0$, then $\mathcal{R}_\epsilon(a)$ is a (possible improper) face of $\mathcal{X}_\epsilon(a)$ with $\mathcal{V}_\epsilon(a) \subseteq \mathcal{R}_\epsilon(a)$.
\end{itemize}
%
% In particular, whenever $\text{vol}(\mathcal{X}_\epsilon(a)) = 0$, the polytope $\mathcal{R}_\epsilon(a)$ output by $\texttt{Find-Polytopes}$ is a (possible improper) face of $\mathcal{X}_\epsilon(a)$ such that $\mathcal{V}_\epsilon(a) \subseteq \mathcal{R}_\epsilon(a)$.
%
As a result each polytope $\mathcal{R}_\epsilon(a)$ can be either $\mathcal{X}_\epsilon(a)$ or a face of $\mathcal{X}_\epsilon(a)$, or it can be empty, depending on receiver's tie breaking. 
In all these cases, it is always guaranteed that $\mathcal{V}_\epsilon(a) \subseteq \mathcal{R}_\epsilon(a)$.

To achieve its goal, the $\texttt{Find-Polytopes}$ procedure works by searching over the space of normalized slices $\X_\epsilon$, so as to learn \emph{exactly} all the separating hyperplanes $H_{ij}$ characterizing the needed vertices.
% of polytopes $\X_\epsilon(a_i)$.
%
The algorithm does so by using and extending tools that have been developed for a related learning problem in Stackelberg games (see~\citep{bacchiocchi2024sample}).
%
% adapting an algorithm developed by~\citet{bacchiocchi2024sample} for Stackelberg games.
%
Notice that our Bayesian persuasion setting has some distinguishing features that do \emph{not} allow us to use such tools off the shelf.
%
% have some distinguishing features that call for such adaptations, which we discuss in the following.
%
We refer the reader to Appendix~\ref{appendix:find_partition} for a complete description of the $\texttt{Find-Polytopes}$ procedure.
%, as well as for the technical results related to it.

A crucial component of $\texttt{Find-Polytopes}$ is a tool to ``query'' a normalized slice $x \in \X_\epsilon$ in order to obtain the action $a(x)$.
%
% The first difference is due to how a normalized slice $x \in \X_\epsilon$ is ``queried'' in order to obtain action $a(x)$.
%
%$a_i \in A$ such that $x \in \X(a_i) \cap \mathcal{X}_\epsilon$.
%
This is done by using a sub-procedure that we call $\texttt{Action-Oracle}$ (see Algorithm~\ref{alg:action_oracle} in Appendix~\ref{appendix:find_partition}), which works by committing to a signaling scheme including slice $x$ until the signal corresponding to such a slice is actually sent.
Under the clean event $\mathcal{E}_1$, the set $\X_\epsilon$ is built in such a way that $\texttt{Action-Oracle}$ returns the desired action $a(x)$ in a number of rounds of the order of $\nicefrac{1}{\epsilon}$, with high probability.
This is made formal by the following lemma.
\begin{restatable}{lemma}{PriorRounds}\label{lem:rounds_signal} 
	%Given $\rho > 0$ and any two signals signaling scheme $x \in \mathcal{X}_\epsilon$ such that $\sum_{\theta \in \Theta} \mu_\theta x_\theta \ge \epsilon$, if the sender commits to the same signaling scheme $x$ for at least $q \coloneqq \left\lceil \nicefrac{\log(\rho)}{\log(1-\epsilon)} \right \rceil$ rounds, then the probability of observing the signal $s_1$ is at least $1-\rho$.
	%
	Under event $\mathcal{E}^1$, given any $\rho \in (0,1)$ and a normalized slice $x \in \X_\epsilon$, if the sender commits to a signaling scheme $\phi: \Theta \to \sset \coloneqq \{ s_1, s_2 \}$ such that $\phi_\theta(s_1) = x_\theta$ for all $\theta \in \Theta$ during $q \coloneqq \left\lceil \frac{1}{\epsilon} \log(\nicefrac 1 \rho) \right \rceil$ rounds, then, with probability at least $1-\rho$, signal $s_1$ is sent at least once.
	%
	% Under the event $\mathcal{E}^1$, given $\rho > 0$ and any two signals signaling scheme $x \in \mathcal{X}_\epsilon$, if the sender commits to the same signaling scheme $x$ for at least $q \coloneqq \left\lceil \nicefrac{\log(1/ \rho)}{\epsilon} \right \rceil$ rounds, then the probability of observing the signal $s_1$ is at least $1-\rho$.
\end{restatable}
Notice that the signaling scheme used to ``query'' an $x \in \X_\epsilon$ only employs \emph{two} signals: one is associated with slice $x$, while the other crafted so as to make the signaling scheme well defined.

The following lemma provides the guarantees of Algorithm~\ref{alg:find_partition} in Appendix~\ref{appendix:find_partition}.
%\footnote{Notice that the number of rounds required by Algorithm~\ref{alg:find_partition} (stated in Lemma~\ref{lem:final_partition}) depends on the bit-complexity of some numbers, or, in other words, the number of bits required to represent such numbers. We refer the reader to Appendix~\ref{appendix:find_partition} for more details about how we manage the bit-complexity of numbers in this paper.}
%
\begin{restatable}{lemma}{FindPartition}\label{lem:final_partition}
	%	Under the event $\mathcal{E}^1$, with probability at least $1-\zeta$, Algorithm~\ref{alg:find_partition} computes every polytope $\mathcal{P}'(a) \cap \mathcal{X}_\epsilon$ with volume larger than zero in at most $\widetilde{\mathcal{O}}\left( \frac{n^2}{\epsilon} \left(d^7L\log\left(\frac{1}{\zeta}\right)+\binom{d+n}{d}\right) \right)$ rounds, where $L=B_\mu +B_u +B_\epsilon +B_{\hat{\mu}}$.
	% If $\sum_{\theta  \in \Theta} \mu_\theta x_\theta \ge \epsilon$ for every $x \in \mathcal{X}_\epsilon$, then with probability at least $1-\zeta$, Algorithm~\ref{alg:find_partition} computes the set of vertices of every polytope $\mathcal{P}'(a) \cap \mathcal{X}_\epsilon$ with volume larger than zero in at most $$\widetilde{\mathcal{O}}\left( \frac{n^2}{\epsilon} \left(d^7L\log\left(\frac{1}{\zeta}\right)+\binom{d+n}{d}\right) \right)$$ rounds, where $L=B +B_\epsilon +B_{\hat{\mu}}$.
	%Under the event $\mathcal{E}^1$, with probability at least $1-\zeta$, Algorithm~\ref{alg:find_partition} computes the set of vertices $\mathcal{V} \coloneqq \hspace{-5mm} \bigcup \limits_{a: \textnormal{vol}(\mathcal{P}'(a))>0} \hspace{-5mm} V(\mathcal{X} \cap \mathcal{P}'(a) )$ in at most: $$\widetilde{\mathcal{O}}\left( \frac{n^2}{\epsilon} \left(d^7L\log\left(\frac{1}{\zeta}\right)+\binom{d+n}{d}\right) \right)$$ rounds, where $L=B +B_\epsilon +B_{\hat{\mu}}$ and $B$ the bit-complexity \bac{aggiornare}.
	%
	Given inputs $\X_\epsilon \subseteq \X$ and $\zeta \in (0,1)$ for Algorithm~\ref{alg:find_partition},	let $L \coloneqq B +B_\epsilon +B_{\hat{\mu}}$, where $B$, $B_\epsilon$, and $B_{\widehat{\mu}}$ denote the bit-complexity of numbers $\mu_\theta u_\theta(a_i)$, $\epsilon$, and $\widehat{\mu}$, respectively.
	Then, under event $\mathcal{E}_1$ and with at probability at least $1-\zeta$, Algorithm~\ref{alg:find_partition} outputs a collection $\mathcal{R}_\epsilon \coloneqq \{\mathcal{R}_\epsilon(a)\}_{a \in \mathcal{A}}$, where $\mathcal{R}_\epsilon(a)$ is a (possibly improper) face of $\mathcal{X}_\epsilon(a)$ such that $\mathcal{V}_\epsilon(a) \subseteq \mathcal{X}_\epsilon(a)$, in a number of rounds $T_2$:
	%
	% Algorithm~\ref{alg:find_partition} computes a set of vertices $\mathcal{V} =  \bigcup_{a_i \in \A: \textnormal{vol}(\X_\epsilon(a_i))>0}  V(\X_\epsilon(a_i) )$ with probability at least $1-\zeta$, in a number of rounds $T_2$ such that:
	% \bollo{[...] computes the set of vertices $\mathcal{V}=\dots$ and the corresponding best responses $(a(x))_{x \in \mathcal{V}}$ [...]}
	%$$
	\begin{equation*}
		T_2 \leq \widetilde{\mathcal{O}}\left( \frac{n^2}{\epsilon} \log^2\left(\frac{1}{\zeta}\right) \left(d^7L+\binom{d+n}{d}\right) \right).
	\end{equation*}
	%$$
	%
\end{restatable}
%\bollo{Perchè qui e in Theorem~\ref{th:main_th_sc} usiamo $\log^2\left(\nicefrac{1}{\zeta}\right)$ e non $\log^2\left(\frac{1}{\zeta}\right)$?}

% \bollo{Non so se si capisca la definizione di B, però scriverla meglio rende il lemma inutilmente lungo. Magari si può metterla in fondo al paragrafo prima} % definizione di B commentata sopra

%We introduce the definition of the \emph{clean event} of Phase 2, as formally stated below.
%
For ease of presentation, we introduce the \emph{clean event} for phase 2 of Algorithm~\ref{alg:main_algorithm}, defined as follows:
\begin{definition}[Phase~2 clean event]
	$\mathcal{E}_2$ is the event in which $\mathcal{V}_\epsilon(a) \subseteq \mathcal{R}_\epsilon(a)$ for every $a \in \mathcal{A}$.
	%
	%We let $\mathcal{E}^2$ be the clean event such that Algorithm~\ref{alg:find_partition} \bac{aggiornare}.
	%
	% We let $\mathcal{E}^2$ be the clean event such that Algorithm~\ref{alg:find_partition} correctly computes the set of vertices $\mathcal{V} \coloneqq \hspace{-5mm} \bigcup \limits_{a: \textnormal{vol}(\mathcal{P}'(a))>0} \hspace{-5mm} V(\mathcal{X}_\epsilon \cap \mathcal{P}'(a) )$.
\end{definition}

%% file: find_ss.tex
\subsection{Phase 3: \texttt{Compute-Signaling}}\label{sec:find_signaling}

Given the collection of polytopes $\mathcal{R}_\epsilon \coloneqq \{\mathcal{R}_\epsilon(a)\}_{a \in \mathcal{A}}$ returned by the $\texttt{Find-Polytopes}$ procedure and an estimated prior $\widehat{\mu}_t \in \Delta_{\Theta}$, the \texttt{Compute-Signaling} procedure (Algorithm~\ref{alg:find_signaling_scheme}) outputs an approximately-optimal signaling scheme by solving an LP (Program~\eqref{eq:lp_h_repr} in Algorithm~\ref{alg:find_signaling_scheme}).
%, as described in the following.

\begin{wrapfigure}[18]{R}{0.47\textwidth}
	\vspace{-.7cm}
	\begin{minipage}{0.47\textwidth}
		\begin{algorithm}[H]
			\caption{\texttt{Compute-Signaling}}\label{alg:find_signaling_scheme}
			\begin{algorithmic}[1]
				\Require $\mathcal{R}_\epsilon \coloneqq \{ \mathcal{R}_\epsilon(a) \}_{a \in \mathcal{A}}$, $\X_\epsilon$, $\widehat{\mu}_t \in \Delta_{\Theta}$
				\State Solve Program~\eqref{eq:lp_h_repr} for $x^\star \coloneqq (x^{\star,a})_{a \in \mathcal{A}}$:
				%\vspace{-4mm}
				\begin{align}
					\max_{(x^a)_{a \in \mathcal{A}}} & \,\, \sum_{a \in \mathcal{A}}  \sum_{\theta \in \Theta}  \widehat{\mu}_{t,\theta} x^a_\theta u^\text{s}_\theta(a) \quad \text{s.t.} \label{eq:lp_h_repr}\\
					& x^a \in \mathcal{R}^{\Hsquare}_\epsilon(a) \quad \forall a \in \mathcal{A} \nonumber \\
					& \sum_{a \in \mathcal{A}} x^a_\theta \le 1 \quad \forall \theta \in \Theta \nonumber
				\end{align}
				%\vspace{-2mm}
				%\end{subequations}
				\State $\sset \gets \{ s^\star \} \cup \left\{ s^a \mid a \in \mathcal{A} \right\}$
				\For{$\theta \in \Theta$}
				\State $
				\phi_\theta  \gets \begin{cases}
				\phi_{\theta}(s^{a}) = x^{\star,a}_\theta \quad \forall a \in \mathcal{A} \\
				\phi_{\theta}(s^{\star}) = 1 - \sum_{a \in \mathcal{A} } x^{\star,a}_\theta
				\end{cases}
				$
				\EndFor
				\State\Return $\phi$
				%				
				%				
				%				\Require $\mathcal{Y}\coloneqq (\mathcal{V},[a(x)]_{x \in \mathcal{V}})$, $\X_\epsilon$, $\widehat{\mu}_t \in \Delta_{\Theta}$
				%				\State Solve Program~\eqref{eq:lp_two_signals} for an optimal solution $\alpha^\star$:
				%				%\vspace{-4mm}
				%				\begin{align}
				%					\max_{\alpha \ge \mathbf{0}} & \,\,  \sum_{x \in \mathcal{V} } \alpha_x \sum_{\theta \in \Theta }  \widehat{\mu}_{t,\theta}    x_\theta u^\text{s}_\theta(a(x)) \quad \text{s.t.} \label{eq:lp_two_signals}\\
				%					& \sum_{x \in \mathcal{V}} \alpha_x x_\theta \le 1 \quad \forall \theta \in \Theta \nonumber
				%				\end{align}
				%				%\vspace{-2mm}
				%				%\end{subequations}
				%				\State $\sset \gets \{ s^\star \} \cup \left\{ s^x \mid x \in \mathcal{V} \wedge \alpha^\star_x > 0 \right\}$
				%				\For{$\theta \in \Theta$}
				%				\State $
				%				\phi_\theta \hspace{-.5mm} \gets \hspace{-.5mm} \begin{cases}
				%				\phi_{\theta}(s^{x}) = \alpha_x x_{\theta} \,\,\,\, \forall x \in \mathcal{V} \hspace{-.5mm}: \hspace{-.5mm} \alpha^\star_x > 0\\
				%				\phi_{\theta}(s^{\star}) = 1 - \sum_{x \in \mathcal{V} }\alpha_x x_{\theta}
				%				\end{cases}
				%				$
				%				\EndFor
				%				\State\Return $\phi$
			\end{algorithmic}
		\end{algorithm}
	\end{minipage}
\end{wrapfigure}

Program~\eqref{eq:lp_h_repr} maximizes an approximated version of sender's expected utility over a suitable space of (partially-specified) signaling schemes.
These are defined by tuples of slices $( x^a )_{a \in \mathcal{A}}$ containing an (unnormalized) slice $x^a \in \mathcal{R}^\Hsquare_\epsilon(a) $ for every receiver's action $a \in \mathcal{A}$.
The objective function being maximized by Program~\eqref{eq:lp_h_repr} accounts for the sender's approximate utility under each of the slices $x^a$, where the approximation comes from the estimated prior $\widehat{\mu}_t$.
The intuitive idea exploited by the LP formulation is that, under slice $x^a$, the receiver always plays the same action $a$ as best response, since $a(x^a) = a$ holds by the way in which $\mathcal{R}_\epsilon(a)$ is constructed by $\texttt{Find-Polytopes}$.
In particular, each polytope $\mathcal{R}^\Hsquare_\epsilon(a)$ is built so as to include all the unnormalized slices corresponding to the normalized slices in the set $\mathcal{R}_\epsilon(a)$.
Formally, for every receiver's action $a \in \mathcal{A}$, it holds:
\begin{equation*}
	\mathcal{R}^\Hsquare_\epsilon(a) \coloneqq \left\{ x \in \mathcal{X}^\Hsquare \mid x = \alpha x' \wedge x' \in \mathcal{R}_\epsilon(a) \wedge \alpha \in [0,1] \right\} \cup \{\mathbf{0}\} ,
\end{equation*}
where $\mathbf{0}$ denotes the vector of all zeros in $\mathbb{R}^d$.
We observe that, since the polytopes $\mathcal{R}_\epsilon(a)$ are constructed as the intersection of some halfspaces $\mathcal{H}_{ij}$ and $\mathcal{X}_\epsilon$, it is possible to easily build polytopes $\mathcal{R}^\Hsquare_\epsilon(a_i)$ by simply removing the normalization constraint $\sum_{\theta \in \Theta} x_\theta = 1$.
% (present in the definition of $\X_\epsilon$).
%intersecting the same halfspaces with $\mathcal{X}^\Hsquare_\epsilon$.
%
Notice that, if $\mathcal{R}_\epsilon(a) = \varnothing$, then $\mathcal{R}^\Hsquare_\epsilon(a) = \{\mathbf{0}\}$, which implies that action $a$ is never induced as a best response, since $x^a = \mathbf{0}$.

After solving Program~\eqref{eq:lp_h_repr} for an optimal solution $x^\star \coloneqq (x^{\star,a})_{a \in \mathcal{A}}$, Algorithm~\ref{alg:find_signaling_scheme} employs such a solution to build a signaling scheme $\phi$.
This employs a signal $s^a$ for every action $a \in A$, plus an additional signal~$s^\star$, namely $\mathcal{S} \coloneqq \{s^\star\} \cup \{s^a \mid a \in \mathcal{A}\}$.
Specifically, the slice of $\phi$ with respect to $s^a$ is set to be equal to $x^a$, while its slice with respect to $s^\star$ is set so as to render $\phi$ a valid signaling scheme (\emph{i.e.}, probabilities over signal sum to one for every $\theta \in \Theta$).
Notice that this is always possible thanks to the additional constraints $\sum_{a \in \mathcal{A}} x^a_\theta \le 1$ in Program~\eqref{eq:lp_h_repr}.
Moreover, such a slice may belong to $\mathcal{X}^\Hsquare \setminus \mathcal{X}^\Hsquare_\epsilon$.
Indeed, in instances where there are some $\mathcal{X}^\Hsquare(a)$ falling completely outside $\mathcal{X}^\Hsquare_\epsilon$, this is fundamental to build a valid signaling scheme.
Intuitively, one may think of $s^\star$ as incorporating all the ``missing'' signals in $\phi$, namely those corresponding to actions $a \in \mathcal{A}$ with $\mathcal{X}^\Hsquare(a)$ outside $\mathcal{X}^\Hsquare_\epsilon$.

The following lemma formally states the theoretical guarantees provided by Algorithm~\ref{alg:find_signaling_scheme}.
\begin{restatable}{lemma}{FindSignaling}
	\label{lem:find_signaling}
	Given inputs $\mathcal{R}_\epsilon \coloneqq \{\mathcal{R}_\epsilon(a)\}_{a \in \mathcal{A}}$, $\mathcal{X}_\epsilon$, and $\widehat{\mu}_t \in \Delta_{\Theta}$ for Algorithm~\ref{alg:find_signaling_scheme}, under events $\mathcal{E}_1$ and $\mathcal{E}_2$, the signaling scheme $\phi$ output by the algorithm is $\mathcal{O}(  \epsilon n d  + \nu)$-optimal for $\nu \leq \left| \sum_{\theta \in \Theta} \widehat{\mu}_{t,\theta} - \mu_\theta \right|$.
	%
	%
	%
	% Let Algorithm~\ref{alg:find_signaling_scheme} have in input the collection of polytopes $\mathcal{R}_\epsilon = \{\mathcal{R}_\epsilon(a)\}_{a \in \mathcal{A}}$ computed by Algorithm~\ref{alg:find_partition}, the search space $\mathcal{X}_\epsilon$ computed by Algorithm~\ref{alg:estimate_prior}, and the current prior estimator $\widehat{\mu}_t$.
	% Then, under the events $\mathcal{E}_1$ and $\mathcal{E}_2$, the signaling scheme $\phi$ output by the algorithm is $\mathcal{O}(  \epsilon n d  + \nu)$-optimal for $\nu \leq \left| \sum_{\theta \in \Theta} \widehat{\mu}_{t,\theta} - \mu_\theta \right|$.
	%
	%
	%
	%	Let Algorithm~\ref{alg:find_signaling_scheme} have in input, for every action $a \in \mathcal{A}$, a face $\mathcal{R}(a)$ of $\mathcal{X}_\epsilon(a)$ (possibly the improper face $\mathcal{X}(a)$ itself) such that every vertex $x \in V(\mathcal{X}_\epsilon(a))$ such that $a(x)=a$ belongs to $\mathcal{R}_\epsilon(a)$.
	%	Then, under the event $\mathcal{E}_1$ the signaling scheme $\phi$ output by the algorithm is $\mathcal{O}(  \epsilon n d  + \nu)$-optimal for $\nu \leq \left| \sum_{\theta \in \Theta} \widehat{\mu}_{t,\theta} - \mu_\theta \right|$.
\end{restatable}
In order to provide some intuition on how Lemma~\ref{lem:find_signaling} is proved, let us assume that each polytope $\mathcal{X}_\epsilon(a)$ is either empty or has volume larger than zero, implying that $\mathcal{R}_\epsilon(a) = \mathcal{X}_\epsilon(a)$.
In Appendix~\ref{appendix:find_signaling}, we provide the complete formal proof of Lemma~\ref{lem:find_signaling}, working even with zero-measure non-empty polytopes.
The first observation the we need is that sender's expected utility under a signaling scheme $\phi$ can be decomposed across its slices, with each slice $x$ providing a utility of $\sum_{\theta \in \Theta} \mu_\theta x_\theta u^\text{s}_\theta(a(x))$.
The second crucial observation is that there always exists an optimal signaling scheme $\phi^\star$ that is \emph{direct} and \emph{persuasive}, which means that $\phi^\star$ employs only one slice $x^a$ for each action $a \in \mathcal{A}$, with $a$ being a best response for the receiver under $x^a$.
It is possible to show that the slices $x^a$ that also belong to $\X_\epsilon$ can be used to construct a feasible solution to Program~\ref{eq:lp_h_repr}, since $x^a \in \mathcal{R}^\Hsquare_\epsilon(a)$ by definition.
Thus, restricted to those slices, the signaling scheme $\phi$ computed by Algorithm~\ref{alg:find_signaling_scheme} achieves an approximate sender's expected utility that is greater than or equal to the one achieved by $\phi^\star$.
Moreover, the loss due to dropping the slices that are \emph{not} in $\X_\epsilon$ can be bounded thanks to point (ii) in Lemma~\ref{lem:prior_estimate}.
Finally, it remains to account for the approximation due to using $\widehat{\mu}_t$ instead of the true prior in the objective of Program~\ref{eq:lp_h_repr}.
All the observations above allow to bound sender's expected utility loss as in Lemma~\ref{lem:find_signaling}.

%% file: lower_bounds_regret.tex
\section{Lower bounds for online Bayesian persuasion}\label{sec:lower_bounds_regret}

In this section, we present two lower bounds on the regret attainable in the setting faced by Algorithm~\ref{alg:main_algorithm}.
%
% showing that the regret guarantees attained by Algorithm~\ref{alg:main_algorithm} (see Theorem~\ref{thm:no_regret_algo}) are tight.
%
The first lower bound shows that an exponential dependence in the number of states of nature $d$ and the number of receiver's actions $n$ is unavoidable.
This shows that one cannot get rid of the binomial coefficient in the regret bound of Algorithm~\ref{alg:main_algorithm} provided in Theorem~\ref{thm:no_regret_algo}.
Formally:
%
% In this section, we present two negative results for the problem of minimizing the cumulative regret suffered by any algorithm in our setting. The first negative result exhibits an exponential dependence on the regret suffered by any algorithm with respect to the size of the problem instance, while the second negative result highlights the dependence with respect to the time horizon $T$ of the regret suffered by any algorithm.
%
\begin{restatable}{theorem}{hardnessfirst}\label{thm:hardness1}
	%There exist two absolute constants $\kappa,\eta >0$ such that no algorithm is guaranteed to return an $\kappa$-optimal signaling scheme with probability of at least $1 - \eta$, employing less than $2^{\Omega(n)}$ rounds in the worst-case scenario, even when the prior distribution is known to the sender.
	%
	For any sender's algorithm, there exists a Bayesian persuasion instance in which $n= d+2$ and the regret $R_T$ suffered by the algorithm is at least $2^{\Omega(d)}$, or, equivalently, $2^{\Omega(n)}$.
	%
	% For any algorithm, there exists an instance in which the number of receiver's actions is equal to $n=d+2$ and the regret suffered by this algorithm is $2^{\Omega(d)}$, or equivalently, $2^{\Omega(n)}$.
\end{restatable}
Theorem~\ref{thm:hardness1} is proved by constructing a collection of Bayesian persuasion instances in which an optimal signaling scheme has to induce the receiver to take an action that is a best response only for a unique posterior belief (among those computable by the receiver at step~(3) of the interaction).
This posterior belief belongs to a set of possible candidates having size exponential in the number of states of nature $d$.
As a result, in order to learn such a posterior belief, any algorithm has to commit to a number of signaling schemes that is exponential in $d$ (and, given how the instances are built, in $n$).
%
%the optimal signaling scheme induce a receiver's action that is inducible only in a single posterior that belongs to a set of cardinality exponential in the number of sates of nature. As a result, any algorithm is required to commit to a number of signaling schemes which is exponential in the number of states to determine the posterior that induce such action.

% Furthermore, the following result holds.
%
The second lower bound shows that the regret bound attained by Algorithm~\ref{alg:main_algorithm} is tight in $T$.
\begin{restatable}{theorem}{hardnessthird}\label{thm:hardness3}
	For any sender's algorithm, there exists a Bayesian persuasion instance in which the regret $R_T$ suffered by the algorithm is at least $\Omega(\sqrt T)$.
	%
	% For any algorithm, there exists an instance such that the regret suffered by this algorithm is $\Omega(\sqrt T)$. 
\end{restatable}
To prove Theorem~\ref{thm:hardness3}, we construct two Bayesian persuasion instances with $\Theta = \{ \theta_1, \theta_2 \}$ such that, in the first instance, $\mu_{\theta_1}$ is slightly greater than $\mu_{\theta_2}$, while the opposite holds in the second instance.
%
% $d=2$ such that, in the first instance, the probability of the first state of nature is slightly greater than the probability of the second state, while the opposite holds in the second instance.
%
Furthermore, the two instances are built so that the sender does \emph{not} gain any information that helps to distinguish between them by committing to signaling schemes.
As a consequence, to make a distinction, the sender can only leverage the information gained by observing the states of nature realized at each round, and this clearly results in the regret being at least $\Omega(\sqrt T)$.

% Theorem~\ref{thm:hardness1} and Theorem~\ref{thm:hardness3} confirm the tightness of regret bound achieved by Algorithm~\ref{alg:main_algorithm} with respect to either the time horizon or the size of the instance. 

%\alb{Add conjecture.} \bac{Finally, we notice that Theorem~\ref{thm:hardness1} and Theorem~\ref{thm:hardness3} still leave open the question whether it possible to design a lower bound that jointly consider the dependence with respect to the time horizon and with respect to the size of the problem or better regret guarantees can be achieved. }

%% file: sample_complexity.tex
\section{The sample complexity of Bayesian persuasion}\label{sec:sample_complexity}

In this section, we show how the no-regret learning algorithm developed in Section~\ref{sec:algo_noregret} can be easily adapted to solve a related \emph{Bayesian persuasion PAC-learning problem}.
Specifically, given an (additive) approximation error $\gamma \in (0,1)$ and a probability $\eta \in (0,1)$, the goal of such a problem is to learn a $\gamma$-optimal signaling scheme with probability at least $1-\eta$, by using the minimum possible number of rounds.
This can be also referred to as the \emph{sample complexity} of learning signaling schemes.
%
%a sender's signaling scheme $\phi: \Theta \to \sset$ such that
%$
%	\mathbb{P} \left(  u^\text{s}(\phi) \geq \text{OPT} - \gamma \right) \geq 1 - \eta
%$,
%by using the minimum possible number of rounds, also referred to as the \emph{sample complexity} of Bayesian persuasion.
%%
%\alb{Dire qualcosa sulla misura.}

As in the regret-minimization problem addressed in Section~\ref{sec:algo_noregret}, we assume that the sender does \emph{not} know anything about both the prior distribution $\mu$ and receiver's utility function $u$.

%\bac{QUESTA ROBA DELL'EPSILON E DA DEFINIRE IN APPENDICE!!!!!!!}
\begin{wrapfigure}[12]{R}{0.45\textwidth}
	\vspace{-0.75cm}
	\begin{minipage}{0.45\textwidth}
		\begin{algorithm}[H]
			\caption{\texttt{PAC-Persuasion-w/o-Clue}}\label{alg:main_algorithm_sc}
			\begin{algorithmic}[1]
				\Require $\gamma \in (0,1)$, $\eta \in (0,1)$
				%
				%\bollo{Soluzione $\gamma$-ottima con probabilità almeno $1-\eta$}
				%		\State $\gamma' \gets \texttt{Stern-Brocot-Tree}(\gamma, B_\mu +B_u)$
				%
				% \State $\epsilon \gets \nicefrac{\gamma}{12nd}$ 
				%
%				\State $\epsilon \gets $ see Appendix~\ref{appendix:sample_complexity}\label{line:epsilon_def_sc}
				\State $\delta \gets \nicefrac{\eta}{2}, \zeta \gets \nicefrac{\eta}{2}, \epsilon_1 \gets \nicefrac{\gamma}{12nd}$, $t \gets 1$
				\State $\epsilon \gets \texttt{Compute-Epsilon}(\epsilon_1)$ \label{line:compute_epsilon_sc}
				\State $T_1 \gets \left\lceil \frac{1}{2\epsilon^2} \log\left(\frac{2d}{\delta} \right) \right\rceil$
				\State $\mathcal{X}_\epsilon \gets \texttt{Build-Search-Space}(T_1, \epsilon)$ %\Comment{\textcolor{gray}{Phase 1}}
				%		\State Set $d' \gets |\Theta'|$
				\State $ \mathcal{R}_\epsilon \gets \texttt{Find-Polytopes}(\mathcal{X}_\epsilon, \zeta)$ \label{line:find_slices_sc}
				\State $\phi \gets \texttt{Compute-Signaling} (\mathcal{R}_\epsilon,\mathcal{X}_\epsilon, \widehat{\mu}) $  		%\Comment{\textcolor{gray}{Phase 3}}
				\State \textbf{return} $\phi $
			\end{algorithmic}
		\end{algorithm}
	\end{minipage}
\end{wrapfigure}

\begin{comment}

In this section we tackle a variation of the online learning problem in Bayesian persuasion section.
We consider again a repeated interaction between the sender and the receiver, but without a given time horizon.
Instead, we evaluate an algorithm in terms of the number of rounds that it requires to compute an $\gamma$-optimal solution for a given $\gamma > 0$.
Such a solution is a signaling schemes that provides the sender an expected utility of at least $\textnormal{OPT}-\gamma$, where $\textnormal{OPT}$ is the expected utility of an optimal signaling scheme.

\end{comment}

We tackle the Bayesian persuasion PAC-learning problem with a suitable adaptation of Algorithm~\ref{alg:main_algorithm}, provided in Algorithm~\ref{alg:main_algorithm_sc}.
The first two phases of the algorithm follow the line of Algorithm~\ref{alg:main_algorithm}, with the \texttt{Build-Search-Space} and \texttt{Find-Polytopes} procedures being called for suitably-defined parameters $\epsilon$, $\delta$, $T_1$, and $\zeta$ (taking different values with respect to their counterparts in Algorithm~\ref{alg:main_algorithm}).
In particular, the value of $\epsilon$ depends on $\gamma$ and is carefully computed so as to control the bit-complexity of numbers used in the \texttt{Find-Polytopes} procedure (see Lemma~\ref{lem:final_partition}), as detailed in Appendix~\ref{appendix:sample_complexity}.
Finally, in its third phase, the algorithm calls \texttt{Compute-Signaling} to compute a signaling scheme $\phi$ that can be proved to $\gamma$-optimal with probability at least $1-\eta$.
%
%such that $u^\text{s}(\phi) \geq \text{OPT} - \gamma$ with probability at least $1-\eta$.

The most relevant difference between Algorithm~\ref{alg:main_algorithm_sc} and Algorithm~\ref{alg:main_algorithm} is the number of rounds used to build the prior estimate defining $\X_\epsilon$.
Specifically, while the latter has to employ $T_1$ of the order of $\nicefrac{1}{\epsilon}$ and rely on a multiplicative Chernoff bound to get tight regret guarantees, the former has to use $T_1$ of the order of $\nicefrac{1}{\epsilon^2}$ and standard concentration inequalities to get an $\mathcal{O}(\epsilon)$-optimal solution.
%, thus obtaining a better estimate of prior probabilities.
%
%The most important difference with respect to the regret problem lies in the estimation of the prior $\mu$.
%
%While Algorithm~\ref{alg:main_algorithm} uses $\mathcal{O}(\nicefrac{1}{\epsilon})$ rounds, in this setting we use standard concentration inequalities and thus $\mathcal{O}(\nicefrac{1}{\epsilon^2})$ rounds to obtain a better estimation of this probability distribution.
%
Formally:
\begin{restatable}{lemma}{PriorEstimateSC}\label{lem:prior_estimate_sc} 
	%	With probability at least $1-\delta$ and in $p = \left\lceil \nicefrac{1}{2\epsilon^2} \log\left(\nicefrac{2d}{\delta} \right) \right\rceil$ rounds, Algorithm~\ref{alg:estimate_prior} computes a set $\Theta' \subseteq \Theta$ such that if $\mu_{\theta} \le \epsilon$, then $\theta \not \in \Theta'$, and if $\mu_\theta \ge 3\epsilon$, then $\theta \in \Theta'$. 
	%	Furthermore, for each $\theta \in \Theta$, it holds that $| \hat{\mu}_\theta - \mu_\theta |  \le \epsilon$, while every $x \in \mathcal{X}_\epsilon$ satisfies $\sum_{\theta  \in \Theta} \mu_\theta x_\theta$. 
	%
	%
	Given $T_1 \coloneqq \left\lceil \frac{1}{2\epsilon^2} \log\left(\nicefrac{2d}{\delta} \right) \right\rceil$ and $\epsilon \in (0,1)$, Algorithm~\ref{alg:estimate_prior} employs $T_1$ rounds and outputs $\X_\epsilon \subseteq \X$ such that, with probability at least $1-\delta$: (i) $\sum_{\theta \in \Theta} \mu_\theta x_\theta \ge \epsilon$ for every slice $x \in \X_\epsilon$, (ii) $\sum_{\theta \in \Theta} \mu_\theta x_\theta \le 6\epsilon$ for every slice $x \in \X \setminus \X_\epsilon$, and (iii) $|\widehat{\mu}_{\theta} - \mu_\theta| \le \epsilon$ for every $\theta \in \Theta$.
	%	
	%	with probability at least $1-\delta$, Algorithm~\ref{alg:estimate_prior} terminates in $T_1$ rounds with a set $\X_\epsilon$ such that the following holds: (i) $\sum_{\theta \in \Theta} \mu_\theta x_\theta \ge \epsilon$ for every slice $x \in \X_\epsilon$, (ii) $\sum_{\theta \in \Theta} \mu_\theta x_\theta \le 6\epsilon$ for every slice $x \in \X \setminus \X_\epsilon$, and (iii) $|\widehat{\mu}_{t,\theta} - \mu_\theta| \le \epsilon$ for every $t>T_1$ ad $\theta \in \Theta$.
	%
	%
	% With probability at least $1-\delta$ in $T_1 = \left\lceil \nicefrac{1}{2\epsilon^2} \log\left(\nicefrac{2d}{\delta} \right) \right\rceil$ rounds, Algorithm~\ref{alg:estimate_prior} computes a set $\mathcal{X}_\epsilon$ such that if $x \in \mathcal{X}_\epsilon$ then $\sum_{\theta \in \Theta} \mu_\theta x_\theta \ge \epsilon$, while  if $x \not \in \mathcal{X}_\epsilon$ then $\sum_{\theta \in \Theta} \mu_\theta x_\theta \le 6\epsilon$.
\end{restatable}
%
%\bollo{Avrebbe senso cambiare lo statement dicendo solo che con probabilità $1-\delta$ l'evento $\mathcal{E}_1$ avviene  e $| \hat{\mu}_\theta - \mu_\theta |  \le \epsilon$ per ogni theta?} \bac{lo staetment di questo lemma dovreebe assomigliare il piu possibile a quello del lemma 1, affinche se questo lemma vale, possiamo usare il clean event. }
%We also redefine the clean event of Phase~1 as follows:
%\begin{definition}[Clean event of Phase~1 for sample complexity]
%	We let $\mathcal{E}^{\text{SC}}_1$ be the event under which Algorithm~\ref{alg:estimate_prior} computes a set $\Theta' \subseteq \Theta$ such that if $\mu_{\theta} \le \epsilon$, then $\theta \not \in \Theta'$, and if $\mu_\theta \ge 3\epsilon$, then $\theta \in \Theta'$. 
%	Furthermore, for each $\theta \in \Theta$, it holds that $| \hat{\mu}_\theta - \mu_\theta |  \le \epsilon$, while every $x \in \mathcal{X}_\epsilon$ satisfies $\sum_{\theta  \in \Theta} \mu_\theta x_\theta$.
%\end{definition} 
%
By Lemma~\ref{lem:prior_estimate_sc}, it is possible to show that the event $\mathcal{E}^1$ holds.
Hence, the probability that a signaling scheme including a slice $x \in \mathcal{X}_\epsilon$ actually ``induces'' such a slice is at least $\epsilon$, and, thus, the results concerning the second phase of Algorithm~\ref{alg:main_algorithm} are valid also in this setting.
Finally, whenever the events $\mathcal{E}_1$ and $\mathcal{E}_2$ hold, we can provide an upper bound on the number of rounds required by Algorithm~\ref{alg:main_algorithm_sc} to compute a $\gamma$-optimal signaling scheme as desired.
Formally:
\begin{restatable}{theorem}{MainSC}\label{th:main_th_sc} 
	Given $\gamma \in (0,1)$ and $\eta \in (0,1)$, with probability at least $1-\eta$, Algorithm~\ref{alg:main_algorithm_sc} outputs a $\gamma$-optimal signaling scheme in a number of rounds $T$ such that: 
	%
	%signaling scheme $\phi$ such that $u^\textnormal{s}(\phi) \geq \textnormal{OPT} - \gamma$ in a number of rounds $T$ such that: %\bac{manca un quadrato nel log?}
	%
	\[
		T \leq \widetilde{\mathcal{O}}\left( \frac{n^3}{\gamma^2} \log^2\left(\frac{1}{\eta}\right) \left(d^8B +d\binom{d+n}{d}\right) \right).
	\]
	%
	% With probability at least $1-\eta$ and in $\widetilde{\mathcal{O}}\left( \nicefrac{n^2}{\gamma^2} \log\left(\nicefrac{1}{\eta}\right) \left(d^8B +d\binom{d+n}{d}\right) \right)$ rounds, Algorithm~\ref{alg:main_algorithm_sc} computes a $\gamma$-optimal signaling scheme. 
\end{restatable}
%\bollo{$n^3$ non $n^2$ (o c'è un conto sbagliato nella proof?)}

We conclude by providing two negative results showing that the result above is tight.
%that the sample complexity result established by means of Algorithm~\ref{alg:main_algorithm_sc} is tight.
%
% In the following, we introduce three negative results to compute an approximately optimal signaling scheme. In the first result we show that to compute an (approximately) optimal signaling scheme the sender is required to interact with the receiver in a number of rounds which is exponential in the number of agent's actions in the worst case scenario. Formally, the following holds.
%
%\alb{Riformulare in maniera coerente i risultati sotto.}
%
\begin{restatable}{theorem}{hardnessfirstsc}\label{thm:hardness1_sc}
	There exist two absolute constants $\kappa,\lambda >0$ such that no algorithm is guaranteed to return a $\kappa$-optimal signaling scheme with probability of at least $1 - \lambda$ by employing less than $2^{\Omega(n)}$ and $2^{\Omega(d)}$ rounds, even when the prior distribution $\mu$ is known to the sender.
	%
	%\bac{okay cosi?}
\end{restatable}

\begin{restatable}{theorem}{hardnessthirssc}\label{thm:hardness3_sc}
	Given $\gamma \in (0,\nicefrac{1}{8})$ and $\eta \in (0,1)$, no algorithm is guaranteed to return a $\gamma$-optimal signaling scheme with probability at least $1 - 
	\eta$ by employing less than $ \Omega \big( \frac{1}{\gamma^2}{\log(\nicefrac 1 \eta)} \big)   $ rounds.
	%, when the prior distribution is unknown to the sender.
\end{restatable}

In Appendix~\ref{appendix:sample_comlexity_known_mu} , we also study the case in which the prior $\mu$ is known to the sender.
In such a case, we show that the sample complexity can be improved by a factor $\nicefrac{1}{\gamma}$, which is tight.%\bac{see Appendix \dots} %\bac{we show that, under the assumption that the sender knows the prior, then ...}

\section*{Acknowledgments}
This work was supported by the Italian MIUR PRIN 2022 Project “Targeted Learning Dynamics: Computing Efficient and Fair Equilibria through No-Regret Algorithms”, by the FAIR (Future Artificial Intelligence Research) project, funded by the NextGenerationEU program within the PNRR-PE-AI scheme (M4C2, Investment 1.3, Line on Artificial Intelligence), and by the EU Horizon project ELIAS (European Lighthouse of AI for Sustainability, No. 101120237). This work was also partially supported by project SERICS (PE00000014) under the MUR National Recovery and Resilience Plan funded by the European Union - NextGenerationEU.

%% file: appendix_additional_related_works.tex
\section*{Appendix}
The appendixes are organized as follows:
\begin{itemize}
	\item Appendix~\ref{appendix:additional_rleated} provides a discussion of the previous works most related to ours.
	\item Appendix~\ref{appendix:additional_preliminaries} presents additional preliminaries.
	\item Appendix~\ref{app:proof_regret_ub} presents the omitted proofs from Section~\ref{sec:algo_noregret}.
	\item Appendix~\ref{appendix:find_signaling} presents the proof of Lemma~\ref{lem:find_signaling} from Section~\ref{sec:find_signaling}.
	%, together with a pair of supporting lemmas.
	\item Appendix~\ref{appendix:find_partition} provides a description of the results and the procedures employed in \emph{phase 2}.
	\item Appendix~\ref{appendix:lower_bounds_regret} presents the omitted proofs from Section~\ref{sec:lower_bounds_regret}.
	\item Appendix~\ref{appendix:sample_complexity} presents the omitted proofs and some technical details from Section~\ref{sec:sample_complexity}.
	\item Appendix~\ref{appendix:sample_comlexity_known_mu} discusses the PAC-learning problem when the prior is known.
	%\item Appendix~\ref{sec:hyperplane} provides all the details about  the \texttt{Find-HS} procedure which is employed by the \texttt{Discover-and-Cover} algorithm, including all the technical lemmas (and their proofs) related to such a procedure.
	%\item Appendix~\ref{sec:app_algo} provides all the details about the \texttt{Discover-and-Cover} algorithm that are omitted from the main body of the paper, including the definitions of the parameters required by the algorithm and the proofs of all the results related to it.
	%	\item Appendix~\ref{sec:other_proofs} provides all the other proofs omitted from the main body of the paper.
\end{itemize}

\section{Additional Related Works }\label{appendix:additional_rleated}
\paragraph{Learning in Bayesian persuasion settings} % Most of the works on online Bayesian persuasion study settings where either the prior or the receiver's utility functions are known.
In addition to the works presented in Section~\ref{sec:related_works}, the problem of learning optimal signaling schemes in Bayesian persuasion settings has received growing attention over the last few years. \citet{Camara2020} study an adversarial setting where the receiver does not know the prior, and the receiver's behavior is aimed at minimizing internal regret.
%\bollo{\citet{mansour2022bayesian} study a setting with multiple receivers. Meglio toglierlo del tutto, ha un setting strano in cui neppure il sender vede lo stato di natura, ci vuole un paragrafo solo per spiegarlo}
\citet{Babichenko2022} consider an online Bayesian persuasion setting with binary actions when the prior is known, and the receiver's utility function has some regularities. \citet{FengOnline2002} study the online Bayesian persuasion problem faced by a platform that observes some relevant information about the state of a product and repeatedly interacts with a population of myopic receivers through a recommendation mechanism. \citet{Shipra2023} design a regret-minimization algorithm in an advertising setting based on the Bayesian persuasion framework, assuming that the receiver's utility function satisfies some regularity conditions. 

\paragraph{Online learning in problems with commitment} From a technical point of view, our work is related to the problem of learning optimal strategies in Stackelberg games when the leader has no knowledge of the follower's utility. \citet{letchford2009learning} propose the first algorithm to learn optimal strategies in Stackelberg games. Their algorithm is based on an initial random sampling that may require an exponential number of samples, both in the number of leader's actions $m$ and in the representation precision $L$. \citet{Peng2019} improve the algorithm of \citet{letchford2009learning}, while \citet{bacchiocchi2024sample} further improve the approach by~\citet{Peng2019} by relaxing some of their assumptions. 

Furthermore, our work is also related to the problem of learning optimal strategies in Stackelberg games where the leader and the follower interaction is modelled by a Markov Decision Process. \citet{lauffer2022noregret} study Stackelberg games with a state that influences the leader's utility and available actions. \citet{bai2021sample} consider a setting where the leader commits to a pure strategy and observes a noisy measurement of their utility.

Finally, our work is also related to online hidden-action principal-agent problems, in which a principal commits to a contract at each round to induce an agent to take favorable actions. \citet{ho2015adaptive} initiated the study by proposing an algorithm that adaptively refines a discretization over the space of contracts, framing the model as a multi-armed bandit problem where the discretization provides a finite number of arms to play with. \citet{cohen2022learning} similarly work in a discretized space but with milder assumptions. \citet{zhu2022online} provide a more general algorithm that works in hidden-action principal-agent problems with multiple agent types. Finally, \citet{Bacchiocchi2023learning} study the same setting and propose an algorithm with smaller regret when the number of agent actions is small.

%% file: appendix_additional_preliminaries.tex
\section{Additional preliminaries}\label{appendix:additional_preliminaries}
\subsection{Additional preliminaries on Bayesian persuasion}
%In this section we describe some additional mathematical objects regarding Bayesian persuasion that will be used to prove the lower bonds on the regret and sample complexity.
%
%We recall that the sender-receiver interaction goes as follows: (1) the sender commits to a signaling scheme $\phi$; (2) the sender observes a state of nature $\theta \sim \mu$ and sends a signal $s \sim \phi_\theta$ to the receiver; (3) the receiver updates their belief over states of nature according to the {\em Bayes rule}; and (4) the receiver plays a best-response action $a \in \A$, with sender and receiver getting payoffs $u_\theta(a)$ and $u^\text{s}_\theta(a)$, respectively.

In step~(3) of the sender-receiver interaction presented in Section~\ref{sec:prelimin_persuasion}, after observing $s \in \sset$, the receiver performs a Bayesian update and infers a posterior belief $\xi^s \in  \Delta_\Theta$ over the states of nature, according to the following equation:
\begin{equation}\label{eq:posterior}
	\xi^s_\theta =  \frac{\mu_\theta \, \phi_\theta(s)}{\sum_{\theta'\in\Theta}\mu_{\theta'}\,\phi_{\theta'}(s)} \quad\quad \forall \theta \in \Theta.
\end{equation}
%
%After computing $\xi^s \in \Delta_{\Theta}$, the receiver plays a \emph{best-response}, which corresponds to the action maximizing their expected utility given the current posterior. 
%Thus, the set of best-responses given the signal $s \in \sset$ defined in Equation~\ref{eq:br_set} can be rewritten as:
%\begin{equation*}
%	\A^\phi(s) = \left\{ a_i \in \A \mid \sum_{\theta \in \Theta} \xi^s_\theta u_\theta(a_i) \geq \sum_{\theta \in \Theta} \xi^s_\theta u_\theta(a_j) \quad \forall a_j \in \A  \right\}
%\end{equation*}
%\bollo{Non credo che serva parlare di b.r. regions, non mi sembra che le usiamo}\\
%As a result, the best-response of the follower depends on the posterior $\xi^s \in \Delta_\Theta$ induced by signal $s \in \sset$ sent by the receiver.
%However, a signaling scheme $\phi$ can induce up to $|\sset|$ different posteriors, depending on the signal $s$ that is drawn.

Consequently, given a signaling scheme $\phi$, we can equivalently represent it as a distribution over the set of posteriors it induces. Formally, we say that $\phi$ induces $\gamma : \Delta_{\Theta} \to [0,1]$ if, for each posterior distribution $\xi \in \Delta_{\Theta}$, we have:
\begin{equation*}
	\gamma(\xi)=\sum_{s \in \sset: \xi^s = \xi }\sum_{\theta \in \Theta} \mu_\theta \phi_\theta(s) \quad\quad\textnormal{and} \quad\quad \sum_{\xi \in \textnormal{supp}(\gamma)} \gamma(\xi) = 1.
\end{equation*}
Furthermore, we say that a distribution over a set of posteriors $\gamma$ is consistent, \emph{i.e.}, there exists a valid signaling scheme $\phi$ inducing $\gamma$ if the following holds:
\begin{equation*}
	\sum_{\xi \in \textnormal{supp}(\gamma)} \gamma(\xi) \xi_\theta= \mu_\theta \,\,\,\ \forall \theta \in \Theta.
\end{equation*}
With an abuse of notation, we will sometimes refer to a consistent distribution over a set of posteriors $\gamma$ as a signaling scheme.
This is justified by the fact that there exists a signaling scheme $\phi$ inducing such distribution, but we are interested only in the distribution over the set of posteriors that $\phi$ induces.

\subsection{Additional preliminaries on the representation of numbers}
In the following, we assume that all the numbers manipulated by our algorithms are rational.
Furthermore, we assume that rational numbers are represented as fractions, by specifying two integers which encode their numerator and denominator. 
Given a rational number $q  \in \mathbb{Q}$ represented as a fraction $\nicefrac{b}{c}$ with $b,c \in \mathbb{Z}$, we denote the number of bits that $q$ occupies in memory, called \emph{bit-complexity}, as $B_q := B_b + B_c$, where $B_b$ ($B_c$) is the number of bits required to represent the numerator (denominator).
For the sake of the presentation, with an abuse of terminology, given a vector in $\mathbb{Q}^D$ of $D$ rational numbers represented as fractions, we let its bit-complexity be the maximum bit-complexity among its entries.

Furthermore, we assume that the bit-complexity encoding both the receiver's utility and the prior distribution is bounded. Formally, we denote by $B_\mu$ the bit-complexity of the prior $\mu$, while we assume $B_u $ to be an upper bound to the bit-complexity of each $d$-dimensional vector $u_\theta(a)$ with $\theta \in \Theta$. Moreover, we let $B \coloneqq B_\mu + B_u$. Finally, we also denote with $B_\epsilon$ the bit-complexity of the parameter $\epsilon$ computed by our algorithms, while we denote with $B_{\widehat{\mu}}$ the bit-complexity of the estimator $\widehat{\mu}$ computed by Algorithm~\ref{alg:estimate_prior}.

%% file: appendix_find_search_space.tex
\section{Omitted proofs from  Section~\ref{sec:algo_noregret} }\label{app:proof_regret_ub}

\PriorEstimate*
\begin{proof}
	% \bac{ma anche nell'algoritmo è cosi } In the following we let $\widehat{\mu} \coloneqq \widehat{\mu}_{T_1}$ be the estimator computed by Algorithm~\ref{alg:estimate_prior}. 
	For each $\theta \in \Theta$ we consider two possible cases.
	\begin{enumerate}
		\item If $\mu_\theta > \epsilon$, then we employ the multiplicative Chernoff inequality as follows:
		\begin{equation*}
			\mathbb{P}\left(| \mu_\theta - \widehat{\mu}_\theta | \ge \frac{1}{2} \mu_\theta \right) \le 2 e^{- \frac{ T_1 \mu_\theta}{12}},
		\end{equation*} 
		where $T_1 \in \mathbb{N}_{+}$ is the number of rounds employed to estimate $\widehat{\mu}_\theta $. As a result, by setting the number of rounds to estimate $\mu_\theta$ equal to $T_1 = \left\lceil \nicefrac{12}{ \epsilon}\log\left(\nicefrac{2d}{\delta} \right) \right\rceil $ we get:
		\begin{equation*}
			\mathbb{P}\left(| \mu_\theta - \widehat{\mu}_\theta | \ge \frac{1}{2} \mu_\theta \right)  \le 2 \left(\frac{\delta}{2d}\right)^{\frac{\mu_\theta}{\epsilon}} \le \frac{\delta}{d},
		\end{equation*} 
		since $\mu_\theta > \epsilon$. Then, we get:
%		\begin{equation}\label{eq:chernoff}
%			\ \mathbb{P}\left(\frac{3 \mu_\theta}{2} \ge \widehat{\mu}_\theta \ge  \frac{\mu_\theta}{2} \right) \ge 1-\frac{\delta}{d}.
%		\end{equation} 
		\begin{equation}\label{eq:chernoff}
			\ \mathbb{P}\left(\frac{{\mu}_\theta}{2} \le \widehat{\mu}_\theta \le \frac{3 {\mu}_\theta}{2} \right) \ge 1-\frac{\delta}{d}.
		\end{equation} 
		
		%As a result, we have:
		%\begin{equation*}
		%	\mathbb{P}\left(\widehat{\mu}_\theta \ge  2{\epsilon} \, | \, \mu_\theta \ge 4\epsilon \right) \ge 1-\delta.
		% \end{equation*} 
	%Consequently, with a probability of at least $1-\delta/d$, if $\mu_\theta \ge 4 \epsilon$, then $\widehat{\mu} \ge 2\epsilon$ and $\theta \in \widetilde{\Theta}$.
	Consequently, with a probability of at least $1-\delta/d$ the estimator $\widehat{\mu}_\theta$ is such that $ \widehat{\mu}_\theta \in [\nicefrac{{\mu}_\theta}{2},\nicefrac{3{\mu}_\theta}{2}]$.
	Thus, if $\mu_\theta \ge 6 \epsilon$, then $\widehat{\mu}_{\theta} \ge 3\epsilon>2\epsilon$ and $\theta \in \widetilde{\Theta}$. 
	We also notice that there always exits a $\theta \in \Theta$ such that $\mu_\theta \ge \nicefrac{1}{d} \ge 6 \epsilon$, since $\epsilon \in (0, \nicefrac{1}{6d})$. 
	Consequently, with a probability of at least $1-\delta/d$, there always exist a $\theta \in \widetilde \Theta$.
	%Thus, if $\theta \not \in \widetilde{\Theta}$, then $\mu_\theta \le 4 \epsilon$ with a probability of at least $1-\delta/d$.
	
	%		\item If $\mu_\theta \ge 4\epsilon$, then we employ the multiplicative Chernoff inequality as follows:
	%		\begin{equation*}
		%			\mathbb{P}\left(\widehat{\mu}_\theta \le \frac{1}{2} \mu_\theta \right) \le e^{- \frac{ n \mu_\theta}{8}}.
		%		\end{equation*} 
	%		Thus, by setting the number of rounds employed to estimate $\mu_\theta$ equal to $n=\frac{32}{\epsilon}\log(\frac{1}{\delta})$, as done in Algorithm~\dots at Line~\dots, we get:
	%		\begin{equation*}
		%			\mathbb{P}\left(\widehat{\mu}_\theta \le  \frac{\mu_\theta}{2} \right) \le  \delta^{\frac{\mu_\theta}{4\epsilon}} \le \delta,
		%		\end{equation*} 
	%		since $\mu_\theta \ge 4\epsilon$. Then, we get:
	%		\begin{equation*}
		%			\mathbb{P}\left(\widehat{\mu}_\theta \ge  2{\epsilon} \right) \ge \mathbb{P}\left(\widehat{\mu}_\theta \ge  \frac{\mu_\theta}{2} \right) \ge 1-\delta.
		%		\end{equation*} 
	%		%As a result, we have:
	%		%\begin{equation*}
	%		%	\mathbb{P}\left(\widehat{\mu}_\theta \ge  2{\epsilon} \, | \, \mu_\theta \ge 4\epsilon \right) \ge 1-\delta.
	%		% \end{equation*} 
%		Consequently, with a probability of at least $1-\delta$, if $\mu_\theta \ge 4 \epsilon$, then $\widehat{\mu} \ge 2\epsilon$ and $\theta \in \widetilde{\Theta}$. Thus, if $\theta \not \in \widetilde{\Theta}$, then $\mu_\theta \le 4 \epsilon$ with a probability of at least $1-\delta$.

\item If $\mu_\theta \le \epsilon$, then we employ the multiplicative Chernoff inequality as follows:
\begin{equation*}
	\mathbb{P}\left(\widehat{\mu}_\theta \ge (1+c) \mu_\theta \right) \le e^{- \frac{c^2 T_1 \mu_\theta}{2+c}},
\end{equation*} 
with $c=\epsilon / \mu_\theta$. Thus, by setting the number of rounds employed to estimate $\mu_\theta$ equal to $T_1 = \left\lceil \nicefrac{12}{ \epsilon}\log\left(\nicefrac{2d}{\delta} \right) \right\rceil $, we get:
\begin{align*}
	%\mathbb{P}\left(\widehat{\mu}_\theta \ge  2 \epsilon \right) \le \mathbb{P}\left(\widehat{\mu}_\theta \ge  \mu_\theta + \epsilon \right) & \le \exp \left( - \frac{ \frac{\epsilon^2}{\mu_\theta^2} \left(\frac{12}{\epsilon}\log\left(\frac{2d}{\delta}\right)\right)\mu_\theta}{ 2 + \frac{\epsilon}{\mu_\theta} }\right) \\
	\mathbb{P}\left(\widehat{\mu}_\theta \ge  \mu_\theta + \epsilon \right) & \le \exp \left( - \frac{ \frac{\epsilon^2}{\mu_\theta^2} \left(\frac{12}{\epsilon}\log\left(\frac{2d}{\delta}\right)\right)\mu_\theta}{ 2 + \frac{\epsilon}{\mu_\theta} }\right) \\
	& \le \exp \left( - \frac{  \frac{12 \epsilon}{\mu_\theta} \log\left(\frac{2d}{\delta}\right)
	}{ 2 + \frac{\epsilon}{\mu_\theta} }\right) \\
	& \le \left(\frac{\delta}{2d}\right)^{4} \le \frac{\delta}{d}, \\
\end{align*} 
since $x/(x+2) \ge 1/3$, for each $x \ge 1$. As a result, we have: 
%\begin{equation*}
%	\mathbb{P}\left(\widehat{\mu}_\theta \le  2{\epsilon}\right) \ge 1-\frac{\delta}{d}. %\, | \, \mu_\theta \le \epsilon \right
%\end{equation*} 
\begin{equation*}
	\mathbb{P}\left(\widehat{\mu}_\theta \le  \mu_\theta + \epsilon \right) \ge 1-\frac{\delta}{d}.
\end{equation*} 
%Thus, with a probability of at least $1-\delta/d$, if $\mu_\theta \le \epsilon$, then $\widehat{\mu} \le 2\epsilon$ and $\theta \not \in \widetilde{\Theta}$. 
Thus, with a probability of at least $1-\delta/d$, if $\mu_\theta \le \epsilon$, then $\widehat{\mu}_\theta \le  \mu_\theta + \epsilon$, which implies that $\widehat{\mu}_\theta \le 2\epsilon$.
Furthermore, if $\mu_\theta \le \epsilon$, then $\widehat{\mu}_\theta \le 2\epsilon$ and $\theta \not \in \widetilde{\Theta}$.
%This, in turn, implies that if $\theta  \in \widetilde{\Theta}$, then $\mu_\theta \ge  \epsilon$ with a probability of at least $1 - \delta$.
\end{enumerate} 

%Thus, by employing a union bound over the set of nature, we have that if $\mu_\theta \le \epsilon$, its corresponding estimate $\widehat\mu_\theta$ falls within the interval $[0, 2\epsilon]$ while if $\mu_\theta \ge 6\epsilon$, then its corresponding estimate is such that $\widehat\mu_\theta > 2\epsilon$, with probability at least $1 - \delta$. 
Thus, by employing a union bound over the set of natures, we have that if $\mu_\theta \le \epsilon$, then its corresponding estimate $\widehat\mu_\theta$ falls within the interval $[0, 2\epsilon]$ and $\theta \not \in \widetilde{\Theta}$, while if $\mu_\theta \ge 6\epsilon$, then its corresponding estimate is such that $\widehat\mu_\theta > 2\epsilon$ and $\theta \in \widetilde{\Theta}$, with a probability of at least $1 - \delta$.
We also notice that, with the same probability, the set $\widetilde\Theta$ is always non empty.
 % Thus, with probability $1-\delta$, all the 
%
%Thus, with the same probability, we have that all the $\theta \not \in \widetilde{\Theta}$ are such that $\mu_\theta \le 6\epsilon$.
%
%Thus, the set $\widetilde{\Theta}$ consists only of those $\theta \in \Theta$ that satisfy $\mu_\theta \ge \epsilon$ and for which the condition in Equation~\ref{eq:chernoff} holds.
%\bac{aggiungere $\theta \not \in \widetilde{\Theta}$}
%Furthermore, for each $x\in \mathcal{X}_\epsilon$ the following holds:
%\begin{equation*}
%	2\epsilon \le \sum_{\theta \in \widetilde{\Theta}} \widehat{\mu}_\theta x_\theta \le \frac{3}{2} \sum_{\theta \in \widetilde{\Theta}} \mu_\theta x_\theta \le \frac{3}{2} \sum_{\theta \in \Theta} \mu_\theta x_\theta.
%\end{equation*}
%As a result, the probability of observing the signal $s_1$ given a two signals signaling scheme $x \in \mathcal{X}_\epsilon$ can be lower bounded as follows:
%\begin{equation*}
%	\sum_{\theta \in \Theta} \mu_\theta \ge \frac{4}{3} \epsilon \ge \epsilon.
%\end{equation*}

Consequently, for each slice $x\in \mathcal{X}_\epsilon$ with respect to a signal $s$, the probability of observing $s$ can be lower bounded as follows:
\begin{equation*}
	\epsilon 
	\le \frac{1}{2} \sum_{\theta \in \widetilde{\Theta}} \widehat{\mu}_\theta x_\theta 
	\le \frac{3}{4} \sum_{\theta \in \widetilde{\Theta}} \mu_\theta x_\theta 
	\le \sum_{\theta \in \widetilde{\Theta}} \mu_\theta x_\theta 
	\le \sum_{\theta \in \Theta} \mu_\theta x_\theta,
\end{equation*}
where the inequalities above hold because of the definition of $\mathcal{X}_\epsilon$ and observing that each $\theta \in \widetilde{\Theta}$ satisfies Equation~\ref{eq:chernoff} with probability at least $1-\delta$. 
%\bac{spiegare}

Furthermore, for each $x \not \in \mathcal{X}_\epsilon$, the two following conditions hold:
\begin{align*}
	 \frac{1}{2} \sum_{\theta \in \widetilde{\Theta}} \mu_\theta x_\theta
	\leq \sum_{\theta \in \widetilde{\Theta}} \widehat{\mu}_\theta x_\theta
	\leq 2 \epsilon \quad \textnormal{and}\quad
	 \sum_{\theta \not \in \widetilde{\Theta}} \mu_\theta x_\theta
	\leq 6 \epsilon \sum_{\theta \not \in \widetilde{\Theta}} x_\theta
	\leq 6 \epsilon, 
\end{align*}
with probability at least $1-\delta$. Thus, by putting the two inequalities above together, for each $x \not \in \mathcal{X}_\epsilon$, we have:
\begin{equation*}
	\sum_{\theta \in \Theta} \mu_\theta x_\theta
	\leq 10 \epsilon,
\end{equation*}
with probability at least $1-\delta$, concluding the proof.
%Finally, by employing a union bound over the set of nature, we have that for each $\theta \in \Theta$ such that $\mu_\theta \le \epsilon$, its corresponding estimate $\widehat\mu_\theta$ falls within the interval $[0, \mu_\theta + \epsilon] \subseteq [0, 2\epsilon]$, while if $\mu_\theta \ge \epsilon$ it holds that $\nicefrac{\mu_\theta}{2} \le \widehat{\mu}_\theta \le \nicefrac{3\mu_\theta}{2}$.
%As a result, the following holds:
%\begin{align*}
%	2\epsilon \le \sum_{\theta  \in \Theta} \widehat{\mu}_\theta x_\theta &= \sum_{\mu_\theta \ge \epsilon} \widehat{\mu}_\theta x_\theta +\sum_{\mu_\theta < \epsilon} \widehat{\mu}_\theta x_\theta \\
%	&\le \sum_{\mu_\theta \ge \epsilon} \frac{3}{2}\mu_\theta x_\theta +\sum_{\mu_\theta < \epsilon}(\mu_\theta +\epsilon)x_\theta \\
%	&\le \frac{3}{2}\sum_{\theta \in \Theta} \mu_\theta x_\theta +\sum_{\mu_\theta < \epsilon} \epsilon x_\theta \\
%	&\le \frac{3}{2}\sum_{\theta \in \Theta} \mu_\theta x_\theta +\epsilon.
%\end{align*}
%Consequently, we can lower bound the probability that the signal $s_1$ is drawn given a two signaling scheme $x \in \mathcal{X}_\epsilon$ as:
%\begin{equation*}
%	\sum_{\theta \in \Theta} \mu_\theta x_\theta \ge \frac{2}{3}\epsilon
%\end{equation*}
\end{proof}

\PriorRounds*
% \bollo{Under the event $\mathcal{E}_1$, given $\rho > 0$ and a slice $x \in \mathcal{X}_\epsilon$ of a signaling scheme $\phi$ with respect to the signal $s$, if the sender commits to $\phi$ for at least $q \coloneqq \left\lceil \nicefrac{\log(1/ \rho)}{\epsilon} \right \rceil$ rounds, then the probability of observing the signal $s$ is at least $1-\rho$.}
\begin{proof}
	In the following, we let $\tau$ be the first round in which the sender commits to $\phi$.
	The probability of observing the signal $s_1$ at a given round $t \ge \tau$ is can be lower bounded as follows:
	\begin{equation*}
		\mathbb{P}\left(s^t = s_1 \right) = \sum_{\theta \in \Theta} \mu_\theta x_\theta \ge \epsilon,
	\end{equation*}
	where the inequality holds under the event $\mathcal{E}_1$.
	Thus, at each round, the probability of sampling the signal $s_1 \in \sset$ is greater or equal to $\epsilon > 0$. Consequently, the probability of never observing the signal $s_1 \in \sset$ in $q$ rounds is given by:
	\begin{align*}
		\mathbb{P}\left( \bigcap_{t=\tau}^{\tau+q-1} \{s^t \neq s_1 \} \right) \le (1- \epsilon)^q \le \rho,
	\end{align*} 
	where the last inequality holds by taking $q 
	= \left\lceil \frac{\log(\rho)}{\log(1-\epsilon)} \right \rceil  \le \left\lceil \frac{\log(1/\rho)}{\epsilon} \right \rceil $, for each $\epsilon \in (0,1)$.
	
	As a result, the probability of observing the signal $s_1$ at least once in $q$ rounds is greater or equal to:
	\begin{align*}
		\mathbb{P}\left( \bigcup_{t=\tau}^{\tau+q-1} \{s^t \neq s_1 \} \right)  = 1- \mathbb{P}\left( \bigcap_{t=\tau}^{\tau+q-1} \{s^t \neq s_1 \} \right) \ge 1- \rho,
	\end{align*} 
	concluding the proof.
\end{proof}

%% file: appendix_regret_th.tex
\NoRegretThm*
\begin{proof}
	In the following, we let $\delta = \zeta= \frac{1}{T}$ and $\epsilon= \frac{\lceil\sqrt{Bn} d^4 \rceil}{\lceil\sqrt T \rceil}$, as defined in Algorithm~\ref{alg:main_algorithm}. To prove the theorem, we decompose the regret suffered in the three phases of Algorithm~\ref{alg:main_algorithm}:
	\begin{enumerate}
		\item Phase~1. We observe that the number of rounds to execute the \texttt{Build-Search-Space} procedure (Algorithm~\ref{alg:estimate_prior}) is equal to $ T_1 =  \mathcal{O}\left(\nicefrac{1}{\epsilon}\log \left( \nicefrac{1}{\delta} \right) \log(d) \right).$ Thus, the cumulative regret of Phase~1 can be upper bounded as follows:
		\begin{equation*}
			R_T^1 \le \widetilde{\mathcal{O}} \left(\frac{1}{\epsilon}\log(T) \log(d) \right) 
			\le \widetilde{\mathcal{O}}\left(\frac{\sqrt{T}}{d^4  \sqrt{nB}} \log(d)  \right)	\le \widetilde{\mathcal{O}}\left({\sqrt{T}} \right) .
		\end{equation*}
		% \bollo{Con $\epsilon= \frac{\lceil\sqrt{B} n d^4 \rceil}{\lceil\sqrt T \rceil}$ (vedi commento in fondo) al secondo passaggio è $n$ al posto di $\sqrt{n}$}
		This is because, at each round, the regret suffered during the execution of Algorithm~\ref{alg:estimate_prior} is at most one.
		
		\item Phase~2. Under the event $\mathcal{E}_1$, which holds with probability $1-\delta$, Algorithm~\ref{alg:find_partition} correctly terminates with probability $1-\zeta$. Thus, with probability at least $1-\delta-\zeta$, the number of rounds employed by such algorithm is of the order:
		\begin{align*}
			T_2 &\le \widetilde{\mathcal{O}}\left( \frac{n^2}{\epsilon} \log^2\left(\frac{1}{\zeta}\right)  \left(d^7(B+B_\epsilon+B_{\widehat{\mu}}) +\binom{d+n}{d}\right)  \right)  \\
			&= \widetilde{\mathcal{O}}\left( \frac{n^2}{\epsilon} \log^2(T) \left(d^7(B+B_\epsilon)+\binom{d+n}{d}\right)  \right),
		\end{align*}
		where the last equality holds because $B_{\widehat{\mu}} = \mathcal{O}(\log(\nicefrac{1}{\epsilon})+\log(d)+\log(T))$.
		As a result, by taking the expectation, the regret suffered in Phase~2 by Algorithm~\ref{alg:main_algorithm} can be upper bounded as follows:
		\begin{equation*}
			R_T^2 \le \widetilde{\mathcal{O}}\left(  %n^{\frac{3}{2}} 
			n^{\nicefrac{3}{2}} d^3 \binom{d+n}{d} \sqrt {BT} \right),
		\end{equation*}
		since, at each round, the regret suffered during the execution of Algorithm~\ref{alg:find_partition} is at most one.
		
		\item Phase~3. Let $\tau$ be the number of rounds required by Phase~1 and Phase~2  to terminate.
		Under the events $\mathcal{E}_1$ and $\mathcal{E}_2$, which hold with probability at least $1-\delta-\zeta$, thanks to Lemma~\ref{lem:find_signaling}, the solution returned by Algorithm~\ref{alg:find_signaling_scheme} at each round $t > \tau$ is $\mathcal{O}(dn\epsilon + \nu_t)$-optimal, where we define $\nu_t = \left| \sum_{\theta \in \Theta}  \mu_\theta - \widehat{\mu}_{t,\theta} \right|$. 
		We introduce the following event:
		\begin{equation*}
			E_t=\left\{|  \mu_\theta - \widehat{\mu}_{t,\theta} | \leq \epsilon_t \,\,\,\,  \forall \theta \in \Theta  \right\},
		\end{equation*}
		where we let $\epsilon_t >0$ be defined as follows:
		\begin{equation*}
			\epsilon_t=\sqrt{\frac{\log \left(2dT /\iota\right)}{2(t-\tau)}}, \ \ \ t > \tau.
		\end{equation*}
		Then, by Hoeffding's inequality and a union bound we have:
		\begin{equation*}
			\mathbb{P} \Big( \bigcap_{t > \tau} E_t \Big) \ge 1 - \iota.
		\end{equation*}
		Thus, by setting $\iota = 1/T$, the regret suffered in Phase~3 by Algorithm~\ref{alg:main_algorithm} can be upper bounded as follows:
		\begin{align*}
			R_T^3  & \le  \sum_{t=\tau +1 }^T \left| \sum_{\theta \in \Theta}  \mu_\theta - \widehat{\mu}_{t,\theta} \right|+ \mathcal{O}\left(d\epsilon T \right)\\
			& \le  d \sum_{t= \tau +1 }^T \epsilon_t + \mathcal{O}\left(dn\epsilon T \right)\\
			% & \le  {{\mathcal{O}}}\left( d \log(d T) \sqrt T +  d\epsilon T \right) \\
			&\le \widetilde{\mathcal{O}}\left( d \sqrt T +  dn\epsilon T \right) \\
			&= \widetilde{\mathcal{O}}\left( d \sqrt T +n^{\nicefrac{3}{2}} d^5\sqrt{BT} \right).
		\end{align*}
	\end{enumerate}
	As a result, the regret of Algorithm~\ref{alg:main_algorithm} is in the order of:
	\begin{equation*}
		R_T \le  \widetilde{\mathcal{O}}\left( n^{\nicefrac{3}{2}} d^3\binom{d+n}{d}\sqrt{BT} +n^{\nicefrac{3}{2}} d^5\sqrt{BT} \right)
		=\widetilde{\mathcal{O}}\left(n^{\nicefrac{3}{2}} d^3\binom{d+n}{d}\sqrt{BT}\right),
	\end{equation*}
	%\bollo{Non dovrebbe essere $n^2$ per via del termine di destra?}	\bac{$\binom{n+d}{d} \le n^d$ se $n,d>1$ }
	%\bollo{Non ho capito. Stiamo dicendo che $nd^2 = \mathcal{O}(\binom{n+d}{d})$? Questo credo sia vero}
	concluding the proof.
\end{proof}

%% file: appendix_find_signaling_scheme.tex
\section{Proof of Lemma~\ref{lem:find_signaling} from  Section~\ref{sec:find_signaling}}
\label{appendix:find_signaling}
In order to prove Lemma~\ref{lem:find_signaling}, we first consider an auxiliary LP (Program~\ref{eq:lp_vertices}) that works on the vertices of the regions $\mathcal{X}_\epsilon(a)$.
This is useful to take into account the polytopes $\mathcal{X}_\epsilon(a)$ with null volume.
Indeed, for every action $a \in \mathcal{A}$ such that $\text{vol}(\mathcal{X}_\epsilon(a))=0$, Algorithm~\ref{alg:find_signaling_scheme} takes in input only a face $\mathcal{R}_\epsilon(a)$ of $\mathcal{X}_\epsilon(a)$ such that $\mathcal{V}_\epsilon(a) \subseteq V(\mathcal{R}_\epsilon(a))$.
By working on the vertices of the regions $\mathcal{X}_\epsilon(a)$, we can show that the vertices in $\mathcal{V}_\epsilon(a)$ are sufficient to compute an approximately optimal signaling scheme.

The auxiliary LP that works on the vertices is the following:
\begin{maxi!}
	{\scriptstyle \alpha \ge \mathbf{0}}{  \sum_{x \in \mathcal{V}_\epsilon } \alpha_x \sum_{\theta \in \Theta }  \mu_\theta    x_\theta u^\text{s}_\theta(a(x))}{}{}\label{eq:lp_vertices}
	\addConstraint { \sum_{x \in \mathcal{V}_\epsilon} \alpha_x x_\theta \le 1 \quad \forall \theta \in \Theta}{}{}.
\end{maxi!}

Program~\ref{eq:lp_vertices} takes in input the set of vertices $\mathcal{V}_\epsilon \coloneqq \bigcup_{a \in \mathcal{A}} V(\mathcal{X}_\epsilon(a))$, along with the corresponding best-responses $(a(x))_{x \in \mathcal{V}_\epsilon}$ and the exact prior $\mu$.
 It then optimizes over the non-negative variables $\alpha_x \ge 0$, one for vertex $x \in \mathcal{V}_\epsilon$.
These variables $\alpha_x$ act as weights for the corresponding slices $x \in \mathcal{V}_\epsilon$, identifying a non-normalized slice $\alpha_x x \in \mathcal{X}^\Hsquare$.

In the following, we show that the value of an optimal solution $v^\star$ to Program~\ref{eq:lp_vertices} is at least $v^\star \ge \textnormal{OPT} -\mathcal{O}( \epsilon n d)$.
Then, we prove that the signaling scheme $\phi$ computed by Algorithm~\ref{alg:find_signaling_scheme} achieves a principal's expected utility of at least $v^\star$ minus a quantity related to the difference between the estimated prior $\widehat{\mu}_t$ and with the actual prior $\mu$.
Thus, by considering that $v^\star \ge \textnormal{OPT} -\mathcal{O}( \epsilon n d)$, we will be able to prove Lemma~\ref{lem:find_signaling}.

%As a first step, let us show that it is possible to decompose each slice $x \in \mathcal{X}_\epsilon$ into a weighted sum over the vertices $x' \in \mathcal{V}_\epsilon$ without detriment to the sender's utility.
%Thus, a generic slice $x$ of an optimal signaling scheme can be substituted by a set of slices $x'\alpha_{x'}$ for some $x' \in \mathcal{V}_\epsilon$.
%This property is formalized in the following Lemma:
As a first step, we show that it is possible to decompose each slice $x \in \mathcal{X}_\epsilon$ into a weighted sum of the vertices $x' \in \mathcal{V}_\epsilon$ without incurring a loss in the sender's utility.
Thus, a generic slice $x$ of an optimal signaling scheme can be written as a convex combination of slices $x'$ with $x' \in \mathcal{V}_\epsilon$.
This property is formalized in the following lemma.
\begin{restatable}{lemma}{VertexDecomposition}\label{lem:decomposition_vertexes}
	%Under event $\mathcal{E}_2$, 
	For every $x \in \mathcal{X}_{\epsilon}$, there exists a distribution $\alpha \in \Delta_{\mathcal{V}_\epsilon}$ such that:
	\begin{equation*}
		x_\theta = \sum_{x' \in \mathcal{V}_\epsilon} \alpha_{x'} x_\theta' \quad \forall \theta \in \Theta.
	\end{equation*}
	Furthermore, the following holds: 
	\begin{align*}
		\sum_{\theta \in \Theta} \hspace{-0.5mm}\mu_\theta x_\theta  u^\textnormal{s}_\theta(a(x)) \hspace{-0.5mm} \le \hspace{-1mm} \sum_{x' \in \mathcal{V}_\epsilon} \hspace{-1.5mm}\alpha_{x'} \hspace{-1.5mm} \sum_{\theta \in \Theta} \hspace{-0.5mm} \mu_\theta   x_\theta'  u^\textnormal{s}_\theta(a({x'})).
	\end{align*}
\end{restatable}
\begin{proof}
	%\bollo{Semplificata di molto. Ora $\mathcal{V}_\epsilon$ include i vertici delel regioni a volume nullo epr definizione. Comunque che questi siano anche vertici di regioni a volume non nullo è un corollario dei lemmi dell fase due, ma non ci serve quì visto che l'LP dei vertici non va risolto in pratica}
	
	%\bac{Mi sembra strano che il lemmi funzioni anche se il clean event $E_1$ non funziona, ad esempio $\mathcal{X}_{\epsilon}$ potrebbbe essere vuoto.}
	
	%\bollo{Se $\mathcal{X}_\epsilon = \emptyset$, allora il lemma è giusto perchè diventa per ogni $x \in \emptyset$. In generale, anche se non siamo in $\mathcal{E}_1$, $\mathcal{X}_\epsilon$ è un politopo e $\mathcal{X}_\epsilon(a)$ pure, anche se non possiamo fare query in $1/\epsilon$.}
	
	Let $a \coloneqq a(x)$.
	Since $x \in \mathcal{X}_\epsilon(a)$, by the Carathéodory theorem, there exists an $\alpha \in \Delta_{V(\mathcal{X}_\epsilon(a))}$ such that:
 	\begin{equation*}
 		\sum_{x' \in V(\mathcal{X}_\epsilon(a))}\alpha_{x'}x'_\theta = x_\theta \,\,\, \forall \theta \in \Theta. 
 	\end{equation*}
 	Furthermore:
 	\begin{align*}
 		\sum_{\theta \in \Theta} \mu_\theta x_\theta  u^\text{s}_\theta(a)
 		&= \sum_{\theta \in \Theta} \mu_\theta  \left( \sum_{x' \in V(\mathcal{X}_\epsilon(a))} {\alpha}_{x'} x_\theta' \right)   u^\text{s}_\theta(a).\\
% 		&= \sum_{\theta \in \Theta} \sum_{x' \in V(\mathcal{X}_\epsilon(a))} \mu_\theta {\alpha}_{x'}   x_\theta'  u^\text{s}_\theta(a)\\		
 		&= \sum_{x' \in V(\mathcal{X}_\epsilon(a))} {\alpha}_{x'} \sum_{\theta \in \Theta}  \mu_\theta x_\theta'  u^\text{s}_\theta(a)\\		
 		&\le \sum_{x' \in V(\mathcal{X}_\epsilon(a))} {\alpha}_{x'} \sum_{\theta \in \Theta}  \mu_\theta x_\theta'  u^\text{s}_\theta(a(x'))\\		
% 		&= \sum_{\theta \in \Theta} \sum_{x' \in V(\mathcal{X}_\epsilon(a))} \mu_\theta {\alpha}_{x'} x_\theta'   u^\text{s}_\theta(a({x'})),
 	\end{align*}
 	where the inequality holds because the receiver breaks ties in favor of the sender.
 	Finally, we observe that for each distribution over the set $V(\mathcal{X}_\epsilon(a))$ for a given $\mathcal{X}_\epsilon(a)$, we can always recover a probability distribution supported in $\mathcal{V}_\epsilon$, since $V(\mathcal{X}_\epsilon(a)) \subseteq \mathcal{V}_\epsilon$ by construction.
\end{proof}
Thanks to the the result above, in the next lemma (Lemma~\ref{lem:lp_vertices_approx_opt}) we prove that an optimal solution of Program~\ref{eq:lp_vertices} has value at least $v^\star \geq \textnormal{OPT} -10 \epsilon n d$.
 To show this, we begin by observing that there exists a set $\mathcal{J}$ of slices of the optimal signaling scheme that belong to the search space $\mathcal{X}_\epsilon$.
By applying Lemma~\ref{lem:decomposition_vertexes} to each of these slices, we obtain a feasible solution for Program~\ref{eq:lp_vertices}.
The value of this solution is at least the sender's expected utility given by the slices in $\mathcal{J}$.
Finally, thanks to the properties of the search space $\mathcal{X}_\epsilon$, we can bound the expected sender's utility provided by the slices that lie outside the search space.
 
\begin{restatable}{lemma}{EpsilonSolution1}\label{lem:lp_vertices_approx_opt}
	Under the event $\mathcal{E}_1$, the optimal solution of Program~\ref{eq:lp_vertices} has value at least $v^\star \ge \textnormal{OPT} -10 \epsilon n d$.
\end{restatable}
\begin{proof}
	In the following we let $\phi_\theta \in \Delta_{\A} $ for each $\theta \in \Theta$ be an optimal signaling scheme, where we assume, without loss of generality, such a signaling scheme to be direct, meaning that $\mathcal{S}=\mathcal{A}$ and $a \in \mathcal{A}^\phi(a)$ for every action $a \in \mathcal{A}$.
	%\bollo{Non mi sembra che li abbiamo definiti}
	Furthermore, for each action $a \in \textnormal{supp}(\phi)$, we define:
	%We define a subset of action $\A' \subseteq \A$ as follows: 
%	$$\mathcal{A}' \coloneqq \{a \in \A \,\, |\,\, \exists \theta \in \Theta' \,\ \textnormal{s.t.} \,\, \phi_{\theta'}(a)>0 \}.$$
%	Then, for each action $a \in \A'$, we define:
	\begin{equation*}
		x^{a}_{\theta} = \frac{\phi_\theta(a) }{\sum_{\theta \in \Theta}\phi_\theta(a)  } \,\,\forall \theta \in \Theta \,\,\,\, \textnormal{and}\,\,\,\, \alpha_{x^{a}} = \sum_{\theta \in \Theta}\phi_\theta(a).
	\end{equation*} 	
	% \bollo{Lo chiamerei $\beta(x^a)$, perchè è difficile da seguire quando compaiono $\alpha(x^a)$ e $\alpha^a(x)$}
%	First, we show that $x^{a} \in \mathcal{P}'(a)$. Indeed, for each action $a \in \A'$, the following holds:
	We observe that each $x^{a} \in \mathcal{X}(a)$, indeed we have:
	\begin{equation*}
		\sum_{\theta \in \Theta} \mu_\theta x_\theta  u_\theta(a) = \frac{1}{\alpha_{x^{a}} } \sum_{\theta \in \Theta} \mu_\theta \phi_\theta(a)  u_\theta(a) \ge \frac{1}{\alpha_{x^{a}}} \sum_{\theta \in \Theta} \mu_\theta \phi_\theta(a) u_{\theta}(a')=\sum_{\theta \in \Theta} \mu_\theta x_\theta  u_\theta(a').
	\end{equation*}
	%\bollo{L'idea sopra è provare che $\sum_{\theta \in \Theta} \mu_\theta x_\theta (u_\theta(a) -u_\theta(a')) \ge 0$? Non ho capito i primi due passaggi}\\
	%\bollo{Quella sopra è l'unica equation in cui si usa $u_\theta$ e non $u^s_\theta$?}
	for every action $a' \in \mathcal{A}$. 
	We define the subset of actions $\A'\subseteq \A$ in a way that if $a \in \A'$, then $x^{a} \in \mathcal{X}_\epsilon $, \emph{i.e.}, $\A' \coloneqq \{a \in \A \mid x^a \in \mathcal{X}_\epsilon\}$. 
	Then, we let $\A'' \coloneqq  \textnormal{supp}(\phi) \setminus \A' $.
	% \bollo{Quindi $\A' \coloneqq \{a \in \A \mid x^a \in \mathcal{X}_\epsilon\}$? }
	Furthermore, for each $a \in \A'$, thanks to Lemma~\ref{lem:decomposition_vertexes}, there exists a distribution $\alpha^{a} \in \Delta_{\mathcal{V}_\epsilon}$ such that:
	% By Lemma ... we have $\alpha^{a} \in \Delta$ such that:
		
	\begin{equation*}
		x_\theta^{a} = \sum_{x' \in \mathcal{V}_\epsilon} \alpha^{a}_{x'} x_\theta' ,% = \sum_{x' \in \mathcal{V}_\epsilon} \alpha^{a}(x') x_\theta'.
	\end{equation*}
	% \bollo{Ma l'ultimo passaggio sopra non cambia nulla?}\\
	%
	%the vertices of a best response regions are always a subset of $\mathcal{V}_\epsilon$. Then, thanks to Equation~\ref{eq:decomposition}, the following holds:
	and the following holds:
	\begin{align}
	 \sum_{\theta \in \Theta} \mu_\theta x^{a}_\theta  u_\theta(a)
		 &= \sum_{\theta \in \Theta} \mu_\theta  \left( \sum_{x' \in \mathcal{V}_\epsilon} \alpha^{a}_{x'} x_\theta' \right)   u^\text{s}_\theta(a). \nonumber \\
		 &= \sum_{\theta \in \Theta} \sum_{x' \in \mathcal{V}_\epsilon} \mu_\theta\alpha^{a}_{x'}   x_\theta'  u^\text{s}_\theta(a) \nonumber \\		
		 &\le \sum_{\theta \in \Theta} \sum_{x' \in \mathcal{V}_\epsilon} \mu_\theta\alpha^{a}_{x'} x_\theta'   u^\text{s}_\theta(a({x'})).  \label{eq:decomposition}
	\end{align}
	%
%	\bollo{Propongo di chiamare qualcosa $\beta$, perchè ci sono $\alpha^\star_{x'}, \alpha^a_{x'}$ e $\alpha_{x^a}$ che sono tre cose diverse}\\
	We also define $\alpha^{\star} : \mathcal{V}_\epsilon \to \mathbb{R}_{+} $ as follows:
%	\bollo{$\alpha^\star = \{\alpha^\star_{x'}\}_{x' \in \mathcal{V}_\epsilon}$?}\bac{va bene, quando usiamo parentesi tonde o parentesi graffe?}
	\begin{equation*}
		\alpha^{\star}_{x'} = \sum_{a \in \A'} {\alpha}^{a}_{x'} \alpha_{x^{a}},
	\end{equation*}
	for each $x' \in \mathcal{V}_\epsilon$. First, we show that $\alpha^{\star} : \mathcal{V}_\epsilon \to \mathbb{R}_{+} $ is a feasible solution to LP~\ref{eq:lp_vertices}. Indeed, for each $\theta \in \Theta$, it holds:
	\begin{align*}
		\sum_{x' \in \mathcal{V}_\epsilon} \alpha^{\star}_{x'}   x'_\theta
		& =  \sum_{x' \in \mathcal{V}_\epsilon}  \sum_{a \in \A'} {\alpha}^{a}_{x'} \alpha_{x^{a}}  x'_\theta \\ 
		& =   \sum_{a \in \A'}  \alpha_{x^{a}} \sum_{x' \in \mathcal{V}_\epsilon}  {\alpha}^{a}_{x'} x'_\theta \\ 
		& =   \sum_{a \in \A'} \alpha_{x^{a}} x^{a}_\theta \\
		&=   \sum_{a \in \A'} \phi_\theta(a) \le   1  \\ 
	\end{align*}
	% \bollo{Ma LP~\ref{eq:lp_vertices} ha il constraint $=1$} \bac{e da aggiornare}
	%
	Where the equalities above holds thanks to Equation~\ref{eq:decomposition} and the definition of $\alpha^\star$.
	Then, we show that the utility achieved by $\alpha^{\star}: \mathcal{V}_\epsilon \to \mathbb{R}^{m}_{+}$ is greater or equal to $\textnormal{OPT} - 10 \epsilon d $. Formally, we have:
 	\begin{align*}
		\sum_{x' \in \mathcal{V}_\epsilon} \alpha^{\star}_{x'} \sum_{\theta \in \Theta} \mu_\theta  x'_\theta  u^\text{s}_\theta (a({x'}))
		& =  \sum_{x' \in \mathcal{V}_\epsilon}  \sum_{a \in \A'} \alpha_{x^{a}} {\alpha}^{a}_{x'}  \sum_{\theta \in \Theta}  \mu_\theta  x'_\theta  u^\text{s}_\theta(a({x'})) \\ 
		& =    \sum_{a \in \A'} \alpha_{x^{a}}  \sum_{x' \in \mathcal{V}_\epsilon} \sum_{\theta \in \Theta}  \mu_\theta {\alpha}^{a}_{x'} x'_\theta  u^\text{s}_\theta(a({x'}))\\ 
		& \ge   \sum_{a \in \A'} \alpha_{x^{a}} \sum_{\theta \in \Theta} x^{a}_\theta \mu_\theta  u^\text{s}_\theta(a)\\ 
		& =  \sum_{a \in \A'} \sum_{\theta \in\Theta} \mu_\theta  \phi_\theta(a) u^\text{s}_\theta(a) \\
		& = \textnormal{OPT} - \sum_{a \in \A''} \sum_{\theta \in \Theta} \mu_\theta  \phi_\theta(a) u^\text{s}_\theta(a)\\
		% & \ge \textnormal{OPT} -\left(  \sum_{\theta \in \Theta } \phi_\theta (a) \right) \sum_{a \in \A''} \sum_{\theta \in \Theta} \mu_\theta  x^{a}_\theta(a) u^\text{s}_\theta(a) \\
		& = \textnormal{OPT} - \sum_{a \in \A''} \sum_{\theta \in \Theta} \mu_\theta  \alpha_{x^a} x^{a}_\theta u^\text{s}_\theta(a) \\
		& \ge \textnormal{OPT} -d \sum_{a \in \A''} \sum_{\theta \in \Theta} \mu_\theta  x^{a}_\theta u^\text{s}_\theta(a) \\
		% & \ge \textnormal{OPT} -d \left( \sum_{a \not \in \A'} \sum_{\theta \in \Theta} \mu_\theta  x^{a}_\theta(a) u^\text{s}_\theta(a) + \sum_{a \not \in \A'} \sum_{\theta \not \in \Theta} \mu_\theta  x^{a}_\theta(a) u^\text{s}_\theta(a) \right)  \\
		%	& \ge \textnormal{OPT} -d \left( \sum_{a \not \in \A'} \sum_{\theta \in \Theta} \mu_\theta  x^{a}_\theta(a) u^\text{s}_\theta(a) + \sum_{a \not \in \A'} \sum_{\theta \not \in \Theta} \mu_\theta  x^{a}_\theta(a) u^\text{s}_\theta(a) \right)  \\
		& \ge \textnormal{OPT} -10 \epsilon n d. \\
		% & \ge \textnormal{OPT} - \sum_{a \in \A} \sum_{\theta \not \in \Theta} \mu_\theta  \phi_\theta(a) u^\text{s}_\theta(a) -  \sum_{a \not \in \A'} \sum_{\theta \in \Theta} \mu_\theta  \phi_\theta(a) u_\theta(a) \\
		% & \ge \textnormal{OPT} - 12 \epsilon d 
	\end{align*}	 
	Where the first inequality holds thanks to Inequality~\eqref{eq:decomposition}, the second inequality holds since $\alpha_{x^a} \le d$ and the last inequality holds since, for each $x \not \in \mathcal{X}_\epsilon$, it holds $\sum_{\theta \in \Theta} \mu_\theta  x^{a}_\theta(a) \le 10 \epsilon$, under the event $\mathcal{E}_1$ and $x^a \notin \mathcal{X}_\epsilon$ for every $a \in \A''$.
	% \bollo{Il secondo e terzo $\ge$ non sono $=$?}
	% \bollo{$x^a_\theta(a)$ credo sia solo $x^a_\theta$}
%	\bollo{Forse è meglio chiarire dove usiamo il Lemma~\ref{lem:decomposition_vertexes} e che stiamo facendo un upper bound di $\alpha(x^a) \le d$ e $u_\theta(a) \le 1$,  così diventa più facile da seguire}
%	where the last inequality holds since, for each $x \not \in \mathcal{X}_\epsilon$, it holds $\sum_{\theta \in \Theta} \mu_\theta  x^{a}_\theta(a) \le 10 \epsilon$, under the event $\mathcal{E}_1$ and $x^a \notin \mathcal{X}_\epsilon$ for every $a \in \A''$. % \bollo{Questo è in $\mathcal{E}_1$} \bollo{, and $x^a \notin \mathcal{X}_\epsilon$ for every $a \in \A''$}. 
	Consequently, $\alpha^\star$ is a feasible solution to LP~\ref{eq:lp_vertices} and provides, under the event $\mathcal{E}_1$, a value of at least $\textnormal{OPT}- 10 \epsilon n d$. 
	As a result, under the event $\mathcal{E}_1$ the optimal solution of LP~\ref{eq:lp_vertices} has value $v^\star \geq \textnormal{OPT}- 10 \epsilon n d$, concluding the proof.
%		
%	Finally, given a feasible solution $\alpha$ to LP~\ref{eq:lp_two_signals}, we show how to recover a signaling scheme achieving greater or equal utility. To do so, we define:
%	\begin{equation*}
%		\phi' = \begin{cases}
%				\phi_{\theta}'(s^{x}) = \alpha_x x_{\theta} \,\,\ \forall x \in \textnormal{supp}(\alpha),\\
%					\phi_{\theta}'(s^{\star}) = 1 - \sum_{x' \in \mathcal{V} }\alpha_{x'} x'_{\theta}.
%		\end{cases}
%	\end{equation*}
%	With the same calculation above, it is easy to verify that the signaling scheme $\phi'$ achieve utility greater or equal to the one achieved by $\alpha$, concluding the proof.
\end{proof}

\FindSignaling*
\begin{proof}
	%	In the following we let $\mathcal{V}$ be the set of all vertices of the regions $\mathcal{X}^{\triangle}_\epsilon(a)$, formally:
	%	\begin{equation*}
		%		\mathcal{V} \coloneqq \bigcup_{a \in \mathcal{A}} V(\mathcal{X}^{\triangle}_\epsilon(a)).
		%	\end{equation*}
	%	Such a set is a parameter of Program~\ref{eq:lp_two_signals}.
	%	With an abuse of notation, we define the set of vertices where an action $a \in \mathcal{A}$ is a best-response as:
	%	\begin{equation*}
		%		\mathcal{V}(a) \coloneqq \{x \in \mathcal{V} \mid a(x)=a\}.
		%	\end{equation*}
	
	We observe that under the events $\mathcal{E}_1$ and $\mathcal{E}_2$, the collection $\mathcal{R}_\epsilon = \{\mathcal{R}_\epsilon(a)\}_{a \in \mathcal{A}}$ is composed of faces $\mathcal{R}_\epsilon(a)$ of $\mathcal{X}_\epsilon(a)$ (possibly the improper face $\mathcal{X}\epsilon(a)$ itself) such that every vertex $x \in V(\mathcal{X}_\epsilon(a))$ that satisfies $a(x)=a$ belongs to $\mathcal{R}_\epsilon(a)$.
	The following statements hold under these two events.
	
	As a first step, we prove that given a feasible solution $\alpha = (\alpha_x)_{x \in \mathcal{V}_\epsilon}$ to Program~\ref{eq:lp_vertices}, one can construct a feasible solution $\varphi$ to Program~\ref{eq:lp_h_repr} with the same value.
	In particular, we consider a solution $\varphi = (\widetilde{x}^a)_{a \in \mathcal{A}}$ defined as:
	\begin{equation*}
		\widetilde{x}^a_\theta \coloneqq \sum_{x \in \mathcal{V}_\epsilon(a)} \alpha_x x_\theta \quad \forall a \in \mathcal{A}, \theta \in \Theta.
	\end{equation*}
	
	We observe that for every $a \in \mathcal{A}$ and $\theta \in \Theta$ we can bound $\widetilde{x}^a_\theta$ as follows:
	\begin{align*}
		0 \le \widetilde{x}^a_\theta = \sum_{x \in \mathcal{V}_\epsilon(a)} \alpha_x x_\theta \le \sum_{x \in \mathcal{V}_\epsilon} \alpha_x x_\theta \le 1,
	\end{align*}
	where the last inequality holds due to the constraints of Program~\ref{eq:lp_vertices}.
	Consequently, the vectors $\widetilde{x}^a$ belong to $\mathcal{X}^{\Hsquare}$.
	
	Now we show that $\widetilde{x}^a$ belongs to $\mathcal{R}^\Hsquare_\epsilon(a)$ for every $a \in \mathcal{A}$.
	This holds trivially for every action $a \in \mathcal{A}$ such that $\sum_{x' \in \mathcal{V}_\epsilon(a)} \alpha_{x'} = 0$, as $\widetilde{x}^a = \mathbf{0} \in \mathcal{R}^\Hsquare_\epsilon(a)$.
	Consider instead an action $a \in \mathcal{A}$ such that $\sum_{x' \in \mathcal{V}_\epsilon(a)} \alpha_{x'} > 0$, and let us define the coefficient:
	\begin{equation*}
		\beta^a_x \coloneqq \frac{\alpha_x}{\sum_{x' \in \mathcal{V}_\epsilon(a)} \alpha_{x'}},
	\end{equation*}
	for every vertex $x \in \mathcal{V}_\epsilon(a)$.
	One can easily verify that $\beta^a \in \Delta_{\mathcal{V}_\epsilon(a)}$.
	Now consider the normalized slice:
	\begin{equation*}
		\widetilde{x}^{\text{N},a} \coloneqq \sum_{x \in \mathcal{V}_\epsilon(a)} \beta^a_x x.
	\end{equation*}
	This slice belongs to $\mathcal{R}^{\triangle}_\epsilon(a) \coloneqq \mathcal{R}_\epsilon(a)$, as it is the weighted sum of the vertices $\mathcal{V}_\epsilon(a) \subseteq V(\mathcal{R}^{\triangle}_\epsilon(a))$ with weights $\beta^a \in \Delta_{\mathcal{V}_\epsilon(a)}$.
	Furthermore, we can rewrite the component $\widetilde{x}^a$ of the solution $\phi$ as:
	\begin{equation*}
		\widetilde{x}^a = \widetilde{x}^{\text{N},a} \sum_{x \in \mathcal{V}_\epsilon(a)} \alpha_x.
	\end{equation*} 
	Thus, by considering that $\widetilde{x}^{\text{N},a} \in \mathcal{R}^{\triangle}_\epsilon(a)$ and $\widetilde{x}^a \in \mathcal{X}^{\Hsquare}$, we have that $\widetilde{x}^a \in \mathcal{R}^{\Hsquare}_\epsilon(a)$.

	Finally, since $\alpha$ is a feasible solution for Program~\ref{eq:lp_vertices}, we can observe that:
	\begin{equation*}
		\sum_{a \in \mathcal{A}} \widetilde{x}^a_\theta = \sum_{a \in \mathcal{A}} \sum_{x \in \mathcal{V}_\epsilon(a)} \alpha_x x = \sum_{x \in \mathcal{V}_\epsilon} \alpha_x x \le 1.
	\end{equation*}
	As a result, $\varphi = (\widetilde{x}^a)_{a \in \mathcal{A}}$ is a feasible solution to Program~\ref{eq:lp_h_repr}.
	
	If the estimator $\widehat{\mu}_t$ coincides with the exact prior $\mu$, then direct calculations show that the solution $\varphi$ to Program~\ref{eq:lp_h_repr} achieves the same value of the solution $\alpha$ to Program~\ref{eq:lp_vertices}.
	
	It follows that, when $\widehat{\mu}_t = \mu$, the optimal solution of Program~\ref{eq:lp_h_repr} has at least the same value of the optimal solution of Program~\ref{eq:lp_vertices}.
	
	In order to conclude the proof, we provide a lower bound on the utility of the signaling scheme computed by Algorithm~\ref{alg:find_signaling_scheme}.
	Let $\phi^{\textnormal{LP}}$ be the signaling scheme computed by Algorithm~\ref{alg:find_signaling_scheme}, while let $\Psi = (x^{\text{LP},a})_{a \in \mathcal{A}}$ be the optimal solution to Program~\ref{eq:lp_h_repr}.
	Furthermore, we let $\psi = (x^{\text{E},a})_{a \in \mathcal{A}}$ be the optimal solution of Program~\ref{eq:lp_h_repr} and $\phi^\text{E}$ the signaling scheme computed by Algorithm~\ref{alg:find_signaling_scheme} when the prior estimator coincides \emph{exactly} with the prior itself, \emph{i.e.}, $\widehat{\mu}:= \widehat \mu_t = \mu$.
	
	Since $x^{\text{LP},a} \in \mathcal{R}^{\Hsquare}_\epsilon(a)$ for every $a \in \mathcal{A}$, we have that $a \in \mathcal{A}^{\phi^{\textnormal{LP}}}(a)$.\footnote{Observe that when $\mathcal{R}^{\triangle}_\epsilon(a) = \emptyset$ and $\mathcal{R}^{\Hsquare}_\epsilon(a) = \{\mathbf{0}\}$, we have $x^{\text{LP},a} = \mathbf{0}$. Thus, $x^{\text{LP},a}$ does not contribute to the sender's utility, $s^a \notin \text{supp}(\phi^{\text{LP}})$ and $\mathcal{A}^{\phi^{\text{LP}}}(s^a) = \mathcal{A}$ by definition.}
	Breaking ties in favor of the sender, the action $a^{\phi^{\textnormal{LP}}}(s^a)$ is such that:
	\begin{equation*}
		\sum_{\theta \in \Theta}  \mu_\theta x^{\text{LP},a}_{\theta} u^\textnormal{s}_\theta (a^{\phi^{\textnormal{LP}}}(s^a)) \ge \sum_{\theta  \in \Theta}  \mu_\theta x^{\text{LP},a}_{\theta} u^\textnormal{s}_\theta (a).
	\end{equation*}
	Then: 
	\begin{align*}
		u (\phi^{\textnormal{LP}}) & =  \sum_{s \in \mathcal{S} }  \sum_{\theta  \in \Theta}  \mu_\theta \phi^{\textnormal{LP}}_{\theta}(s) u^\textnormal{s}_\theta (a^{\phi^{\textnormal{LP}}}(s)) \\
		&\ge \sum_{s \in \mathcal{S} \setminus \{s^\star\} }  \sum_{\theta  \in \Theta}  \mu_\theta \phi^{\textnormal{LP}}_{\theta}(s) u^\textnormal{s}_\theta (a^{\phi^{\textnormal{LP}}}(s)) \\
		&= \sum_{a \in \mathcal{A}}  \sum_{\theta  \in \Theta}  \mu_\theta x^{\text{LP},a}_{\theta} u^\textnormal{s}_\theta (a^{\phi^{\textnormal{LP}}}(s^a)) \\
		&\ge \sum_{a \in \mathcal{A}}  \sum_{\theta  \in \Theta}  \mu_\theta x^{\text{LP},a}_{\theta} u^\textnormal{s}_\theta (a) \\
		&\ge \sum_{a \in \mathcal{A}}  \sum_{\theta  \in \Theta}  \widehat{\mu}_\theta x^{\text{LP},a}_{\theta} u^\textnormal{s}_\theta (a) - \left|\sum_{\theta \in \Theta} \widehat{\mu}_\theta - \mu_\theta \right| \\
		&\ge \sum_{a \in \mathcal{A}}  \sum_{\theta  \in \Theta}  \widehat{\mu}_\theta x^{\text{E},a}_{\theta} u^\textnormal{s}_\theta (a) - \left|\sum_{\theta \in \Theta} \widehat{\mu}_\theta - \mu_\theta \right| \\
		&\ge \sum_{a \in \mathcal{A}}  \sum_{\theta  \in \Theta}  \mu_\theta x^{\text{E},a}_{\theta} u^\textnormal{s}_\theta (a) - 2\left|\sum_{\theta \in \Theta} \widehat{\mu}_\theta - \mu_\theta \right|.
	\end{align*}
	We recall that the value of $\psi = (x^{\text{E},a})_{a \in \mathcal{A}}$ is at least the optimal value of Program~\ref{eq:lp_vertices} when $\widehat{\mu} = \mu$, and such a value is at least $v^\star \ge \text{OPT}-10\epsilon n d$ according to Lemma~\ref{lem:lp_vertices_approx_opt}, 
	Thus, we have that:
	\begin{align*}
		u (\phi^{\textnormal{LP}}) \ge \sum_{a \in \mathcal{A}}  \sum_{\theta  \in \Theta}  \mu_\theta x^{\text{E},a}_{\theta} u^\textnormal{s}_\theta (a) - 2\left|\sum_{\theta \in \Theta} \widehat{\mu}_\theta - \mu_\theta \right| \ge \text{OPT}-10\epsilon n d - 2\left|\sum_{\theta \in \Theta} \widehat{\mu}_\theta - \mu_\theta \right|,
	\end{align*}
	concluding the proof.
\end{proof}

%% file: appendix_find_partition.tex
\section{Omitted proofs and sub-procedures of Phase~2}\label{appendix:find_partition}

\subsection{\texttt{Action-Oracle}}
%\bollo{Decidere se togliere il limite q di round - il clean even diventa l'event in cui Action-Oracle ritorna la best response in un numer ofinito di round}
%The function $\texttt{Action-Oracle}$ (Algorithm~\ref{alg:action_oracle}) samples a two signals signaling scheme $x \in \mathcal{X}_\epsilon$. 
%It commits to $x$ multiple rounds, until the signal $s_1$ is drawn or $q$ rounds have elapsed.
%Thanks to the definition of $\mathcal{X}_\epsilon$, we can bound the number $q$ of rounds required to sample $x$ with large probability, as stated by Lemma~\ref{lem:rounds_signal}.
The goal of the $\texttt{Action-Oracle}$ procedure (Algorithm~\ref{alg:action_oracle}) is to assign the corresponding best-response to a slice $x \in \mathcal{X}_\epsilon$ received as input.
In order to do so, it repeatedly commits to a signaling scheme $\phi$ such that $x$ is the slice of $\phi$ with respect to the signal $s_1$.
When the signal $s_1$ is sampled, the procedure returns the best-response $a(x)$.
%The $\texttt{Action-Oracle}$ procedure (Algorithm~\ref{alg:action_oracle}) receives as input a two signals signaling scheme $x \in \mathcal{X}_\epsilon$ and prescribes the sender to commit to such $x$ over multiple rounds until the signal $s_1$ is sampled. Finally, once the signal $s_1$ is sampled, it returns the receiver's best response $b_{x}=b_{\xi^{s_1}}$.

\begin{algorithm}[H]
	\caption{\texttt{Action-Oracle}}\label{alg:action_oracle}
	\begin{algorithmic}[1]
		\Require $x \in \mathcal{X}_\epsilon$
		\State $\phi_{\theta}(s_1) \gets x_\theta$ and $\phi_{\theta}(s_2) \gets 1- x_\theta$ $\forall \theta \in \Theta$
		\Do
		\State Commit to $\phi^t=\phi$, observe $\theta^t$, and send $s^t$
		\State Observe feedback $a^t$
		\State $a \gets a^t$
		\doWhile{$s^t \ne s_1$}
		\State \textbf{return} $a$						
	\end{algorithmic}
\end{algorithm}

In the following, for the sake of analysis, we introduce the definition of a \emph{clean event} under which the $\texttt{Action-Oracle}$ procedure always returns the action $a(x) \in \A$, as formally stated below.
\begin{definition}[Clean event of \texttt{Action-Oracle}]
	We denote $\mathcal{E}^\textnormal{a}$ as the event in which Algorithm~\ref{alg:action_oracle} correctly returns the follower's best response $a(x)$ whenever executed.
\end{definition}
In the proof of Lemma~\ref{lem:final_partition} we show that, thanks to Lemma~\ref{lem:rounds_signal} and the definition of $\mathcal{X}_\epsilon$, it is possible to bound the number of rounds required to Algorithm~\ref{alg:action_oracle} to ensure that it always returns the best response $a(x) \in \A$ with high probability.
% We notice that, thanks to the definition of $\mathcal{X}_\epsilon$, we can bound the number of rounds required to sample the signal $s_1$ with high probability, as stated by Lemma~\ref{lem:rounds_signal}.

\subsection{\texttt{Find-Polytopes}}
\begin{algorithm}[H]
	\caption{\texttt{Find-Polytopes}}\label{alg:find_partition}
	\begin{algorithmic}[1]
		\Require Search space $\mathcal{X}_\epsilon \subseteq \mathcal{X}$, parameter $\zeta \in (0,1)$
		\State $(\mathcal{C},\{\mathcal{X}_\epsilon(a)\}_{a \in \mathcal{C}}) \gets \texttt{Find-Fully-Dimensional-Regions}(\mathcal{X}_\epsilon,\zeta)$
		\ForAll{$a_j \in \mathcal{A} \setminus \mathcal{C}$} \label{line:find_partition_faces_loop}
		\State $\mathcal{R}_\epsilon(a_j) \gets \texttt{Find-Face}(\mathcal{C},\{\mathcal{X}_\epsilon(a)\}_{a \in \mathcal{C}},a_j)$
		\EndFor
		
		\State $\mathcal{R}_\epsilon(a_k) \gets \mathcal{X}_\epsilon(a_k) \quad \forall a_k \in \mathcal{C}$

		\State \textbf{return} $ \{\mathcal{R}_\epsilon(a) \}_{a \in \mathcal{A}}$.
	\end{algorithmic}
\end{algorithm}

Algorithm~\ref{alg:find_partition} can be divided in two parts.
First, by means of Algorithm~\ref{alg:find_non_zero_vol}, it computes the polytopes $\mathcal{R}_\epsilon(a) = \mathcal{X}_\epsilon(a)$ with volume strictly larger than zero.
This procedure is based on the algorithm developed by~\cite{bacchiocchi2024sample} for Stackelberg games.
The main differences are the usage of $\texttt{Action-Oracle}$ to query a generic normalized slice, and some technical details to account for the shape of the search space $\mathcal{X}_\epsilon$.

In the second part (loop at Line~\ref{line:find_partition_faces_loop} Algorithm~\ref{alg:find_partition}), it finds, for every polytope $\mathcal{X}_{\epsilon}(a)$ with null volume, a face $\mathcal{R}_\epsilon(a)$ such that $\mathcal{V}_\epsilon(a) \subseteq \mathcal{R}_\epsilon(a)$.
Let us remark that $\mathcal{R}_\epsilon(a)$ could be the improper face $\mathcal{X}_\epsilon(a)$ itself, and it is empty if $\mathcal{V}_\epsilon(a) = \emptyset$.
%Furthermore, as we will discuss in the following sections, Algorithm~\ref{alg:find_partition} requires polynomial CPU time and memory space at each time step $t \in [T]$.

\FindPartition*
\begin{proof}
	In the following we let $L=B +B_\epsilon +B_{\widehat{\mu}}$.
	
	As a first step, Algorithm~\ref{alg:find_partition} invokes the procedure $\texttt{Find-Fully-Dimensional-Regions}$.
	Thus, according to Lemma~\ref{lem:find_non_zero_vol}, under the event $\mathcal{E}^\textnormal{a}$, with probability at least $1-\nicefrac{\zeta}{2}$, Algorithm~\ref{alg:find_partition} computes every polytope $\mathcal{X}_\epsilon(a)$ with volume larger than zero by performing at most:
	\begin{equation*}
		C_1 = \widetilde{\mathcal{O}}\left(n^2\left(d^7L\log(\nicefrac{1}{\zeta})+\binom{d+n}{d}\right)\right).	
	\end{equation*}
	calls to the $\texttt{Action-Oracle}$ procedure.
	Together with these polytopes, it computes the set $\mathcal{C} \subseteq \mathcal{A}$ containing the actions $a$ such that $\text{vol}(\mathcal{X}_\epsilon(a))>0$.
	
	Subsequently, Algorithm~\ref{alg:find_partition} employs the procedure $\texttt{Find-Face}$ at most $n$ times and, according to Lemma~\ref{lem:find_face}, it computes the polytopes $\mathcal{R}_\epsilon(a_j)$ for every $a_j \notin \mathcal{C}$.
	Overall, this computation requires:
	\begin{equation*}
		C_2 = n^2 \binom{d+n}{d}
	\end{equation*} 
	calls to the $\texttt{Action-Oracle}$ procedure.
	Thus, under the event $\mathcal{E}^\textnormal{a}$ and with probability at least $1-\nicefrac{\zeta}{2}$, Algorithm~\ref{alg:find_partition} correctly computes the polytopes $\mathcal{R}_\epsilon(a_j)$ for every action $a_j \in \mathcal{A}$.
	Furthermore, the number of calls $C \ge 0$ to the $\texttt{Action-Oracle}$ procedure can be upper bounded as:
	\begin{equation*}
		C \coloneqq C_1 + C_2 \le \widetilde{\mathcal{O}}\left(n^2\left(d^7L\log(\nicefrac{1}{\zeta})+\binom{d+n}{d}\right)\right).
	\end{equation*}
	
	Moreover, by setting $\rho = \zeta / 2C$, and thanks to Lemma~\ref{lem:rounds_signal}, with a probability of at least $1-\rho$, every execution of $\texttt{Action-Oracle}$ requires at most $N \ge 0$ rounds, where $N$ can be bounded as follows:
	\begin{equation*}
		N \le \mathcal{O}\left(\frac{\log(1/\rho)}{\epsilon}\right)
		=\mathcal{O}\left(\frac{1}{\epsilon}\log\left(\frac{2C}{\zeta}\right)\right).
	\end{equation*}
	Consequently, since the number of calls to the $\texttt{Action-Oracle}$ procedure is equal to $C$, by employing a union bound, the probability that each one of these calls requires $N$ rounds to terminate is greater than or equal to:
	\begin{equation*}
		1-C\rho = 1-\frac{\zeta}{2C}C = 1-\frac{\zeta}{2}.
	\end{equation*}
	To conclude the proof, we employ an union bound over the probability that every execution of $\texttt{Action-Oracle}$ terminates in at most $N$ rounds, and the probability that Algorithm~\ref{alg:find_partition} performs $C$ calls to the $\texttt{Action-Oracle}$ procedure. 
	Since the probability that each one of these two events hold is at least $1-\nicefrac{\zeta}{2}$, with probability at least $1-\zeta$, Algorithm~\ref{alg:find_partition} correctly terminates by using a number of samples of the order:
	\begin{equation*}
		\widetilde{\mathcal{O}}\left(\frac{C}{\epsilon}\log\left(\frac{2C}{\zeta}\right)\right) \le
		%\widetilde{\mathcal{O}}\left(\frac{C}{\epsilon} \log\left(\frac{1}{\zeta}\right)\right) =
		\widetilde{\mathcal{O}}\left( \frac{n^2}{\epsilon} \log^2\left(\frac{1}{\zeta}\right) \left(d^7L +\binom{d+n}{d}\right) \right),
	\end{equation*}
	concluding the proof.
\end{proof}

\subsection{\texttt{Find-Fully-Dimensional-Regions}}
\label{sec:find_non_zero_vol}
% \bollo{Ha senso che l'algoritmo sia descritto parlando di $a$, ma nello pseudocodice usiamo $a_j$ (per la notazione dell'iperpiano, che e è $H_{jk}$)}
\begin{algorithm}[H]
	\caption{\texttt{Find-Fully-Dimensional-Regions}}\label{alg:find_non_zero_vol}
	\begin{algorithmic}[1]
		\Require Search space $\mathcal{X}_\epsilon \subseteq \mathcal{X}$, parameter $\delta \in (0,1)$
		\State $\delta \gets \nicefrac{\zeta}{2n^2(2(d+n)+n)}$\label{line:delta}
		\State $\mathcal{C} \gets \varnothing$ %\hfill \LComment{\textnormal{Set of \emph{closed} follower's actions}}
		\While{$\bigcup_{a_j \in \mathcal{C}} \mathcal{U}(a_j) \ne \mathcal{X}_\epsilon$ }\label{line:partition_loop1}
		\State $x^{\text{int}} \gets $ Sample a point from $\text{int} \big(\mathcal{X}_\epsilon \setminus \bigcup_{a_k \in \mathcal{C}} \mathcal{U}(a_k) \big) $ \label{line:sample_point}
		\State $a_j \gets \texttt{Action-Oracle}(x^\text{int})$
		\State $\mathcal{U}(a_j) \gets \mathcal{X}_\epsilon$  %\hfill \LComment{\textnormal{Initialize upper bound of $\mathcal{P}_j$}} 
		% \State $p^\text{int} \gets p$ \hfill \LComment{\textnormal{Interior point of $\mathcal{P}_j$}}
%		\State $\mathcal{V} \gets V(\mathcal{U}(a_j))$ %\hfill \LComment{\textnormal{Set of unchecked vertices of $\mathcal{U}_j$}}\alglinelabel{line:set_vertexes}
		\State $B_{x} \gets$ Bit-complexity of $x^{\text{int}}$
		\State $\lambda \gets d 2^{-d(B_{x}+4(B +B_\epsilon +B_{\widehat{\mu}}))-1}$  %\hfill \LComment{\textnormal{See Lemma~\ref{lem:vertex_two}}}
%		\While{$\mathcal{V} \neq \varnothing$}
		\ForAll{$v \in V(\mathcal{U}(a_j))$}\label{line:partition_loop3}
%		\State $v \gets $ Take any vertex in $\mathcal{V}$ \label{line:take_vertex}
		\State $x \gets \lambda x^\text{int} + (1-\lambda) v$ \label{line:convex_comb}
		\State $a \leftarrow \texttt{Action-Oracle}(x)$\label{line:d_j}
		\If{$a \neq a_j$}
		\State $ H_{jk} \gets\texttt{Find-Hyperplane} (a_j,\mathcal{U}(a_j), x^\text{int}, v, \delta)$ \label{line:compute_hyperpane}
		\State $\mathcal{U}(a_j) \leftarrow \mathcal{U}(a_j) \cap \mathcal{H}_{jk}$  %\hfill \LComment{\textnormal{Update upper bound}}
%		\State $\mathcal{V} \gets V(\mathcal{U}(a_j))$ %\hfill \LComment{\textnormal{Update unchecked vertices}} \alglinelabel{line:re_init_vertexes}
		\Else
%		\State $\mathcal{V} \gets \mathcal{V} \setminus \{v\}$ %\hfill \LComment{\textnormal{Vertex $v$ has been checked}}
		\State restart the for-loop at Line~\ref{line:partition_loop3}
		\EndIf
		\EndFor
		\State $\mathcal{C} \gets \mathcal{C}  \cup \{a_j\}$ %\hfill \LComment{\textnormal{Action $a_j$ has been closed}}
		%		\State \bac{ritorna l'unione dei verici.} 
		\State $\mathcal{X}_\epsilon(a_j) \gets \mathcal{U}(a_j)$
		\EndWhile
%		\State $\mathcal{V} \gets \bigcup_{a_j \in \mathcal{C}} V(\mathcal{U}(a_j) )$
		\State \textbf{return} $\mathcal{C},\{\mathcal{X}_\epsilon(a_j)\}_{a_j \in \mathcal{C}}$
%		\State \textbf{return} $ \mathcal{V}, (\texttt{Action-Oracle}(x))_{x \in \mathcal{V}}$\label{line:query_vertices}
		%\State   $p^\star \gets \argmax_{p \in  \bigcup_{a_j \in \mathcal{C}} V(\mathcal{U}_{j} )}  u_\ell(p)$\label{line:optimal}
	\end{algorithmic}
\end{algorithm}

At a high level, Algorithm~\ref{alg:find_non_zero_vol} works by keeping track of a set $\mathcal{C} \subseteq \A$ of closed actions, meaning that the corresponding polytope $\mathcal{X}_\epsilon(a)$ has been completely identified.
%The algorithm closes one action after the other. 
First, at Line~\ref{line:sample_point}, Algorithm~\ref{alg:find_non_zero_vol} samples at random a normalized slice $x^\textnormal{int}$ from the interior of one of the polytopes $\mathcal{X}_\epsilon(a_j)$ that have not yet been closed and queries it, observing the best-response $a_j \in \A$. 
Then, it initializes the entire $\mathcal{X}_\epsilon$ as the upper bound $\mathcal{U}(a_j)$ of the region $\mathcal{X}_\epsilon(a_j)$. 
As a further step, to verify whether the upper bound $\mathcal{U}(a_j) $ coincides with $\mathcal{X}_\epsilon(a_j)$, Algorithm~\ref{alg:find_non_zero_vol} queries at Line~\ref{line:d_j} one of the vertices of the upper bound $\mathcal{U}(a_j)$. 
Since the same vertex may belong to the intersection of multiple regions, the $\texttt{Action-Oracle}$ procedure is called upon an opportune convex combination of $x^\textnormal{int}$ and the vertex $v$ itself (Line~\ref{line:convex_comb}). 
If the vertex does not belong to $\mathcal{X}_\epsilon(a_j)$, then a new separating hyperplane can be computed (Line~\ref{line:compute_hyperpane}) and the upper bound $\mathcal{U}(a_j)$ is updated accordingly. 
%Then, the upper bound $\mathcal{U}(a)$ is refined by finding new separating hyperplanes until it coincides with the polytope $\mathcal{X}_\epsilon(a) $. 
In this way, the upper bound $\mathcal{U}(a_j)$ is refined by finding new separating hyperplanes until it coincides with the polytope $\mathcal{X}_\epsilon(a_j) $. 
Finally, such a procedure is iterated for all the receiver's actions $a_j \in \mathcal{A}$ such that $\textnormal{vol}(\mathcal{X}_\epsilon(a_j) )>0$, ensuring that all actions corresponding to polytopes with volume larger than zero are \emph{closed}.

We observe that the estimator $\widehat{\mu}_t$ is updated during the execution of Algorithm~\ref{alg:find_non_zero_vol} according to the observed states.
However, let us remark that the search space $\mathcal{X}_\epsilon = \left\{ x \in \X \mid \sum_{\theta  \in \widetilde{\Theta}} \widehat{\mu}_{\theta}x_\theta \ge 2\epsilon \right\}$ does \emph{not} change during the execution of this procedure.

%To execute Algorithm~\ref{alg:find_non_zero_vol} with polynomial per-round running time and memory space, one has to implement the loop at Line~\ref{line:partition_loop3} efficiently.
%It is not possible to compute the whole set of vertices $\text{V}(\mathcal{U}_\epsilon(a_j))$ in a single step in polynomial time, as its size is bounded by $\binom{d+n}{d}$.
%Thus, at each iteration, it is necessary to compute a new vertex and then discard the old one.
%As a result, this loop can be performed by iterating over the possible subsets of $d$ hyperplanes between the ones (at most $d+n$) defining $\text{V}(\mathcal{U}_\epsilon(a_j))$.
%At each iteration, the vertex $v$ is the intersection of the currently selected $d$ hyperplanes (if the intersection is empty, skip the iteration).
%We observe that this implementation may iterate over the same vertex multiple times, but this does not alter the result of the computation or the asymptotic number of samples required by Algorithm~\ref{alg:find_non_zero_vol}.

%\bollo{Non so se mettergli un lemma: sarebbe una one-line proof}
\begin{restatable}{lemma}{findNonZeroVol}
	\label{lem:find_non_zero_vol}
	Given in input $\mathcal{X}_\epsilon \subseteq \mathcal{X}$ and $\zeta \in (0,1)$, then under the event $\mathcal{E}^{\textnormal{a}}$ with probability at least $1 -\nicefrac{\zeta}{2}$ Algorithm~\ref{alg:find_non_zero_vol} computes the collection of polytopes $\{\mathcal{X}_\epsilon(a_j)\}_{a_j \in \mathcal{C}}$ with volume larger than zero, and the corresponding set of actions $\mathcal{C}$.
	Furthermore, it employs at most:
	\begin{equation*}
		\widetilde{\mathcal{O}}\left(n^2\left(d^7L\log(\nicefrac{1}{\zeta})+\binom{d+n}{d}\right)\right)
	\end{equation*}
	calls to the $\textnormal{\texttt{Action-Oracle}}$ procedure.
\end{restatable}
\begin{proof}
	Thanks to Lemma~\ref{lem:find_hp} and Lemma~\ref{lem:sample_point_first}, with an approach similar to the one proposed in Theorem~4.3 by~\citet{bacchiocchi2024sample}, we can prove that, under the event $\mathcal{E}^\textnormal{a}$, with probability at least $1-\delta n^2(2(d+n)^2 +n)$, Algorithm~\ref{alg:find_non_zero_vol} computes every polytope $\mathcal{X}_\epsilon(a)$ with volume larger than zero by performing at most:
	\begin{equation*}
		\mathcal{O}\left(n^2\left(d^7L\log(\nicefrac{1}{\delta}) +\binom{d+n}{d}\right)\right)
	\end{equation*}
	calls to the $\texttt{Action-Oracle}$ procedure.
	Together with these polytopes, it computes the set $\mathcal{C} \subseteq \mathcal{A}$ containing the actions $a$ such that $\text{vol}(\mathcal{X}_\epsilon(a))>0$.
	
	Furthermore, we observe that $\zeta = 2 \delta n^2(2(d+n)+n)$, as defined at Line~\ref{line:delta} in Algorithm~\ref{alg:find_non_zero_vol}.
	As a result, under the event $\mathcal{E}^\textnormal{a}$, with probability at least $1-\nicefrac{\zeta}{2}$, the number of calls $C_1 \ge 0$ performed by Algorithm~\ref{alg:find_non_zero_vol} to the $\texttt{Action-Oracle}$ procedure can be bounded as follows: 
	\begin{equation*}
		C_1 \le \widetilde{\mathcal{O}}\left(n^2\left(d^7L\log(\nicefrac{1}{\zeta})+\binom{d+n}{d}\right)\right),
	\end{equation*}
	concluding the proof.
\end{proof}

\subsection{\texttt{Find-Hyperplane}}
\begin{algorithm}[H]
	\caption{\texttt{Find-Hyperplane}}
	\label{alg:hyperplane}
	\begin{algorithmic}[1]
		\Require $a_j, \mathcal{U}(a_j), x^\text{int}, v, \delta$ %\hfill \LComment{\textnormal{As given by Algorithm~\ref{alg:learning_commitment}}}
		\State $x \gets \texttt{Sample-Int}(\mathcal{U}(a_j), \delta)$
		\State $x^1 \gets x$ %\hfill \LComment{\textnormal{$x^1$ always selected at random}}
		\If{ $\texttt{Action-Oracle} (x)=a_j$ }
		\State $x^2 \gets v $ %\hfill \LComment{\textnormal{$v \notin \mathcal{C}_{a_j}$ by design}}
		\Else
		\State $x^2 \gets x^\text{int}$ %\hfill \LComment{\textnormal{$\texttt{Oracle} (x)\neq a_j$}}
		\EndIf
		\State $x^\circ \gets \texttt{Binary-Search}(a_j,x^1, x^2)$ \label{line:hyperplane_binary_1} %\alglinelabel{line:binary_search_hp_1}
		\State $\alpha \gets 2^{-4d(B_x +B +B_{\widehat{\mu}} +B_\epsilon)} /d$ \label{line:alpha_def} %\hfill \LComment{\textnormal{$B = $ bit-complexity of $x^\circ$}}
		%\bollo{$\alpha \gets 2^{-m(B+4L)-1} / m$}
		\State $\mathcal{S}_j \gets \varnothing$; $\mathcal{S}_k \gets \varnothing$
		\For{$i = 1 \ldots d$} %\alglinelabel{alg:hyperplane_for}
		\State $ x \gets \texttt{Sample-Int}(H_i \cap \mathcal{X}, \delta)$ %\alglinelabel{line:sample_from_facet} \hfill \LComment{\textnormal{Sample a facet}}
		%\State $ p^k \gets (1-\alpha) p^* + \alpha p$ \label{line:point_simlex}
		\State ${x}^{+i} \gets x^\circ + \alpha (x-x^\circ)$ \label{line:point_simlex}	
		\State $x^{-i} \gets x^\circ - \alpha (x-x^\circ)$
		\If{ $\texttt{Action-Oracle} ({x}^{+i})=a_j$ }
		\State $\mathcal{S}_{j} \gets \mathcal{S}_{j} \cup \{x^{+i} \}  \wedge  \mathcal{S}_{k} \gets \mathcal{S}_{k} \cup \{ x^{-i} \} $
		\Else
		\State $\mathcal{S}_{k} \gets \mathcal{S}_{k} \cup \{x^{+i} \}  \wedge  \mathcal{S}_{j} \gets \mathcal{S}_{j} \cup \{ x^{-i} \} $
		\EndIf
		\EndFor
		\State Build $H_{jk}$ by $\texttt{Binary-Search}(a_j,x^1, x^{2})$ for $d-1$ pairs of linearly-independent points $x^1 \in \mathcal{S}_j, x^2 \in \mathcal{S}_k$
		%		\State $p^1 \gets $ take any point in $\mathcal{S}_{k}$
		%		\FOR{$a_i \in \mathcal{A}_\ell$}
		%		\State $p^2 \gets $ point $p^{+i}$ or $p^{-i}$ belonging to $\mathcal{S}_j$
		%		\State $p^{jk i} \gets \texttt{Binary-Search}(p^1, p^{2})$\label{line:binary_search_hp_2}
		%		\ENDFOR
		%		\State \textbf{return} $H_{jk}$ built with the $m$ points $p^{jki}$
	\end{algorithmic}
\end{algorithm}

%The goal of $\texttt{Find-HS}$ procedure is to compute a new separating hyperplane \bac{between what?}.
The goal of the $\texttt{Find-Hyperplane}$ procedure is to compute a new separating hyperplane $H_{jk}$ between a given region $\mathcal{X}_\epsilon(a_j)$ and some other polytope $\mathcal{X}_\epsilon(a_k)$, with $a_j, a_k \in \A$.
To do so, it receives as input an upper bound $\mathcal{U}(a_j)$ of some polytope $\mathcal{X}_\epsilon(a_j)$, an interior point $x^\textnormal{int} \in \text{int}(\mathcal{U}(a_j))$, a vertex $v \in \text{V}(\mathcal{U}(a_j))$ that does not belong to $\mathcal{X}_\epsilon(a_j)$, and a parameter $\delta>0$ as required by the $\texttt{Sample-Int}$ procedure.
As a first step, Algorithm~\ref{alg:hyperplane} samples at random a slice $x$ from the interior of the upper bound $\mathcal{U}(a_j)$.
Subsequently, it performs a binary search on the segment between $x$ and either $v$ or $x^\textnormal{int}$, depending on the best response $a(x)$ in $x$.
%The pseudocode for the binary-search procedure is provided in Algorithm~\ref{alg:binary_search}.
This binary search returns a point $x^\circ$ on some new separating hyperplane $H_{jk}$ (Line~\ref{line:hyperplane_binary_1}).
As a further step, the algorithm computes two sets of normalized slices, $\mathcal{S}_j \subseteq \mathcal{X}_\epsilon(a_j)$ and $\mathcal{S}_k \subseteq \mathcal{X}_\epsilon(a_k)$.
Finally, it performs $d-1$ binary searches between different couples of points, one in $\mathcal{S}_j$ and the other in $\mathcal{S}_k$, in order to completely identify the separating hyperplane.

\begin{restatable}{lemma}{FindHP}\label{lem:find_hp} 
	With probability at least $1-(d+n)^2\delta$, under the event $\mathcal{E}^\textnormal{a}$ Algorithm~\ref{alg:hyperplane} returns a separating hyperplane $H_{jk}$ by using $\mathcal{O}(d^7(B +B_\epsilon +B_{\widehat{\mu}}) +d^4\log(\nicefrac{1}{\delta}))$ calls to Algorithm~\ref{alg:action_oracle}. 
\end{restatable}

\begin{proof}
	We observe that, with the same analysis provided in Lemma~4.7 by~\citet{bacchiocchi2024sample}, we can prove that, under the event $\mathcal{E}^\textnormal{a}$, the binary-search procedure described in Algorithm~\ref{alg:binary_search} correctly computes a point on a separating hyperplane by calling the $\texttt{Action-Oracle}$ procedure at most $\mathcal{O}(d(B_x +B))$ times.
	
	 %\bac{Qui non capisco... stai dicendo che devi garantire una proprieta sugli $x^{i}$??} \bollo{sì} Furthermore, as long as every $x^{+i}$ belongs to either $\mathcal{P}'(a_j) \cap \mathcal{X}_\epsilon$ or $\mathcal{P}'(a_k) \cap \mathcal{X}_\epsilon$, under the event $\mathcal{E}^\textnormal{a}$, Algorithm~\ref{alg:hyperplane} correctly computes a new separating hyperplane with probability at least $1-(d+n)^2\delta$. This can be proved with the same reasoning applied in Lemma 4.6 by \bollo{Stackelberg}.
	
%	\bollo{Forse così si capisce meglio}
	With the same reasoning applied in Lemma~4.4 and~4.5 by~\citet{bacchiocchi2024sample}, we can prove that with probability at least $1-(d+n)^2\delta$, under the event $\mathcal{E}^\textnormal{a}$, the points $x^{+i}$ are linearly independent and do not belong to $H_{jk}$. To conclude the proof, we have to show that every $x^{+i}$ belongs to either $\mathcal{X}_\epsilon(a_j)$ or $\mathcal{X}_\epsilon(a_k)$. This is because, if the previous condition holds, under the event $\mathcal{E}^\textnormal{a}$, Algorithm~\ref{alg:hyperplane} correctly computes a new separating hyperplane with probability at least $1-(d+n)^2\delta$.
	% Furthermore, if every $x^{+i}$ belongs to either $\mathcal{P}'(a_j) \cap \mathcal{X}_\epsilon$ or $\mathcal{P}'(a_k) \cap \mathcal{X}_\epsilon$, under the event $\mathcal{E}^\textnormal{a}$, Algorithm~\ref{alg:hyperplane} correctly computes a new separating hyperplane with probability at least $1-(d+n)^2\delta$.
	
	To do that, we show that the constant $\alpha$ defined at Line~\ref{line:alpha_def} in Algorithm~\ref{alg:hyperplane} is such that all the points $x^{+i}$ and $x^{-i}$ either belong to $\mathcal{X}_\epsilon(a_j)$ or $\mathcal{X}_\epsilon(a_k)$, given that $x^\circ$ belongs to the hyperplane between these polytopes.
	With an argument similar to the one proposed in~\citet{bacchiocchi2024sample}, the distance between $x^i$ and any separating hyperplane can be lower bounded by ${2^{-d(B_x +4B)}}$, where $B_x$ is the bit-complexity of $x^\circ$.
	
	Similarly, the distance between $x^\circ$ and the hyperplane $\widehat{H}=\{x \in \mathbb{R}^d \mid \sum_{\theta \in \widetilde{\Theta}} \widehat{\mu}_\theta x_\theta \geq 2\epsilon\}$ can be lower bounded as follows:
	\begin{equation}\label{eq:d_x_diam_h_hat}
		d(x^\circ,\widehat{H}) = \frac{\left| \sum_{\theta  \in \widetilde{\Theta}} x_\theta \widehat{\mu}_\theta +2\epsilon \right|}{\sqrt{\sum_{\theta  \in \widetilde{\Theta}} \widehat{\mu}_\theta^2}}
		\geq \frac{1}{d2^{3d(B_x +B_{\widehat{\mu}} +B_\epsilon)}},
	\end{equation}
	where $B_{\widehat{\mu}}$ is the bit-complexity of $\widehat{\mu}$. 
	The inequality follows by observing that the denominator of the fraction above is at most $d$, while to lower bound the numerator we define the following quantities:
	\begin{equation*}
		\sum_{\theta  \in \widetilde{\Theta}} x_\theta \widehat{\mu}_\theta = \frac{\alpha}{\beta}
		\quad \text{and} \quad \epsilon = \frac{\gamma}{\nu},
	\end{equation*}
	where $\alpha$ and $\beta$ are integers numbers, while $\gamma$ and $\nu$ are natural numbers.
%	Thus, the nominator \bac{numerator? of the fraction defined in Eq?} can be lower bounded as follows:
	Thus, the numerator $\left| \sum_{\theta  \in \widetilde{\Theta}} x_\theta \widehat{\mu}_\theta +2\epsilon \right|$ of the fraction defined in Equation~\ref{eq:d_x_diam_h_hat} can be lower bounded as follows:
	\begin{equation*}
		\left| \sum_{\theta  \in \widetilde{\Theta}} x_\theta \widehat{\mu}_\theta +2\epsilon \right| = \left| \frac{\alpha \nu + 2 \beta \gamma}{\beta \nu} \right| \ge \left| \frac{1}{\beta \nu} \right| \ge 2^{-3d(B_x +B_{\widehat{\mu}} +B_\epsilon)},
	\end{equation*}
	where the last inequality follows from the fact that the bit-complexity of $\nu$ is at most $B_\epsilon$, while the bit-complexity of $\beta$ cannot exceed $3d(B_x +B_{\widehat{\mu}})$ as stated by Lemma~D.1 of~\citet{bacchiocchi2024sample}.
	
	Overall, the distance between $x^\circ$ and the boundary of the polytope $\mathcal{X}_\epsilon(a_j) \cap \mathcal{X}_\epsilon(a_k)$ is strictly larger than $\alpha \coloneqq 2^{-4d(B_x +B +B_{\widehat{\mu}} +B_\epsilon) -\log_2(d) }$.
	%	Thus, with $\alpha = 2^{-4d(B_x +B +B_{\widehat{\mu}} +B_\epsilon)-\log_2(d)}$, every signaling scheme $x^{+i}$ and $x^{-i}$ belongs to $\mathcal{X}_\epsilon$ and either $\mathcal{P}'(a_j)$ or $\mathcal{P}'(a_k)$.
	Thus, every signaling scheme $x^{+i}$ and $x^{-i}$ belongs to $\mathcal{X}_\epsilon$ and either $\mathcal{X}(a_j)$ or $\mathcal{X}(a_k)$.
	
	Finally, we observe that the bit-complexity of $x$ is bounded by $B_x = \mathcal{O}(d^3(B +B_\epsilon +B_{\widehat{\mu}}) +\log(\nicefrac{1}{\delta}))$, as stated by Lemma~\ref{lem:sample_point_first}. 
	Thus, the first binary-search requires $\mathcal{O}(d(B_x +B)) = \mathcal{O}(d^4L +d\log(\nicefrac{1}{\delta}))$ calls to $\texttt{Action-Oracle}$, where $L=B +B_\epsilon +B_{\widehat{\mu}}$.
	
	Furthermore, the bit-complexity of $x^\circ$ is bounded by $\mathcal{O}(d^4L +d\log(\nicefrac{1}{\delta}))$.
	As a result, the bit-complexity of the slices $x^{+i}$ and $x^{-i}$ is bounded by $\mathcal{O}(d^5L +d\log(\nicefrac{1}{\delta}))$, given that the bit-complexity of $\alpha$ is $\mathcal{O}(dL)$.
	It follows that each binary search between two points in $\mathcal{S}_j$ and $\mathcal{S}_k$ requires $\mathcal{O}(d^6L +d^2\log(\nicefrac{1}{\delta}))$ calls to $\texttt{Action-Oracle}$.
	
	Overall, Algorithm~\ref{alg:hyperplane} invokes the $\texttt{Action-Oracle}$ procedure at most $\mathcal{O}(d^7L +d^4\log(\nicefrac{1}{\delta}))$ times, accounting for the $d-1$ binary searches, concluding the proof.
\end{proof}

\subsection{\texttt{Binary-Search}}
The $\texttt{Binary-Search}$ procedure performs a binary search on the segment connecting two points, $x^1, x^2 \in \mathcal{X}_\epsilon$ such that $x^1 \in \mathcal{X}_\epsilon(a_j)$ and $x^2 \notin \mathcal{X}_\epsilon(a_j)$ for some $a_j \in \mathcal{A}$, in order to find a point $x^\circ$ on some separating hyperplane $H_{jk}$.
At each iteration, the binary search queries the middle point of the segment.
Depending on the receiver's best-response in such a point, it keeps one of the two halves of the segment for the subsequent iteration.
The binary search ends when the segment is sufficiently small, so that it contains a single point with a bit-complexity appropriate for a point that lies on both the hyperplane $H_{jk}$ and the segment connecting $x^1$ and $x^2$.
Such a point can be found traversing the Stern-Brocot-Tree.
For an efficient implementation, see~\citet{Fori2007}. 
Overall, Algorithm~\ref{alg:binary_search} performs $\mathcal{O}(d(B_x+B))$ calls to the $\texttt{Action-Oracle}$ procedure, and returns a normalized slice $x^\circ$ with bit-complexity bounded by $\mathcal{O}(d(B_x+B))$, where $B_x$ is the bit-complexity of the points~$x^1$~and~$x^2$.
%\bac{Lemma? Descrizione?}
%\bollo{Non, so se mettere un lemma, la proof sarebbe esattamente come quella sul paper degli Stackelberg} \bac{okay}
\begin{algorithm}[H]
	\caption{\texttt{Binary-Search}}\label{alg:binary_search}
	\begin{algorithmic}[1]
		\Require $a_j$, $x^1, x^2 $ of bit-complexity bounded by some $B_x>0$ %\hfill \LComment{\textnormal{From Algorithm~\ref{alg:hyperplane}}}
		\State $\lambda_1 \gets 0 $; $\lambda_2 \gets 1 $
		\While{$|\lambda_2 - \lambda_1 | \ge 2^{-6d(5B_x+8B)}$}
		\State $\lambda \gets {(\lambda_1 + \lambda_2)}/{2} $; $x^\circ \gets x^1 + \lambda (x^2 - x^1)$
		\If{ $\texttt{Action-Oracle}(x^\circ) = a_j$ }
		\State $\lambda_1 \gets \lambda $ %\hfill \LComment{\textnormal{New point with $a_j$ as best response}}
		\Else
		\State $\lambda_2 \gets \lambda $ %\hfill \LComment{\textnormal{New point with best response $\neq a_j$}}
		\EndIf
		\EndWhile
		\State $\lambda \gets \texttt{Stern-Brocot-Tree}(\lambda_1, \lambda_2,3d(5B_x +8B))$
		\State $x^\circ \gets \lambda x^1 + (1- \lambda) x^2$
	\end{algorithmic}
\end{algorithm}

\subsection{\texttt{Sample-Int}}
\begin{algorithm}[H]
	\caption{\texttt{Sample-Int}}
	\label{alg:sample_int}
	\begin{algorithmic}[1]
		\Require $\mathcal{P} \subseteq \mathcal{X}_\epsilon: \textnormal{vol}_{d-1}(\mathcal{P})>0$, and $\delta$ %\bac{perche d'}
		\State $\mathcal{V} \gets d$ linearly-independent vertexes of $\mathcal{P}$\label{line:sample_li_vertices}
		\State $x^\diamond \gets \frac{1}{d} \sum_{v \in \mathcal{V}} v$\label{line:sample_1}
		\State $\rho \gets \left(d^3 2^{9d^3L +4dL} \right)^{-1}$; $M \gets \left \lceil \nicefrac{\sqrt{d}}{\delta} \right \rceil$\label{line:sample_rho}
		\State $y \sim \text{Uniform}(\{- 1, -\frac{M-1}{M}, \ldots, 0, \ldots, \frac{M-1}{M}, 1 \}^{d-1})$\label{line:sample_2}
		
		\ForAll{$i \in 1,\dots,d-1$}
		\State $x_i \gets x^\diamond_i +\rho y_i$ 
		\EndFor
		\State $x_d \gets  1-\sum_{i= 1}^{d-1} x_i$
	\end{algorithmic}
\end{algorithm}

The $\texttt{Sample-Int}$ procedure (Algorithm~\ref{alg:sample_int}) samples at random a normalized slice from the interior of a given polytope $\mathcal{P}$. 
We observe that each polytope Algorithm~\ref{alg:find_non_zero_vol} is required to sample from is defined as the intersection of $\mathcal{X}_\epsilon$ with some separating half-spaces as the ones defined in Section~\ref{sec:slices}. 
This procedure provides theoretical guarantees both on the bit-complexity of the point $x$ being sampled and on the probability that such a point belongs to a given hyperplane.
Furthermore, it can be easily modified to sample a point from a facet of the simplex $\mathcal{X} = \Delta_d$ (intuitively, this is equivalent to sample a point from $\Delta_{d-1}$).
As a first step, Algorithm~\ref{alg:sample_int} computes a normalized slice $x^\diamond$ in the interior of $\mathcal{P}$ (Line~\ref{line:sample_1}).
Subsequently, it samples randomly a vector $y$ from a suitable grid belonging to the $(d-1)$-dimensional hypercube with edges of length $2$.
As a further step, it sums each component $x^\diamond_i$ of the normalized slice $x^\diamond$ with the corresponding component $y_i$ of $y$, scaled by an opportune factor $\rho$, where the constant $\rho$ is defined to ensure that $x$ belongs to the interior of $\mathcal{P}$.
%The last component is set to guarantee that the resulting signaling scheme $x$ belongs to $\Delta_d$.

The theoretical guarantees provided by Algorithm~\ref{alg:sample_int} are formalized in the following lemma:
\begin{restatable}{lemma}{samplepointfirst}
	\label{lem:sample_point_first}
	Given a polytope $\mathcal{P} \subseteq \mathcal{X}_\epsilon:\textnormal{vol}_{d-1}(\mathcal{P}) >~0$ defined by separating or boundary hyperplanes, Algorithm~\ref{alg:sample_int} computes $x \in \textnormal{int}(\mathcal{P})$ such that, for every linear space $H \subset \mathbb{R}^d : \mathcal{P} \not\subseteq H$ of dimension at most $d-1$, the probability that $x \in H$ is at most $\delta$.
	Furthermore, the bit-complexity of $x$ is $\mathcal{O}(d^3(B +B_\epsilon +B_{\widehat{\mu}})+\log(\nicefrac{1}{\delta}))$.
\end{restatable}
\begin{proof}	
	In the following, we prove that $x$ belongs to the interior of the polytope $\mathcal{P}$. To do that, we observe that the point $x^\diamond$ belongs to the interior of $\mathcal{P}$, while, with the same analysis proposed in Lemma 4.8 by~\citet{bacchiocchi2024sample}, the distance between $x^\diamond$ and $x$ can be upper bounded by $\rho n$. %$d(x^\diamond,x) \le \rho n$.
	To ensure that $x $ belongs to $ \textnormal{int}(\mathcal{P})$, we have to show that $d(x^\diamond,x)$ is smaller than the distance between $x^\diamond$ and any hyperplane defining the boundary of $\mathcal{P}$.
	
	Let us denote with $v^h$ the $h$-vertex of the set $\mathcal{V}$, so that $\mathcal{V} = \{v^1,v^2,\dots,v^d\}$.
	Furthermore, we define: \[	v^h_\theta \coloneqq \frac{\gamma^h_\theta}{\nu_h}\]
	for each $h \in [d]$ and $\theta \in \Theta$.
	Since $x^\diamond$ belongs to $\textnormal{int}(\mathcal{P})$, then the distance between $x^\diamond$ and any separating hyperplane $H_{jk}$ can be lower bounded as follows:
%	\bac{qui mancano di quadratii?}
	\begin{align*}
		d(x^\diamond,H_{jk}) &= \left| \frac{\sum_{\theta  \in \Theta} x^\diamond_\theta \mu_\theta (u_\theta(a_j)-u_\theta(a_k))}{\sqrt{\sum_{\theta  \in \Theta}\mu^2_\theta (u_\theta(a_j)-u_\theta(a_k))^2}} \right| \\
		&= \left| \frac{\sum_{\theta  \in \Theta} \sum_{h=1}^{d} v^h_\theta \mu_\theta (u_\theta(a_j)-u_\theta(a_k))}{d\sqrt{\sum_{\theta  \in \Theta}\mu^2_\theta (u_\theta(a_j)-u_\theta(a_k))^2}} \right| \\
		&\geq \frac{1}{d^2} 2^{-4dB -9d^3(B +B_\epsilon +B_{\widehat{\mu}})}.
	\end{align*}
	To prove the last inequality, we observe that the denominator of the fraction above can be upper bounded by $d^2$. To lower bound the nominator, we define:
	\begin{equation*}
		\mu_\theta (u_\theta(a_j)-u_\theta(a_k)) \coloneqq \frac{\alpha_\theta}{\beta_\theta}
	\end{equation*}
	for each $\theta \in \Theta$. As a result, we have:
	\begin{align*}
		\left| \sum_{\theta  \in \Theta} \sum_{h=1}^{d} v^h_\theta \mu_\theta (u_\theta(a_j)-u_\theta(a_k)) \right| &= \left| \sum_{\theta  \in \Theta} \sum_{h=1}^{d} \frac{\alpha_\theta \gamma^h_\theta}{\beta_\theta \nu_h} \right| \\
		&= \left| \frac{\sum_{\theta  \in \Theta} \sum_{h=1}^{d} \alpha_\theta \gamma^h_\theta \left( \prod_{\theta' \neq \theta}\beta_{\theta'} \prod_{h' \neq h} \nu_{h'} \right) }{\prod_{\theta \in \Theta}\beta_\theta \prod_{h=1}^d \nu_h} \right| \\
		&\geq \left( \prod_{\theta \in \Theta}\beta_\theta \prod_{h=1}^d \nu_h \right)^{-1} \\
		&\geq 2^{-4dB -9d^3(B +B_\epsilon +B_{\widehat{\mu}})}.
	\end{align*}
	The first inequality holds because the numerator of the fraction above can be lower bounded by one.
	The denominator can be instead upper bounded observing that the bit-complexity of each $\beta_\theta$ is at most $4B$ while the bit-complexity of each $\nu_h$ is at most $9d^2(B +B_\epsilon +B_{\widehat{\mu}})$, as stated by Lemma~\ref{lem:vertex_bits}.
	
	In a similar way, we can lower bound the distance between $x^\diamond$ and $ \widehat{H} \hspace{-0.5mm}\coloneqq \hspace{-0.5mm} \{x \in \mathbb{R}^d \hspace{-0.5mm} \mid  \hspace{-0.5mm} \sum_{\theta  \in \widetilde{\Theta}} \widehat{\mu}_\theta x_\theta \hspace{-0.5mm} \geq \hspace{-0.5mm} 2\epsilon\}$.
	\begin{align*}
		d(x^\diamond,\widehat{H}) &= \left| \frac{\sum_{\theta  \in \widetilde{\Theta}} \widehat{\mu}_\theta x^\diamond_\theta +2\epsilon}{\sqrt{\sum_{\theta  \in \widetilde{\Theta}}\widehat{\mu}^2}} \right| \\
		&= \left| \frac{\sum_{\theta  \in \widetilde{\Theta}} \sum_{h=1}^d \widehat{\mu}_\theta v^h_\theta +2\epsilon}{d\sqrt{\sum_{\theta  \in \widetilde{\Theta}}\hat{\mu}^2}} \right| \\
		&\geq \frac{1}{d^2} 2^{-B_{\widehat{\mu}} -B_\epsilon -9d^3(B +B_\epsilon +B_{\widehat{\mu}})}.
	\end{align*}
	To prove the last inequality, we observe that the denominator of the fraction above can be upper bounded by $d^2$, while to lower bound the numerator, we define:
	\begin{equation*}
		\hat{\mu}_\theta \coloneqq \frac{N_\theta}{p} \text{ and } \epsilon = \frac{\alpha}{\beta}
	\end{equation*}
	for each $\theta \in \widetilde{\Theta}$. As a result, we have:
	\begin{align*}
		\left| \sum_{\theta  \in \widetilde{\Theta}} \sum_{h=1}^d \widehat{\mu}_\theta v^h_\theta +2\epsilon \right| &= \left| \sum_{\theta  \in \widetilde{\Theta}} \sum_{h=1}^d \frac{N_\theta \gamma^h_\theta}{p \nu_h} +2\epsilon \right| \\
		&= \left| \frac{\sum_{\theta \in \widetilde{\Theta}}\sum_{h=1}^d \beta N_\theta \gamma^h_\theta \prod_{h'\neq h} \nu_{h'} +2\alpha p \prod_{h=1}^{d}\nu_h} {p \prod_{h=1}^{d} \nu_h \beta} \right| \\
		&\ge 2^{-B_{\widehat{\mu}} -B_\epsilon -9d^3(B +B_\epsilon +B_{\widehat{\mu}})}
	\end{align*}
	
	Thus, the distance between $x^\diamond$ and any hyperplane $H$ defining the boundary of $\mathcal{P}$ can be lower bounded as follows:
	\begin{equation*}
		d(x^\diamond,H) \ge \frac{1}{d^2} 2^{-9d^3(B +B_\epsilon +B_{\widehat{\mu}}) -4dB -B_{\widehat{\mu}} -B_\epsilon} \ge \frac{1}{d^2} 2^{-10d^3(B +B_\epsilon +B_{\widehat{\mu}})}.
	\end{equation*}
	
	We observe that, given the definition of $\rho$ at Line~\ref{line:sample_rho}, the distance $d(x^\diamond,x)$ is strictly smaller than $d(x^\diamond,H)$, showing that $x \in \textnormal{int}(\mathcal{P})$. %\bac{showing that $x \in \textnormal{int}(\mathcal{P})$ ?? }
	
	Furthermore, we can prove that the bit-complexity of $x$ is bounded by $\mathcal{O}(d^3(B +B_\epsilon +B_{\widehat{\mu}}) + \log(\nicefrac{1}{\delta}))$.
	%\bac{To do so, we observe that?}
	To do so, we observe that the denominator of $x^\diamond$ is equal to $d \prod_{h=1}^{d}\nu_{h}$, while the denominator of $y_i$ is equal to $M=\lceil \nicefrac{\sqrt{d}}{\delta} \rceil$.
	As a result, the denominator of every $x_i= x^\circ_i +\rho y_i$, with $i \in [d-1]$, can be written as follows:
	\begin{equation*}
		D = d \prod_{h=1}^{d}\nu_{h} D_\rho M,
	\end{equation*}
	where $D_\rho$ is the denominator of the rational number $\rho$.
	Similarly, the last component $x_{d}$ can be written with the same denominator.
	As a result, the bit complexity of $x \in [0,1]^{d}$ can be upper bounded as follows:
	\begin{align*}
		B_x &\le 2\left\lceil \log(D) \right\rceil \\
		&=\mathcal{O}(\log(d \prod_{h=1}^{d}\nu_{h} D_\rho M)) \\
		&= \mathcal{O}\left( \log\left(\prod_{h=1}^{d}2^{9d^2(B +B_\epsilon +B_{\widehat{\mu}})}\right) +\log(d2^{10d^3(B +B_\epsilon +B_{\widehat{\mu}})}) + \log(\nicefrac{\sqrt{d}}{\delta})\right) \\
		&=\mathcal{O}\left(d^3(B +B_\epsilon +B_{\widehat{\mu}}) + \log\left(\frac{1}{\delta}\right)\right).
	\end{align*}
	Finally, with the same analysis performed in Lemma 4.8 by~\citet{bacchiocchi2024sample}, we can show that the probability that $x$ belongs to a given hyperplane $H$ is at most $\delta$.
\end{proof}

\begin{restatable}{lemma}{vertexbits}
	\label{lem:vertex_bits}
	Each vertex $v$ of a polytope $\mathcal{P} \subseteq \mathcal{X}_\epsilon:\textnormal{vol}_{d-1}(\mathcal{P}) >~0$, defined by separating or boundary hyperplanes, has bit-complexity at most $9d^2(B + B_\epsilon +B_{\widehat{\mu}})$. 
	Furthermore, with a bit-complexity of $9d^2(B + B_\epsilon +B_{\widehat{\mu}})$, all the components of the vector $v$ identifying a vertex can be written as fractions with the same denominator.
\end{restatable}

\begin{proof}
	We follow a line of reasoning similar to the proof of Lemma~D.2 in \citet{bacchiocchi2024sample}.
	Let $v$ be a vertex of the polytope $\mathcal{P}$.
	Then such a vertex lies on the hyperplane $H'$ ensuring that the sum of its components is equal to one. 
	Furthermore, it also belongs to a subset of $d-1$ linearly independent hyperplanes.
	These can be separating hyperplanes: 
	\[
	H_{ij}=\left \{ x \in \mathbb{R}^d \mid \sum_{\theta \in \Theta} \mu_\theta x_\theta (u_{\theta}(a_i) -u_{\theta}(a_j) ) = 0 \right \}
	\]
	 with $a_i,a_j \in \A$, boundary hyperplanes of the form $H_i = \{x \in \mathbb{R}^d \mid x_i > 0\}$, or the hyperplane $\widehat{H} \coloneqq \left\{ x \in \mathbb{R}^d \mid \sum_{\theta  \in \widetilde{\Theta}} \widehat{\mu}_{\theta}x_\theta \ge 2\epsilon \right\}$.
	Consequently, there exists a matrix $A \in \mathbb{Q}^{d \times d}$ and a vector $b \in \mathbb{Q}^{d}$ such that $Av=b$.
	%Each entry $a_{ij}$ of $A$ is either $1$, $\mu_{\theta} (u_{\theta}(a) -u_{\theta}(a')$ for some $\theta \in \Theta,a,a' \in \mathcal{A}$, or $\tilde{\mu}_{\theta}$ for some $\theta \in \Theta$. 
	%Thus, the bit-complexity of $a_{ij}$ is at most $\max(L_1+2L_2,2\lceil\log(p)\rceil)$, where $p = \left\lceil \frac{1}{2 \epsilon^2}\log\left(\frac{2d}{\delta} \right) \right\rceil$.
	
	Suppose that $v$ is not defined by the hyperplane $\widehat{H} \coloneqq \left\{ x \in \mathbb{R}^d \mid \sum_{\theta  \in \Theta} \widehat{\mu}_{\theta}x_\theta \ge 2\epsilon \right\}$.
	Then, each entry of the matrix $A$ is either equal to one or the quantity $\mu_{\theta} (u_{\theta}(a) -u_{\theta}(a')$ for some $\theta \in \Theta $ and $a,a' \in \mathcal{A}$.
	Thus, its bit-complexity is bounded by $B$.
	Similarly, each entry of the vector $b$ is either equal to one or zero.
	With a reasoning similar to the one applied in \citet{bacchiocchi2024sample}, the bit-complexity of $v$ is at most $9d^2(B_\mu + B_u)$. %\bollo{Rifacendo bene i calcoli credo si possa ridurre a qualcosa come $kB_\mu + 9B_u$, con $k<9$, ma non credo sia utile}.
	
	Suppose instead that $v$ is defined also by the hyperplane $\widehat{H}$, corresponding to the last row of the matrix $A$ and the last component of the vector $b$.
	This hyperplane can be rewritten as:
	\begin{equation*}
		\widehat{H} = \left\{ x \in \mathbb{R}^d \mid \sum_{\theta  \in \widetilde{\Theta}} \frac{\widehat{\mu}_{\theta}}{2\epsilon} x_\theta \ge 1 \right\}.
	\end{equation*}
	Thus, each element of the last row of $A$ is either zero or the quantity $\nicefrac{\widehat{\mu}_{\theta}}{2\epsilon}$ for some $\theta \in \widetilde{\Theta}$.
	We observe that $\nicefrac{\widehat{\mu}_{\theta}}{2\epsilon}$ is a rational number with numerator bounded by $2^{B_{\widehat{\mu}} +B_\epsilon}$ and denominator bounded by $2^{B_{\widehat{\mu}} +B_\epsilon +1}$.
	Thus, we multiply the last row of $A$ and the last component of $b$ by a constant bounded by $2^{d(B_{\widehat{\mu}} +B_\epsilon+1)}$.
	The other rows of $A$ and the corresponding components of $b$ are multiplied instead by some constants bounded by $2^{4dB}$.
	This way, we obtain an equivalent system $A'v=b'$ with integer coefficients. 
	
	We define ${A'}(j)$ as the matrix obtained by substituting the $j$-th column of $A'$ with $b'$. Then, by Cramer's rule, the value of the $j$-th component of $v_j$ can be computed as follows:
	\begin{equation*}
		v_j = \frac{\det(A'(j))}{\det(A')} \,\,\,\ \textnormal{$\forall j \in [d]$}.
	\end{equation*} 
	We observe that both determinants are integer numbers as the entries of both $A'$ and $b'$ are all integers, thus by Hadamard's inequality we have:
	\begin{align*}
		|\det(A')| &\leq \prod_{i \in [d]}\sqrt{\sum_{j \in [d]} {{a}'_{ji}}^2 } \\
		&\leq \left( \prod_{i \in [d-1]}\sqrt{\sum_{j \in [d]} (2^{4dB})^2 } \right) \sqrt{\sum_{j \in [d]} (2^{d(B_{\widehat{\mu}} +B_\epsilon+1)})^2 } 	\\
		&=	\left( \prod_{i \in [d-1]}\sqrt{d (2^{4dB})^2 } \right)
		\sqrt{d 2^{2d(B_{\widehat{\mu}} +B_\epsilon+1)} } \\
		&= \left( \prod_{i \in [d-1]} d^{\frac{1}{2}} (2^{4dB}) \right)
		d^{\frac{1}{2}} 2^{d(B_{\widehat{\mu}} +B_\epsilon+1)} \\
		&= d^{\frac{d}{2}} (2^{4d(d-1)B}) 2^{d(B_{\widehat{\mu}} +B_\epsilon+1)} \\
		&\leq d^{\frac{d}{2}} (2^{4d^2B}) 2^{d(B_{\widehat{\mu}} +B_\epsilon+1)}
	\end{align*}
	With a reasoning similar to \citet{bacchiocchi2024sample}, we can show that that the bit-complexity $D_v$ of the vertex $v$ is bounded by:
	\begin{equation*}
		D_v \leq 9Bd^2 + 2(d(B_{\widehat{\mu}} +B_\epsilon+1)) 
		\leq 9d^2(B + B_\epsilon +B_{\widehat{\mu}})
	\end{equation*}
	Furthermore, this result holds when the denominator of every component $v_j$ of the vertex $v$ is written with the same denominator $\det(A')$, concluding the proof.
\end{proof}

\subsection{\texttt{Find-Face}}
\begin{algorithm}[H]
	\caption{\texttt{Find-Face}}\label{alg:find_face}
	\begin{algorithmic}[1]
		\Require The set of polytopes $\{\mathcal{X}_\epsilon(a_i)\}_{a_i \in \mathcal{C}}$, with volume larger than zero, and action $a_j \notin \mathcal{C}$ 
		\State Compute the minimal H-representation $\mathcal{M}(\mathcal{X}_\epsilon(a_i))$ for every polytope $\mathcal{X}_\epsilon(a_i), a_i \in \mathcal{C}$ 
		\State $\mathcal{H}(a_i) \gets \emptyset \quad \forall a_i \in \mathcal{C}$ 
		%		\State $x \gets $ the first vertex of $\mathcal{X}_\epsilon(a_j)$
		%		\State \bollo{linea sopra andrebbe riscritta meglio}
		%		\State $\mathcal{H}(a_i) \gets \{H \in \mathcal{M}(\mathcal{X}_\epsilon(a_i)) \mid x \in H \} \quad \forall a_i \in \mathcal{C}$% : x \in V(\mathcal{X}_\epsilon(a_i))$
		\State $\texttt{First} \gets \texttt{True}$
		\ForAll{$x \in \bigcup_{a_i \in \mathcal{C}}V(\mathcal{X}_\epsilon(a_i))$}
		\State $a \gets \texttt{Action-Oracle}(x)$
		\If{$a = a_j$}
		\If{$\texttt{First} = \texttt{False}$}
		\State $\mathcal{H}(a_i) \gets \{H \in \mathcal{M}(\mathcal{X}_\epsilon(a_i)) \mid x \in H, H \in \mathcal{H}(a_i) \} \quad \forall a_i \in \mathcal{C}$% : x \in V(\mathcal{X}_\epsilon(a_i))$
		%		\State $\mathcal{H}(a_i) \gets \emptyset \quad \forall a_i \in \mathcal{C} : x \notin V(\mathcal{X}_\epsilon(a_i))$
		\Else
		\State $\mathcal{H}(a_i) \gets \{H \in \mathcal{M}(\mathcal{X}_\epsilon(a_i)) \mid x \in H \} \quad \forall a_i \in \mathcal{C}$% : x \in V(\mathcal{X}_\epsilon(a_i))$
		\State $\texttt{First} \gets \texttt{False}$
		\EndIf
		\EndIf
		\EndFor
		\If{\texttt{First} = \texttt{True}}
		\State $\mathcal{F}_\epsilon(a_j) \gets \emptyset$
		\EndIf
		\State $\mathcal{F}_\epsilon(a_j) \gets \mathcal{X}_\epsilon(a_i) \cap \bigcap_{H \in \mathcal{H}(a_i)}H$ for any $a_i$ such that $\mathcal{X}_\epsilon(a_i) \cap \bigcap_{H \in \mathcal{H}(a_i)}H \neq \emptyset$ \label{line:return_h_representation}
		\State \textbf{Return} $\mathcal{F}_\epsilon(a_j)$ 
	\end{algorithmic}
\end{algorithm}

Algorithm~\ref{alg:find_face} takes in input the collection of polytopes $\{\mathcal{X}_\epsilon(a_i)\}_{a_i \in \mathcal{C}}$ and another action $a_j \notin \mathcal{C}$ such that $\text{vol}(\mathcal{X}_\epsilon(a_j))=0$, and outputs the H-representation of a (possibly improper) face of $\mathcal{X}_\epsilon(a_j)$ that contains all those vertices $x \in V(\mathcal{X}_\epsilon(a_j))$ where $a(x)=a_j$.
As we will show by means of a pair of technical lemmas, the polytope $\mathcal{X}_\epsilon(a_j)$ is a face of some other polytope $\mathcal{X}_\epsilon(a_k)$, with $a_k \in \mathcal{C}$.
Consequently, Algorithm~\ref{alg:find_face} looks for a face of some polytope $\mathcal{X}_\epsilon(a_k)$ containing the set of vertices $\mathcal{V}_\epsilon(a_j)$.

As a first step, Algorithm~\ref{alg:find_face} computes, for every action $a_i \in \mathcal{C}$, the set of hyperplanes $\mathcal{M}(\mathcal{X}_\epsilon(a_i))$ corresponding to the minimal H-representation of $\mathcal{X}_\epsilon(a_i)$.
This set includes every separating hyperplane $H_{ik}$ found by Algorithm~\ref{alg:find_non_zero_vol}, together with the non-redundant boundary hyperplanes that delimit $\mathcal{X}_\epsilon$. 
Subsequently, Algorithm~\ref{alg:find_face} iterates over the vertices of the regions with volume larger than zero, which we prove to include all the vertices of the region $\mathcal{X}_\epsilon(a_j)$.
While doing so, it builds a set of hyperplanes $\mathcal{H}(a_i) \subseteq \mathcal{M}(\mathcal{X}_\epsilon(a_i))$ for every action $a_i \in \mathcal{C}$.
Such a (possibly empty) set includes all and only the hyperplanes in $\mathcal{M}(\mathcal{X}_\epsilon(a_i))$ that contain all the vertices where the action $a_j$ has been observed, \emph{i.e}, $a(x)=a_j$.

Finally, at Line~\ref{line:return_h_representation} Algorithm~\ref{alg:find_face} intersects every region $\mathcal{X}_\epsilon(a_i)$ with the corresponding hyperplanes in $\mathcal{H}(a_i)$, obtaining a (possibly empty) face for every polytope $\mathcal{X}_\epsilon(a_i)$, $a_i \in \mathcal{C}$.
At least one of these faces is the face the algorithm is looking for, and corresponds to the output of Algorithm~\ref{alg:find_face}.

%We observe that \ref{alg:find_face} can be run with polynomial per-step running time, as long as the iteration over the vertices is executed with the same technique described in Appendix~\ref{sec:find_non_zero_vol}.

The main result concerning Algorithm~\ref{alg:find_face} is the following:
\begin{restatable}{lemma}{FindHReprZeroVol}
	\label{lem:find_face}
	Given the collection of polytopes $\{\mathcal{X}_\epsilon(a_i)\}_{a_i \in \mathcal{C}}$ with volume larger than zero and an another action $a_j$, then, under the event $\mathcal{E}^\textnormal{a}$, Algorithm~\ref{alg:find_face} returns a (possibly improper) face $\mathcal{F}_\epsilon(a_j)$ of $\mathcal{X}_\epsilon(a_j)$ such that $\mathcal{V}_\epsilon(a_j) \subseteq \mathcal{F}_\epsilon(a_j)$.
	Furthermore, the Algorithm requires $\mathcal{O}(n\binom{d+n}{d})$ calls to Algorithm~\ref{alg:action_oracle}.
%	 and has polynomial time and space complexity between each two consecutive queries.
\end{restatable}

In order to prove it, we first need to introduce two technical lemmas to characterize the relationship between regions with null volume and those with volume larger than zero.
\begin{restatable}{lemma}{zeroVolFace2}
	\label{lem:intersection_zero_vol_face}
	Let $a_i, a_j \in \mathcal{A}$ such that $\text{vol}(\mathcal{X}_\epsilon(a_i))>0$ and $\text{vol}(\mathcal{X}_\epsilon(a_j))=0$.
	Then $\mathcal{X}_\epsilon(a_i) \cap \mathcal{X}_\epsilon(a_j)$ is a (possibly improper) face of $\mathcal{X}_\epsilon(a_i)$ and $\mathcal{X}_\epsilon(a_j)$.
\end{restatable}
\begin{proof}
	In the following we assume that $\mathcal{X}_\epsilon(a_i) \cap \mathcal{X}_\epsilon(a_j)$ is non-empty, as the empty set is an improper face of every polytope.
	
%	\bac{in generale se devi provare l'uguaglianza tra due set $A,B$ il modo piu chiaro secondo me e $A \subseteq B$ e $B \subseteq A$}
	
	We prove that $\mathcal{X}_\epsilon(a_i) \cap \mathcal{X}_\epsilon(a_j) = \mathcal{X}_\epsilon(a_i) \cap H_{ij} = \mathcal{X}_\epsilon(a_j) \cap H_{ij}$.
	In order to do that, we first show that $\mathcal{X}_\epsilon(a_i) \cap H_{ij} \subseteq \mathcal{X}_\epsilon(a_i) \cap \mathcal{X}_\epsilon(a_j)$.
	Consider a normalized slice $x \in \mathcal{X}_\epsilon(a_i) \cap H_{ij}$.
	Then we have that$x \in \mathcal{X}_\epsilon(a_j)$.
	As this holds for every $x \in \mathcal{X}_\epsilon(a_i) \cap H_{ij}$, it follows that $\mathcal{X}_\epsilon(a_i) \cap H_{ij} \subseteq \mathcal{X}_\epsilon(a_i) \cap \mathcal{X}_\epsilon(a_j)$.
	
	Similarly, we show that $\mathcal{X}_\epsilon(a_i) \cap \mathcal{X}_\epsilon(a_j) \subseteq \mathcal{X}_\epsilon(a_i) \cap H_{ij}$.
	Take any normalized slice $x \in \mathcal{X}_\epsilon(a_i) \cap \mathcal{X}_\epsilon(a_j)$. 
	Then $x \in H_{ij}$ as it belongs to both $\mathcal{X}_\epsilon(a_i)$ and $\mathcal{X}_\epsilon(a_j)$, thus $x \in \mathcal{X}_\epsilon(a_i) \cap H_{ij}$.
	This implies that $\mathcal{X}_\epsilon(a_i) \cap \mathcal{X}_\epsilon(a_j) \subseteq \mathcal{X}_\epsilon(a_i) \cap H_{ij}$.
	
	Consequently, we have that $\mathcal{X}_\epsilon(a_i) \cap \mathcal{X}_\epsilon(a_j) = \mathcal{X}_\epsilon(a_i) \cap H_{ij}$.
	With a similar argument, we can prove that $\mathcal{X}_\epsilon(a_i) \cap \mathcal{X}_\epsilon(a_j) = \mathcal{X}_\epsilon(a_j) \cap H_{ij}$.
	As a result, we have that $\mathcal{X}_\epsilon(a_i) \cap \mathcal{X}_\epsilon(a_j) = \mathcal{X}_\epsilon(a_i) \cap H_{ij} = \mathcal{X}_\epsilon(a_j) \cap H_{ij}$.
	
	In order to conclude the proof, we show that $\mathcal{X}_\epsilon(a_i) \cap \mathcal{X}_\epsilon(a_j) = \mathcal{X}_\epsilon(a_i) \cap H_{ij} = \mathcal{X}_\epsilon(a_j) \cap H_{ij}$ is a face of both $\mathcal{X}_\epsilon(a_i)$ and $\mathcal{X}_\epsilon(a_j)$.
	We observe that $\mathcal{X}_\epsilon(a_i) \subseteq \mathcal{H}_{ij}$,
%	\bac{non capisco il senso di questa inclusione} 
	thus the non-empty region $\mathcal{X}_\epsilon(a_i) \cap H_{ij}$ is by definition a face of $\mathcal{X}_\epsilon(a_i)$.
	Similarly, $\mathcal{X}_\epsilon(a_j) \subseteq \mathcal{H}_{ji}$, thus the non-empty region $\mathcal{X}_\epsilon(a_j) \cap H_{ij}$ is a face of $\mathcal{X}_\epsilon(a_j)$ (possibly the improper face $\mathcal{X}_\epsilon(a_j)$ itself).
\end{proof}

\begin{restatable}{lemma}{zeroVolFace}
	\label{lem:zero_vol_face}
	Let $\mathcal{X}_\epsilon(a_j)$ be a polytope such that $\text{vol}(\mathcal{X}_\epsilon(a_j))=0$.
	Then $\mathcal{X}_\epsilon(a_j)$ is a face of some polytope $\mathcal{X}_\epsilon(a_i)$ with $\text{vol}(\mathcal{X}_\epsilon(a_i))>0$.
\end{restatable}
\begin{proof}
	First, we observe that if $\mathcal{X}_\epsilon(a_j)$ is empty, then it is the improper face of any region $\mathcal{X}_\epsilon(a_i)$ with $\text{vol}(\mathcal{X}_\epsilon(a_i))>0$.
	Thus, in the following, we consider $\mathcal{X}_\epsilon(a_j)$ to be non-empty.
	
	As a first step, we observe that any normalized slice $x \in \mathcal{X}_\epsilon(a_j)$ belongs also to some region $\mathcal{X}_\epsilon(a_k)$, where $a_k \in \mathcal{A}$ depends on $x$, such that $\text{vol}(\mathcal{X}_\epsilon(a_k))>0$.
	Suppose, by contradiction, that $x \in \text{int}(\mathcal{X}_\epsilon(a_i))$.
	Then $a_i$ is a best-response in $\mathcal{X}_\epsilon(a_j) \cap \mathcal{H}_{ij}$, \emph{i.e.}, $\mathcal{X}_\epsilon(a_j) \cap \mathcal{H}_{ij} \subseteq \mathcal{X}_\epsilon(a_i)$.
	One can easily observe that such a region has positive volume, thus contradicting the hypothesis that $\text{vol}(\mathcal{X}_\epsilon(a_j))=0$.
	
	Now we prove that there exists some $\mathcal{X}_\epsilon(a_i)$ with $\text{vol}(\mathcal{X}_\epsilon(a_i))>0$ such that $\mathcal{X}_\epsilon(a_j) \subseteq \mathcal{X}_\epsilon(a_i)$. 
	If $\mathcal{X}_\epsilon(a_j)$ is a single normalized slice $x$, then this trivially holds.
	
	Suppose instead that $\mathcal{X}_\epsilon(a_j)$ has dimension at least one.
	Consider a fixed normalized slice $\bar{x} \in \text{int}(\mathcal{X}_\epsilon(a_j))$, where the interior is taken relative to the subspace that contains $\mathcal{X}_\epsilon(a_j)$ and has minimum dimension.
	There exists a region $\mathcal{X}_\epsilon(a_i)$ with $\text{vol}(\mathcal{X}_\epsilon(a_i))>0$ such that $\bar{x} \in \mathcal{X}_\epsilon(a_i)$.
	
	We prove that $\mathcal{X}_\epsilon(a_j) \subseteq \mathcal{X}(a_i)$.
	Suppose, by contradiction, that there exists a normalized slice $x \in \mathcal{X}_\epsilon(a_j)$ such that $x \notin \mathcal{X}_\epsilon(a_i)$.
	It follows that the line segment $\text{co}(\bar{x},x)$ intersect the separating hyperplane $H_{ij}$ in some normalized slice $\widetilde{x} \in \text{co}(\bar{x},x) \cap H_{ij}$.
	Furthermore, since $\widetilde{x} \neq x$ and $\bar{x} \in \text{int}(\mathcal{X}_\epsilon(a_j))$, then $\widetilde{x} \in \text{int}(\mathcal{X}_\epsilon(a_j))$.
	However, if the internal point $\widetilde{x}$ belongs to the hyperplane $H_{ij}$ and $\mathcal{X}_\epsilon(a_j) \subseteq \mathcal{H}_{ji}$, then it must be the case that $\mathcal{X}_\epsilon(a_j) \subseteq H_{ij}$.
	This implies that $\mathcal{X}_\epsilon(a_j) \subseteq \mathcal{X}_\epsilon(a_i)$ and thus $x \in \mathcal{X}_\epsilon(a_i)$, which contradicts the hypothesis.. 
	
	Given that there exists some $\mathcal{X}_\epsilon(a_i)$ with $\text{vol}(\mathcal{X}_\epsilon(a_i))>0$ such that $\mathcal{X}_\epsilon(a_j) \subseteq \mathcal{X}_\epsilon(a_i)$, then $\mathcal{X}_\epsilon(a_j) = \mathcal{X}_\epsilon(a_i) \cap \mathcal{X}_\epsilon(a_j)$ is a face of $\mathcal{X}_\epsilon(a_i)$ by Lemma~\ref{lem:intersection_zero_vol_face}.
\end{proof}

\FindHReprZeroVol*
\begin{proof}
	In the following, for the sake of notation, given a polytope $\mathcal{X}_\epsilon(a)$ and the a set of hyperplanes $\mathcal{H}(a)$, with an abuse of notation we denote with $\mathcal{X}_\epsilon(a) \cap \mathcal{H}(a)$ the intersection of $\mathcal{X}_\epsilon(a)$ with every hyperplane in $\mathcal{H}(a)$.
	Formally:
	\begin{equation}
		\label{eq:abuse_intersect_x_h}
		\mathcal{X}_\epsilon(a) \cap\ \mathcal{H}(a) \coloneqq \mathcal{X}_\epsilon(a) \cap \bigcap_{H \in \mathcal{H}(a)} H.
	\end{equation}
	
	Suppose that $\mathcal{X}_\epsilon(a_j) = \emptyset$.
	Then, one can easily verify that Algorithm~\ref{alg:find_face} returns $\emptyset$.
	Thus, in the following we assume $\mathcal{X}_\epsilon(a_j) \neq \emptyset$.
	
	Let $a_i$ be action selected at Line~\ref{line:return_h_representation} Algorithm~\ref{alg:find_face}.
	We denote with $\mathcal{F}_\epsilon(a_i)$ the face returned by Algorithm~\ref{alg:find_face}:
	\begin{equation*}
		\mathcal{F}_\epsilon(a_i) \coloneqq \mathcal{X}_\epsilon(a_i) \cap \mathcal{H}(a_i).
	\end{equation*}
	%	Furthermore, we denote with $\mathcal{W}$ the set of vertices where action $a_j$ is observed:
	%	\begin{equation*}
		%		\mathcal{W} \coloneqq \{x \in V(\mathcal{X}_\epsilon(a_j)) \mid a(x) = a_j\}.
		%	\end{equation*}
	We observe that by Lemma~\ref{lem:zero_vol_face}, there exists an action $a_k \in \mathcal{C}$ such that $\mathcal{X}_\epsilon(a_j)$ is a face of $\mathcal{X}_\epsilon(a_k)$.
	Consequently, Algorithm~\ref{alg:find_face} queries every vertex $x \in \mathcal{V}_\epsilon(a_j)$.	
	
	As a first step we show that $\mathcal{F}_\epsilon(a_i)$ actually is a face of $\mathcal{X}_\epsilon(a_i)$ and contains every vertex of $\mathcal{V}_\epsilon(a_j)$.
	Being the non-empty intersection of $\mathcal{X}_\epsilon(a_i)$ with some hyperplanes in $\mathcal{M}(\mathcal{X}_\epsilon(a_i))$, $\mathcal{F}_\epsilon(a_i)$ is a face of $\mathcal{X}_\epsilon(a_i)$.
	One can easily prove by induction that $\mathcal{H}(a_i)$ includes all and only the hyperplanes within $\mathcal{M}(\mathcal{X}_\epsilon(a_i))$ containing every vertex in $\mathcal{V}_\epsilon(a_j)$.
	Thus, $\mathcal{V}_\epsilon(a_j) \subseteq \mathcal{F}_\epsilon(a_i)$.
	
	%	Furthermore, there exists an action $a_i$ that Algorithm~\ref{alg:find_face} can selects at Line~\ref{line:return_h_representation}, \emph{i.e.}, such that $\mathcal{F}_\epsilon(a_i)$ is non-empty.
	%	Thanks to Lemma~\ref{lem:zero_vol_face}, there exists some action $a_k \in \mathcal{C}$ such that $\mathcal{X}_\epsilon(a_j)$ is a face of $\mathcal{X}_\epsilon(a_k)$.
	%	Consequently the algorithm can always select $a_i = a_j$. 
	
	Now we show that $\mathcal{F}_\epsilon(a_i)$ is not only a (proper) face of $\mathcal{X}_\epsilon(a_i)$ containing the set $\mathcal{V}_\epsilon(a_j)$, but also a face of $\mathcal{X}_\epsilon(a_j)$ (possibly the improper face $\mathcal{X}_\epsilon(a_j)$ itself).
	We consider the set $\mathcal{X}_\epsilon(a_i) \cap \mathcal{X}_\epsilon(a_j)$, which is a face of both $\mathcal{X}_\epsilon(a_i)$ and $\mathcal{X}_\epsilon(a_j)$ thanks to Lemma~\ref{lem:intersection_zero_vol_face}.
	Thus, there exists some set of hyperplanes $\mathcal{H}'(a_i) \subset \mathcal{M}(\mathcal{X}_\epsilon(a_i))$ such that:
	\begin{equation}
		\label{eq:h_repr_i_cap_j}
		\mathcal{X}_\epsilon(a_i) \cap \mathcal{H}'(a_i) = \mathcal{X}_\epsilon(a_i) \cap \mathcal{X}_\epsilon(a_j).
	\end{equation}
	Furthermore, we observe that $\mathcal{V}_\epsilon(a_j) \subseteq \mathcal{X}_\epsilon(a_i) \cap \mathcal{X}_\epsilon(a_j)$.
	Indeed, we have that $\mathcal{V}_\epsilon(a_j) \subseteq \mathcal{X}_\epsilon(a_i)$ since $\mathcal{V}_\epsilon(a_j) \subseteq \mathcal{F}_\epsilon(a_i)$ and $\mathcal{F}_\epsilon(a_i)$ is a face of $\mathcal{X}_\epsilon(a_i)$, and $\mathcal{V}_\epsilon(a_j) \subseteq \mathcal{X}_\epsilon(a_j)$ by definition.
	
	We want to prove that $\mathcal{H}'(a_i) \subseteq \mathcal{H}(a_i)$, where the set $\mathcal{H}(a_i)$ contains all and only the hyperplanes within $\mathcal{M}(\mathcal{X}_\epsilon(a_i))$ that include the whole set $\mathcal{V}_\epsilon(a_j)$.
	In order to do that, suppose, by contradiction, that there exists a vertex $x \in \mathcal{V}_\epsilon(a_j)$ such that $x \notin H$ for some $H \in \mathcal{H}'(a_i)$.
	Then, $x \notin \mathcal{X}_\epsilon(a_i) \cap \mathcal{H}'(a_i) \subseteq H$.
	However, we proved that $\mathcal{V}_\epsilon(a_j) \subseteq \mathcal{X}_\epsilon(a_i) \cap \mathcal{X}_\epsilon(a_j)$ and $\mathcal{X}_\epsilon(a_i) \cap \mathcal{X}_\epsilon(a_j) = \mathcal{X}_\epsilon(a_i) \cap \mathcal{H}'(a_i)$ by definition of $\mathcal{H}'(a_i)$, reaching a contradiction.
	
	Consequently:
	\begin{align*}
		\mathcal{F}_\epsilon(a_i) \coloneqq 	\mathcal{X}_\epsilon(a_i) \cap \mathcal{H}(a_i) &= \mathcal{X}_\epsilon(a_i) \cap \bigcap_{H \in \mathcal{H}(a_i)} H \\
		&\subseteq \mathcal{X}_\epsilon(a_i) \cap \bigcap_{H \in \mathcal{H}'(a_i)} H \\
		&= \mathcal{X}_\epsilon(a_i) \cap \mathcal{H}'(a_i) \\
		&= \mathcal{X}_\epsilon(a_i) \cap \mathcal{X}_\epsilon(a_j),
	\end{align*}
	where we applied Equation~\ref{eq:abuse_intersect_x_h}, the fact that $\mathcal{H}'(a_i) \subseteq \mathcal{H}(a_i)$, and Equation~\ref{eq:h_repr_i_cap_j}.
	
	Finally, we can show that $\mathcal{F}_\epsilon(a_i)$ is a face of $\mathcal{X}_\epsilon(a_j)$.
	We have that $\mathcal{F}_\epsilon(a_i)$ is a face of $\mathcal{X}_\epsilon(a_i)$ and $\mathcal{F}_\epsilon(a_i) \subseteq \mathcal{X}_\epsilon(a_i) \cap \mathcal{X}_\epsilon(a_j)$, which is itself a face of $\mathcal{X}_\epsilon(a_i)$ by Lemma~\ref{lem:intersection_zero_vol_face}.
	Thus, $\mathcal{F}_\epsilon(a_i)$ is a face of $\mathcal{X}_\epsilon(a_i) \cap \mathcal{X}_\epsilon(a_j)$.
	Furthermore, Lemma~\ref{lem:intersection_zero_vol_face} states that $\mathcal{X}_\epsilon(a_i) \cap \mathcal{X}_\epsilon(a_j)$ is also a face of $\mathcal{X}_\epsilon(a_j)$.
	This implies that $\mathcal{F}_\epsilon(a_i)$ is a face of a face of $\mathcal{X}_\epsilon(a_j)$, and thus a face of $\mathcal{X}_\epsilon(a_j)$ itself. 
	
	In order to conclude the proof, we have to prove that at Line~\ref{line:return_h_representation} Algorithm~\ref{alg:find_face} can actually find an action $a_i$ such that $\mathcal{F}_\epsilon(a_i) = \mathcal{X}_\epsilon(a_i) \cap \mathcal{H}(a_i)$ is non-empty.
	Let $a_k \in \mathcal{C}$ be such that $\mathcal{X}_\epsilon(a_j)$ is a face of $\mathcal{X}_\epsilon(a_k)$, which exists thanks to Lemma~\ref{lem:zero_vol_face}.
	Let $x$ be any vertex in the set $\mathcal{V}_\epsilon(a_j)$ and define $\mathcal{H}''(a_j)$ as:
	\begin{equation*}
		\mathcal{H}''(a_k) \coloneqq \{H \in \mathcal{R}^H(\mathcal{X}_\epsilon(a_k)) \mid x \in H \}.
	\end{equation*}
	Then $\mathcal{X}_\epsilon(a_k) \cap \mathcal{H}''(a_k) = \{x\}$.
	Consequently, $\mathcal{H}(a_k) \subseteq \mathcal{H}''(a_k)$, and thus:
	\begin{align*}
		\mathcal{X}_\epsilon(a_k) \cap \mathcal{H}(a_k) &= \mathcal{X}_\epsilon(a_k) \cap \bigcap_{H \in \mathcal{H}(a_k)} H \\
		&\subseteq \mathcal{X}_\epsilon(a_k) \cap \bigcap_{H \in \mathcal{H}''(a_k)} H \\
		&= \mathcal{X}_\epsilon(a_k) \cap \mathcal{H}''(a_k) \\
		&= \{x\} \neq \emptyset.
	\end{align*}
	As a result, there is always an action $a_k \in \mathcal{C}$ such that $\mathcal{X}_\epsilon(a_k) \cap \mathcal{H}(a_k) \neq \emptyset$.
	
	Finally, we observe that Algorithm~\ref{alg:find_face} executes Algorithm~\ref{alg:action_oracle} once for every vertex in the set $\bigcup_{a_i \in \mathcal{C}} V(\mathcal{X}_\epsilon(a_i))$, which has size $\mathcal{O}(n\binom{d+n}{d})$.
\end{proof}

%% file: appendix_lb_regret.tex
\section{Omitted proofs from Section~\ref{sec:lower_bounds_regret}}\label{appendix:lower_bounds_regret}
\hardnessfirst*
\begin{proof}
In the following, for the sake of the presentation, we consider a set of instances characterized by an even number $d \in \mathbb{N}_{+}$ of states of nature and $n = d+2$ receiver's actions. All the instances share the same uniform prior distribution and the same sender's utility, given by $u^\text{s}_\theta(a_{d+1})=1$ for all $\theta \in \Theta$, and $u^\text{s}_\theta(a)=0$ for all $\theta \in \Theta$ and $\forall a \in {\mathcal{A}} \setminus \{a_{d+1}\}$. 
%\bollo{$a_{n+1}$ e $a_{n+2}$ sono $a_{d+1}$ e $a_{d+2}$}
Each instance is parametrized by a vector $p$ belonging to a set $\mathcal{P} $ defined as follows: 
\[
\mathcal{P} \coloneqq \left \{ p \in \{ 0,1\}^d \, | \, \sum_{i=1}^{d} p_i =  \frac{d}{2} \right \}.
\]
Furthermore, we assume that the receiver's utility in each instance $I_p$ is given by:
\[
\circled{$I_p$}\begin{cases}
	u_{\theta_i}(a_j)= \delta_{ij} \,\,\, \forall i,j \in [d],\\% \bollo{$a_j$?} \\
	u_{\theta_i}(a_{d+1})= \frac{2}{d} p_i \,\,\, \forall i \in [d], \\
	u_{\theta_i}(a_{d+2})=\frac{2}{d} \,\,\, \forall i \in [d].
\end{cases}
\]
% \bollo{La prima equazione non dovrebbe essere  $u_{\theta_i}(a_j)= \delta_{ij} \,\,\, \forall i \in [d],j \in [d]$?}
We show that $\xi_{\theta_i}' \coloneqq \frac{2}{d}p_i$ for each $ i \in [d]$ is the only posterior inducing the receiver's action $a_{d+1} \in \mathcal{A}$ in the instance $I_p$, since the receiver breaks ties in favor of the sender. To prove that, we observe that the action $a_{d+1}$ is preferred to the action $a_{d+2}$ only in those posteriors that satisfy the condition $\xi_{\theta_i}=0$ for each $i \in [d]$ with $p_i=0$. Furthermore, to incentivize the action $a_{d+1}$ over the set of actions $a_i$ with $i \in [d]$, the following condition must hold:
\[ \sum_{i \in [d]: p_i >0} \xi_{\theta_i} u_{\theta_i}(a_{n+1}) = \frac{2}{d}\sum_{i \in [d]: p_i >0} \xi_{\theta_i} 	= \frac{2}{d} \ge  \max_{i \in [d]: p_i >0} \xi_{\theta_i}.\,\, 
\]
We notice that the last step holds only if $\xi_{\theta_i} \le 2/d$ for each $i \in [d]$ such that $p_i>0$. Consequently, since the number of $\xi_{\theta_i} >0$ is equal to $d/2$, it holds $\xi_{\theta_i} = 2/d$ for each $i \in [d]$ such that $p_i>0$.
%	\bollo{L'ultimo passaggio non è al contrario?
%	\begin{equation*}
%		\frac{2}{d}\sum_{i \in [d]: p_i >0} \xi_{\theta_i} = \frac{2}{d} \ge \max_{i \in [d]: p_i >0} \xi_{\theta_i}
%	\end{equation*}
%	Quindi $\xi_{\theta_i} \le 2/d$. Ci sono $d/2$ elementi per cui $\xi_{\theta_i} \neq 0$, e devono sommare ad uno, l'unico modo è che siano tutti $2/d$.}
%	Furthermore, we observe that: $$ \max_{i \in [d]: p_i >0} \xi_{\theta_i} = \frac{2}{d} \quad \textnormal{iff}\quad \xi_{\theta_i} = \frac{2}{d} \,\, \forall i \in [d]:p_i>0.$$
Thus, the only posterior inducing action $a_{d+1}$ is equal to $\xi_{\theta_i}' \coloneqq \frac{2}{d}p_i$.
%Indeed, the action $a_{n+1}$ is strictly dominated by action $a_{n+2}$ whenever the following condition holds $\sum_{i \in n:p_i=0} \xi_{\theta_i}>0.$
%Moreover, if $\xi_{\theta_i} > \frac{2}{d}$ for at least one $i \in [d]$, then the utility of action $a_{i}$ is strictly greater that $\frac{2}{d}$, while the utility of action $a_{n+1}$ is at most $\frac{2}{d}$.

We also notice that, given $p \in \mathcal{P}$, the optimal signaling scheme $\gamma$ is defined as $\gamma(\xi') = \nicefrac{1}{2}$ and $\gamma(\xi'') = \nicefrac{1}{2}$, with $\xi_\theta'' = \mu_\theta - \frac{1}{2} \xi'_\theta$ for each $\theta \in \Theta$. With a simple calculation, we can show that the expected sender's utility in $\gamma$ is equal to $\nicefrac{1}{2}$. 
%Furthermore, we notice that for each different $p, p' \in \mathcal{P}$  it holds $I_p \cap I_{p'} = \emptyset$, and given two instances $I_p$ and $I_{p'}$, the only case in which the sender receives a different feedback after committing to a signaling scheme is if the induced posterior is respectively $\frac{2}{n}p$ or $\frac{2}{n}p'$. Thus, to determine if there exits a $p \in \mathcal{P}$ inducing the action $a_{n+1}$, any deterministic algorithm is required to check if there exists a $p \in \mathcal{P}$ inducing such an action or not. This requires at least $\binom{n}{\lceil n/2 \rceil }$ samples, \emph{i.e.}, the commitment of $2^{\Omega(n)}$ signaling schemes.
%To conclude the poof, we assign the following distribution to the inputs $p(I_p)= \nicefrac{1}{\binom{n}{\lceil n/2 \rceil }}$ and applying Yao’s minimax principle we can show that the number of rounds required to compute an approximately signaling scheme of at least $2^{\Omega(n)}$.
%By the Yao's minimax principle, it is sufficient to show that any deterministic algorithm fails with large probability against a distribution over instances. In the following, we consider the uniform distribution over instances $\mathcal{P}$.

We set the time horizon $T= \lfloor\nicefrac{|\mathcal{P}|}{4}\rfloor$ to show that any algorithm suffers regret of at least $2^{\Omega(d)}$. This is sufficient to prove the statement. We start by making the following simplifications about the behavior of the algorithm. First, we observe that if the algorithm can choose any posterior (instead of a signaling scheme), then this will only increase the performance of the algorithm. Consequently, we assume that the algorithm chooses a posterior $\xi_t$ at each round $t \in [T]$.

Thus, we can apply Yao's minimax principle to show that any deterministic algorithm fails against a distribution over instances. In the following, we consider a uniform distribution over instances $I_p$ with $p \in \mathcal{P}$. Furthermore, we observe that the feedback of any algorithm is actually binary. Thus, it is easy to see that an optimal algorithm works as follows: (i) it ignores the feedback whenever the action is not $a_{d+1}$, and (ii) it starts to play the optimal posterior when the action is $a_{d+1}$ since it found an optimal posterior.

This observation is useful for showing that any deterministic algorithm does not find a posterior that induces action $a_{d+1}$ with a probability of at least $1 - \nicefrac{|\mathcal{P}|}{(4|\mathcal{P}|)} = \nicefrac{3}{4}$ (since it can choose only $\lfloor\nicefrac{|\mathcal{P}|}{4}\rfloor$ posteriors among the $|\mathcal{P}|$ possible optimal posteriors). Hence, by Yao's minimax principle, for any (randomized) algorithm there exists an instance such that the regret suffered in the $T$ rounds is at least:
%	\[   \frac{T}{2} - \frac{|\mathcal{P}|}{4} T\ge \frac{|\mathcal{P}|}{4}. \]
\begin{equation*}
	R_T \ge \frac{3}{4}\frac{T}{2} \ge \frac{1}{4} \left\lfloor \frac{|\preg|}{4} \right\rfloor \ge   \frac{|\preg|}{32},
\end{equation*}
since $\lfloor x \rfloor \ge x -1 \ge x/2$, for each $x \ge 2$. 
%\left \lfloor \frac{|\preg|}{32} \right\rfloor. $$
Finally, we notice that $|\mathcal{P}|=\binom{d}{\nicefrac{d}{2}} = 2^{\Omega(d)}$, which concludes the proof. 
%Then, we set $T=2^{n}/4$ rounds, and show that any algorithm suffers regret at least $\Omega(2^n)$. This is sufficient to prove the statement.
% We start making the following simplifications about the behavior of the algorithm. First, we observe that if the algorithm can choose any posterior (intead of a signaling scheme), then this will only increase the performance of the algorithm. Hence, we assume that the algorithm chooses a posterior $\xi_t$ at each round $t \in [T]$. 
%	Then, we can apply the Yao's minimax principle. Hence, it is sufficient to show that any deterministic algorithm fails against a distribution over instances. In the following, we consider the uniform distribution over instances $\mathcal{P}$.
%		
%	 Finally, we observe that the feedback of the algorithm is actually binary. Hence, it is easy to see that an optimal algorithm works as follow: i) ignores the feedback whenever the action is not $a_{n+1}$, and ii) start playing the optimal posterior when the action is $a_{n+1}$ since it found an optimal posterior.
%	 
%	This observation is useful to show that the algorithm does not find an posterior that induces action $a_{n+1}$ with probability at least $1- \frac{2^{n}/4}{2^n}=\frac{3}{4}2^n$ (since it can choose only $2^{n}/4$ over $2^n$).
%	 
%	Hence, the regret in the T rounds is at least 
%	\[   \frac{T}{2} - \frac{2^n}{4} T\ge \frac{2^n}{4}. \]
\end{proof}

\hardnessthird*
\begin{proof}
	To prove the theorem, we introduce two instances characterized by two states of nature and four receiver's actions. In both the instances the sender's utility is given by $u^\text{s}_\theta(a_{1})=u^\text{s}_\theta(a_{2})=0$ for all $\theta \in \Theta$, while $u^\text{s}_{\theta_1}(a_{3})=1$, $u^\text{s}_{\theta_2}(a_{3})=0$ and $u^\text{s}_{\theta_2}(a_{4})=0$ $u^\text{s}_{\theta_2}(a_{4})=1$. The receiver's utilities and the prior distributions in the two instances are:
	\begin{alignat*}{2}
		&\circled{1 }
		\begin{cases}
			\mu^{1}_{\theta_1}= \frac{1}{2}+\epsilon   ,\, \mu^{1}_{\theta_2} =\frac{1}{2}- \epsilon\\
			u^{1}_{\theta_1}(a_1) =\frac{1}{\left(2+4\epsilon \right)} ,\, u^{1}_{\theta_2}(a_1) =\frac{1}{\left(10-20\epsilon \right)}  \\
			u^{1}_{\theta_1}(a_2) =\frac{1}{\left(10+20\epsilon \right)},\, u^{1}_{\theta_2}(a_2) =\frac{1}{\left(2-4\epsilon \right)}\\
			u^{1}_{\theta_1}(a_3) =\frac{3}{10},\, u^{1}_{\theta_2}(a_3) =\frac{3}{10}\\
			u^{1}_{\theta_1}(a_4) =0,\, u^{1}_{\theta_2}(a_4) = \frac{1}{\left(2-4\epsilon \right)}  \\
		\end{cases}
		&&  
		\hspace{1mm}\circled{2 }
		\begin{cases}
			\mu^{2}_{\theta_1}=\frac{1}{2} - \epsilon,\, \mu^{2}_{\theta_2} =\frac{1}{2} + \epsilon  \\
			u^{2}_{\theta_1}(a_1) =\frac{1}{\left(2-4\epsilon \right)} ,\, u^{2}_{\theta_2}(a_1) =\frac{1}{\left(10+20\epsilon \right)}  \\
			u^{2}_{\theta_1}(a_2) =\frac{1}{\left(10-20\epsilon \right)},\, u^{2}_{\theta_2}(a_2) =\frac{1}{\left(2+4\epsilon \right)}\\
			u^{2}_{\theta_1}(a_3) =\frac{3}{10},\, u^{2}_{\theta_2}(a_3) =\frac{3}{10}\\
			u^{2}_{\theta_1}(a_4) =0,\, u^{2}_{\theta_2}(a_4) = \frac{1}{\left(2+4\epsilon \right)}  \\
		\end{cases}
	\end{alignat*}
	with $\epsilon \in (0, \nicefrac{1}{4})$. With a simple calculation, we can show that, in both the two instances, for any signaling scheme $\phi$ employing a generic set of signals $\mathcal{S}$, the sender receives the following feedback:
	\begin{enumerate}
		\item	$\forall s \in \mathcal{S}$ such that $\phi_{\theta_1}(s)>  \phi_{\theta_2}(s)$, then $a^\phi(s) = a_1$. %$ b_{\xi^s}=a_1.$
		\item	$\forall s \in \mathcal{S}$ such that $0<\phi_{\theta_1}(s)<  \phi_{\theta_2}(s)$, then $a^\phi(s) = a_2$. %$ b_{\xi^s}=a_2.$
		\item	$\forall s \in \mathcal{S}$ such that $\phi_{\theta_1}(s)= \phi_{\theta_2}(s)$, then $a^\phi(s) = a_3$. %$ b_{\xi^s}=a_3.$
		\item	$\forall s \in \mathcal{S}$ such that $0=\phi_{\theta_1}(s)$, then $a^\phi(s) = a_4$. %$ b_{\xi^s}=a_4.$
	\end{enumerate}
%\begin{align*}
%	\begin{cases}
%			\forall s \in \mathcal{S} \,\, \textnormal{such that } \,\, \phi_{\theta_1}(s)>  \phi_{\theta_2}(s), \,\, \textnormal{then} \, b_{\xi^s}=a_1.\\
%			\forall s \in \mathcal{S} \,\, \textnormal{such that } \,\, 0<\phi_{\theta_1}(s)<  \phi_{\theta_2}(s), \,\, \textnormal{then} \,  b_{\xi^s}=a_2.\\
%			\forall s \in \mathcal{S} \,\, \textnormal{such that } \,\, \phi_{\theta_1}(s)= \phi_{\theta_2}(s), \,\, \textnormal{then} \,  b_{\xi^s}=a_3.\\
%			\forall s \in \mathcal{S} \,\, \textnormal{such that } \,\,  0=\phi_{\theta_1}(s), \, \textnormal{then} \,\,  b_{\xi^s}=a_4.\\
%	\end{cases}
%\end{align*}
	As a result, for any signaling scheme the sender may commit to, the resulting feedback in each signal of the signaling scheme is the same. Thus, we assume without loss of generality, that the sender only commits to signaling schemes that maximizes the probability of inducing actions $a_3$ or $a_4$. This is because, the sender does not gain any information by committing to one signaling scheme over another, while the signaling schemes that induce these two actions are the only ones that provide the sender with strictly positive expected utility.
	
	Furthermore, thanks to what we have observed before, we can restrict our attention to direct signaling schemes, \emph{i.e.}, those in which the set of signals coincides with the set of actions. Thus, at each round, we assume that the sender commits to a signaling scheme $\phi^t$ of the form:
	\begin{align}\label{eq:signaling_lower_bound}
		\phi^t\coloneqq \begin{cases}
				\phi^t_{\theta_1}(a_3)=	\phi^t_{\theta_2}(a_3) \coloneqq 	\phi^t_1, \,\,\   \\
				\phi^t_{\theta_1}(a_4)=0, \,	\phi^t_{\theta_2}(a_4)= 1- 	\phi^t_1, \,\,\   \\
				\phi^t_{\theta_1}(a_1)=1 - \phi^t_1,\,	\phi^t_{\theta_2}(a_1)= 0, \,\,\   \\
		\end{cases}
	\end{align}
	with $\phi_1^t \in [0,1]$. We also notice that, in each round, the optimal signaling scheme in the first instance is the one that induces action $a_3$, meaning $\phi^t_1=1$ for each $t \in [T]$. While the optimal signaling scheme in the second instance at each round is the one that reveals the state of nature to the receiver, meaning $\phi^t_1=0$ for each $t \in [T]$. In such a way, the learning task undertaken by the sender reduces to select a value of $\phi_1^t \in [0,1]$ for each $t \in [T]$ controlling the probability of inducing action $a_3$ over action $a_4$.
	
	In the following, we define $\mathbb{P}^{1}$ $(\mathbb{P}^{2})$ as the probability distribution generated by the execution of a given algorithm in the first (second) instance and we let $\mathbb{E}^{1}$ $(\mathbb{E}^{2})$ be  the expectation induced by such a distribution.
	
	The cumulative regret in the first instance can be written as follows:
%	\begin{align*}
%		R^{1}_T&=\frac{3}{10} \mathbb{E}^{1} \left[\sum_{t=1}^{T}\left( \frac{1}{2} + \epsilon -\phi^t_1 \left( \frac{1}{2}+ \epsilon \right) - \left(1-\phi^t_1\right) \left( \frac{1}{2}- \epsilon \right)  \right) \right]\\
%		&=\frac{3\epsilon}{5} \mathbb{E} \left[\sum_{t=1}^{T} \left(1-\phi_1^t\right) \right].
%	\end{align*}
	%\bollo{
	\begin{align*}
	R^{1}_T&= \mathbb{E}^{1} \left[\sum_{t=1}^{T}\left( \frac{1}{2} + \epsilon -\phi^t_1 \left( \frac{1}{2}+ \epsilon \right) - \left(1-\phi^t_1\right) \left( \frac{1}{2}- \epsilon \right)  \right) \right]\\
	&=2\epsilon \mathbb{E} \left[\sum_{t=1}^{T} \left(1-\phi_1^t\right) \right].
	\end{align*}
% }
	Similarly, in the second instance, the cumulative regret is given by:
	\begin{align*}
	R^{2}_T= 2\epsilon  \mathbb{E} \left[\sum_{t=1}^{T} \phi_1^t\right].
	\end{align*}
	% \bollo{$R^{2}_T= 2\epsilon  \mathbb{E} \left[\sum_{t=1}^{T} \phi_1^t\right].$}
	%
	Furthermore, it is easy to check that:
	\[R^1_T \ge \mathbb{P}^{1}\left(\sum_{t=1}^T \phi^1_t \le T/2\right) \epsilon T \quad \text{and} \quad R^2_T \ge \mathbb{P}^{2}\left(\sum_{t=1}^T \phi^1_t \ge T/2\right) \epsilon T.\]
	%\bollo{Senza 3/10} \bollo{Perchè compare $\pi$?}
	By employing the relative entropy identities divergence decomposition we also have that: 
	\begin{align*}
		\mathcal{KL} \left(\mathbb{P}^{1},\mathbb{P}^{2} \right) & =  T\cdot\mathcal{KL} \left(\mu^1, \mu^2 \right)\\
		& \le \frac{64}{3} T \epsilon^2 \le 22 T \epsilon^2,
	\end{align*}
	% \bollo{Mi sembra giusto, ma perchè $32$? A me basta $64/3<22$}
	where we employed the fact that for two Bernoulli distribution it holds 
	\[\mathcal{KL} (\textnormal{Be}(p), \textnormal{Be}(q)) \le \frac{(p-q)^2}{q(1-q)}.\]
	Then, by employing the Bretagnolle–Huber inequality we have that,
	\begin{align*}
		R_T^1 + R_T^2  & \ge  \epsilon T \left( \mathbb{P}^1 \left(\sum_{t=1}^T \phi^1_t \le T/2 \right)  + \mathbb{P}^2 \left(\sum_{t=1}^T \phi^1_t \ge T/2\right) \right)   \\
		& \ge\frac{1}{2} \epsilon T \exp{ \left( - \mathcal{KL} \left(\mathbb{P}^{1}, \mathbb{P}^{2} \right) \right) } \\
		& \ge  \frac{1}{2} \epsilon T \exp{ \left(- {22} T \epsilon^2 \right) }.
	\end{align*}
%	\bollo{Senza $3/10$ rimane $1/2$}\bollo{$T_1(T) \coloneqq \sum_{t=1}^{T} \phi_1^t$?}
	By taking $\epsilon = \sqrt{\frac{1}{22T }}$ we get:
	\begin{equation*}
		R_T^1 + R_T^2   \ge  C_1 \sqrt{T}.
	\end{equation*}
	 with $C_1=\nicefrac{e^{-1}}{(2 \sqrt{22} ) }$ %\bollo{$C_1=\nicefrac{e^{-1}}{(2 \sqrt 32) }$}.
	 Thus, we have:
	\begin{equation*}
		R_T^1 \ge \frac{C_1}{2} \sqrt{T}  \quad \vee \quad R_T^2 \ge \frac{C_1}{2} \sqrt{T},
	\end{equation*}
	concluding the proof.
\end{proof}

%% file: appendix_ub_sc.tex
%\section{Omitted proofs from Section~\ref{sec:sample_complexity}}\label{appendix:sample_complexity}
\section{Details and omitted proofs from Section~\ref{sec:sample_complexity}}\label{appendix:sample_complexity}\label{sec:appendix_sample_complexity}
%\bac{manca la roba dell $\epsilon$ con bit complexity limitata}
\subsection{\texttt{Compute-Threshold} procedure}
The \texttt{Compute-Threshold} procedure takes as input a real parameter $\epsilon_1 > 0$. Then, it iteratively halves the value of a different parameter $\epsilon$, initially set to one, until it is smaller than or equal to $\epsilon_1$. 
In this way, Algorithm~\ref{alg:compute_epsilon} computes a parameter $\epsilon \in [\epsilon_1/2, \epsilon_1]$ in $\mathcal{O}(\log(\nicefrac{1}{\epsilon_1}))$ rounds with bit complexity $B_\epsilon = \mathcal{O}(\log(\nicefrac{1}{\epsilon_1}))$. 
This technical component is necessary to ensure that the bit-complexity of the parameter $\epsilon$ is not too large while guaranteeing that the solution returned by Algorithm~\ref{alg:main_algorithm_sc} is still $\gamma$-optimal with probability at least $1-\eta$. %This is indeed possible since Algorithm~\ref{alg:compute_epsilon} requires $\mathcal{O}(\log(\nicefrac{1}{\epsilon_1}))$ rounds to compute a parameter $\epsilon$ with bit complexity $B_\epsilon = \mathcal{O}(\log(\nicefrac{1}{\epsilon_1}))$ that satisfies $\epsilon \in [\epsilon_1, 2\epsilon_1]\).

\begin{algorithm}[H]
	\caption{\texttt{Compute-Threshold}}\label{alg:compute_epsilon}
	\begin{algorithmic}[1]
		\Require $\epsilon_1 \in (0,1)$
		\State $\epsilon \gets 1$
		\While{$\epsilon \ge \epsilon_1$}
			\State $\epsilon \gets \epsilon/2$
		\EndWhile
		\State \textbf{Return} $\epsilon$
	\end{algorithmic}
\end{algorithm}
%We observe that, since $\epsilon \le \epsilon_1$, Algorithm~\ref{alg:main_algorithm_sc} still computes a $\gamma$-optimal solution with probability at least $\eta$.
%Furthermore, Algorithm~\ref{alg:compute_epsilon} requires $\mathcal{O}(\log(\nicefrac{1}{\epsilon_1}))$ rounds, and returns a value $\epsilon$ with bit-complexity $B_\epsilon = \mathcal{O}(\log(\nicefrac{1}{\epsilon_1}))$.

\subsection{Omitted proofs from Section~\ref{sec:sample_complexity}}
\PriorEstimateSC*
\begin{proof}
	%\bac{aggiornare statement}
	Thanks to the definition of $T_1 \coloneqq \left\lceil \nicefrac{1}{2\epsilon^2} \log\left(\nicefrac{2d}{\delta} \right) \right\rceil$ in Algorithm~\ref{alg:main_algorithm_sc} and employing both a union bound and the Hoeffding bound we have:
	\begin{align*}
			\mathbb{P}\left(|\widehat{\mu}_{\theta} - \mu_\theta| \le \epsilon\right) \ge 1 - \delta, \,\,\, \forall \theta \in \Theta.
	\end{align*} 
%	In particular, it holds that $\mathbb{P}\left(|\widehat{\mu}_{\theta} - \mu_\theta| \le \epsilon\right) \ge 1 - \delta$ for each $\theta \in \Theta$.
	Consequently, if $|\widehat{\mu}_\theta -\mu_\theta| \le \epsilon$ for each $\theta \in \Theta$, then for each $x \in \mathcal{X}_\epsilon$, the probability of inducing the slice $x$ can be lower bounded as follows:
	\begin{equation*}
		\epsilon \le
		\sum_{\theta  \in \Theta'} \widehat{\mu}_\theta x_\theta - \epsilon \le
		\sum_{\theta \in \Theta'}(\widehat{\mu}_\theta -\epsilon)x_\theta \le
		\sum_{\theta \in \Theta'} \mu_{\theta} x_\theta \le
		\sum_{\theta \in \Theta} \mu_\theta x_\theta 
	\end{equation*}
	where the above inequalities hold because each $x \in \mathcal{X}_\epsilon$ satisfies the constraint $\sum_{\theta \in \Theta'} \widehat{\mu}_\theta x_\theta \ge 2\epsilon$. Furthermore, if $|\widehat{\mu}_\theta -\mu_\theta| \le \epsilon$ for each $\theta \in \Theta$, then for each $x \not \in \mathcal{X}_\epsilon$ the following holds:
	\begin{equation}\label{eq:slice_sc_1}
		\sum_{\theta  \in \Theta'} {\mu}_\theta  x_\theta \le
		\epsilon+ \sum_{\theta  \in \Theta'} ({\mu}_\theta - \epsilon )x_\theta \le
		\epsilon +\sum_{\theta  \in \Theta'} \widehat{\mu}_\theta x_\theta \le
		3 \epsilon,
	\end{equation}
	since, if $x \not \in \mathcal{X}_\epsilon$, it holds $\sum_{\theta \in \Theta'} \widehat{\mu}_\theta x_\theta \le 2\epsilon$, and, 
	\begin{equation}\label{eq:slice_sc_2}
		\sum_{\theta  \not \in \Theta'} {\mu}_\theta  x_\theta  \le  \sum_{\theta  \not \in \Theta'} ({\mu}_\theta-\epsilon)  x_\theta + \epsilon \le \sum_{\theta  \not \in \Theta'} \widehat{\mu}_\theta  x_\theta + \epsilon \le 3 \epsilon.
		%\sum_{\theta \in \Theta} \mu_\theta x_\theta \ge \sum_{\theta \in \Theta'} \mu_{\theta} x_\theta \ge \sum_{\theta \in \Theta'}(\widehat{\mu}_\theta -\epsilon)x_\theta = \sum_{\theta  \in \Theta'} \widehat{\mu}_\theta x_\theta - \epsilon \ge \epsilon,
	\end{equation}
	Thus, by combining Inequality~\eqref{eq:slice_sc_1} and Inequality~\eqref{eq:slice_sc_2}, we have:
	\begin{equation*}
		\sum_{\theta  \in \Theta} {\mu}_\theta  x_\theta  \le  6 \epsilon,
		%\sum_{\theta \in \Theta} \mu_\theta x_\theta \ge \sum_{\theta \in \Theta'} \mu_{\theta} x_\theta \ge \sum_{\theta \in \Theta'}(\widehat{\mu}_\theta -\epsilon)x_\theta = \sum_{\theta  \in \Theta'} \widehat{\mu}_\theta x_\theta - \epsilon \ge \epsilon,
	\end{equation*}
	when $x \not \in \mathcal{X}_\epsilon$, concluding the proof.
\end{proof}

\MainSC*
\begin{proof}
	Thanks to Lemma~\ref{lem:prior_estimate_sc}, with probability at least $1-\delta = 1-\nicefrac{\eta}{2}$ Algorithm~\ref{alg:main_algorithm_sc} correctly completes Phase 1 in $T_1=\mathcal{O}(\nicefrac{1}{\epsilon^2}\log(\nicefrac{1}{\eta})\log(d))$ rounds.
	Thus, with probability at least $1-\nicefrac{\eta}{2}$, both the event $\mathcal{E}_1$ and the inequalities $|\widehat{\mu}_{\theta} - \mu_\theta| \le \epsilon$ for each $\theta \in \Theta$ hold.
%	As a result, with probability at least $1-\nicefrac{\eta}{2}$, it also holds that $|\widehat{\mu}_{t,\theta} - \mu_\theta| \le \epsilon$ for each $\theta \in \Theta$ and $t>T_1$, since $\widehat{\mu} \coloneqq \widehat{\mu}_{T_1+1}$.
%	Thus, the event $\mathcal{E}^{\text{SC}}_1$ occurs with probability at least $1-\nicefrac{\eta}{2}$.
%	Under the event $\mathcal{E}^\text{SC}_1$, for each $\theta \in \Theta$ it holds that $|\hat{\mu}_\theta - \mu_\theta| \le \epsilon$.
%	Furthermore, if $x \in \mathcal{X}_\epsilon$, then the probability of observing the signal $s_1$ in a single round is at least $\sum_{\theta \in \Theta}\mu_\theta x_\theta \ge \epsilon$.
	
	Consequently, under the event $\mathcal{E}_1$, with probability at least $1-\zeta = 1-\nicefrac{\eta}{2}$, Algorithm~\ref{alg:main_algorithm_sc} correctly partitions the search space $\mathcal{X}_\epsilon$ in at most: 
	\begin{equation*}
		\widetilde{\mathcal{O}}\left(\frac{n^2}{\epsilon} \log^2\left(\frac{1}{\eta}\right) \left(d^7L +\binom{d+n}{d}\right)\right)
	\end{equation*}
	rounds, as stated by Lemma~\ref{lem:final_partition}.
	Furthermore, we notice that $L= B +B_\epsilon +B_{\widehat{\mu}}$, with: 
	\begin{equation*}
		B_{\widehat{\mu}} = \mathcal{O}(\log(T_1)) = \mathcal{O}\left(\log(\nicefrac{1}{\epsilon}) +\log(\nicefrac{1}{\eta}) +\log(d)\right).
	\end{equation*}
	As a result, with probability at least $1-\eta$, Algorithm~\ref{alg:main_algorithm_sc} correctly terminates in a number of rounds $N$ which can be upper bounded as follows:
	\begin{equation*}
		N \le \widetilde{\mathcal{O}}\left(\frac{1}{\epsilon^2}\log\left(\frac{1}{\eta}\right)\log(d) +\frac{n^2}{\epsilon} \log^2\left(\frac{1}{\eta}\right) \left(d^7(B +B_\epsilon) +\binom{d+n}{d}\right)\right).
	\end{equation*}
	Furthermore, we observe that if $|\widehat{\mu}_{\theta} - \mu_\theta| \le \epsilon$ for each $\theta \in \Theta$, then the following holds: %for each $t > T_1$:
	\begin{equation*}
		\left| \sum_{\theta  \in \Theta}\widehat{\mu}_{\theta} - \mu_\theta \right| \le  \sum_{\theta  \in \Theta}|\widehat{\mu}_{\theta} - \mu_\theta| \le \sum_{\theta  \in \Theta} \epsilon = \epsilon d.
	\end{equation*}
	Consequently, thanks to the result provided by Lemma~\ref{lem:find_signaling}, with probability at least $1-\eta$, Algorithm~\ref{alg:main_algorithm_sc} computes a $12n\epsilon d$-optimal solution.
%	Thus, by setting $\epsilon \coloneqq \nicefrac{\gamma}{12nd}$, with probability at least $1-\eta$ Algorithm~\ref{alg:main_algorithm_sc} computes a $\gamma$-optimal solution in a number of rounds $N$ bounded by:
	Thus, by setting $\epsilon_1 \coloneqq \nicefrac{\gamma}{12nd}$ and $\epsilon \le \epsilon_1$, with probability at least $1-\eta$ Algorithm~\ref{alg:main_algorithm_sc} computes a $\gamma$-optimal solution in a number of rounds $N$ bounded by:
	\begin{align*}
		N &\le \widetilde{\mathcal{O}}\left( \frac{n^2 d^2}{\gamma^2}\log\left(\frac{1}{\eta}\right) +\frac{n^3}{\gamma} \log^2\left(\frac{1}{\eta}\right) \left(d^8(B +B_\epsilon) +d\binom{d+n}{d}\right) \right) \\
		&= \widetilde{\mathcal{O}}\left( \frac{n^3}{\gamma^2} \log^2\left(\frac{1}{\eta}\right) \left(d^8(B +B_\epsilon) +d\binom{d+n}{d}\right) \right) \\
		&= \widetilde{\mathcal{O}}\left( \frac{n^3}{\gamma^2} \log^2\left(\frac{1}{\eta}\right) \left(d^8B +d\binom{d+n}{d}\right) \right),
	\end{align*}
	where the last equality holds because the bit-complexity of $\epsilon$ is  $B_\epsilon = \mathcal{O}(\log(\nicefrac{nd}{\gamma}))$, concluding the proof.
%	Furthermore, Under the events $\mathcal{E}^\text{SC}_1$ and $\mathcal{E}_2$, then it computes a $\dots$-optimal solution, as stated by Lemma~\ref{...}.
%	Since $\epsilon = \dots$ \bollo{($\epsilon = $ quel che serve per avere $\gamma$-optimal)} and \bollo{Risultato che serve dalla Phase 1}, the solution is $\gamma$-optimal and the number of rounds required by the algorithm is: $\dots$.
\end{proof}

\hardnessfirstsc*
\begin{proof}
	In the proof of Theorem~\ref{thm:hardness1} we showed that, with probability $\nicefrac{3}{4}$, in $N=\lfloor \nicefrac{\mathcal{|P|}}{4}\rfloor$ rounds any algorithm does not correctly identify the posterior inducing action $a_{d+1}$. This is because, any deterministic algorithm can identify the optimal posterior only in $\lfloor\nicefrac{\mathcal{|P|}}{4}\rfloor$ instances, as observed in the proof of Theorem~\ref{thm:hardness1}.
	
	As a result, in the remaining $|\preg|-N$ instances, any deterministic algorithm will receive the same feedback and thus will always output the same posterior after $N$ rounds, which will result in the optimal one in only a single instance.
	
	Thus, thanks to the Yao's minimax principle, there is no algorithm that is guaranteed to return an optimal solution with probability at least: 
	\begin{equation*}
		\frac{3}{4} \left(\frac{|\preg| -N -1 }{|\preg| -N }\right)\ge \frac{3}{8} \left(\frac{|\preg| -N }{|\preg| -N }\right) = \frac{3}{8}.
	\end{equation*}
%	a solution which is $1/2$-optimal in $\lfloor|\preg|/4 \rfloor= 2^{\Omega(d)}$ rounds.
	Finally, we observe that $\textnormal{OPT}=\nicefrac{1}{2}$, while any algorithm that does not induce the posterior $\xi'$ provides an expected utility equal to zero.
	As a result, for each $\kappa<\nicefrac{1}{2}$ there is no algorithm that is guaranteed to return a solution which is $\kappa$-optimal in $\lfloor|\preg|/4 \rfloor= 2^{\Omega(d)}$ rounds with probability at least $\nicefrac{3}{8}$.
	%Thus, after $N$ rounds, any algorithm does not correctly identify the optimal posterior with probability $\nicefrac{3}{4}$. Consequently, among the $\nicefrac{3}{4}|P|$ instances in which the algorithm does not identify the posterior inducing action $a_{n+1}$, there exists a single instance in which the algorithm will return the optimal posterior, and thus, the optimal solution. Thus, by Yao theorem, there is no algorithm that is  to return with probability: $$3/4 (3|P|-1)/(3|P|)\ge \dots$$, a solution which is $1/2$-optimal.
\end{proof}

\hardnessthirssc*
\begin{proof}
	To prove the theorem, we consider the same instance and the same definitions introduced in the proof of Theorem~\ref{thm:hardness3}. In this case, we let $\mathbb{P}^{1}$ $(\mathbb{P}^{2})$ be the probability distribution generated by the execution of a given algorithm in the first (second) instance for $N=\lceil \nicefrac{ \log(1/ 4\eta)}{22 \epsilon^2} \rceil$ rounds. Furthermore, we introduce the event $\mathcal{E}$, under which the signaling scheme returned at the round $N$, according to the definition presented in Equation~\ref{eq:signaling_lower_bound}, is such that $\phi^N \le 1/2$. We notice that, if such signaling scheme is such that $\phi^N < 1/2$, then the sender’s expected utility in the first instance is smaller or equal to $1/2$, thus being $\epsilon/2$-optimal. At the same time, if $\phi^N \ge 1/2$ in the second instance, then the solution returned by the algorithm is not $\epsilon/2$-optimal.
	
	Thus, by employing the Bretagnolle–Huber inequality we have:
	\[ \mathbb{P}^{1}\left( \mathcal{E} \right)  + \mathbb{P}^{2}\left( \mathcal{E}^{C} \right) \ge \frac{1}{2}\exp{ \left( - \mathcal{KL} \left(\mathbb{P}^{1}, \mathbb{P}^{2} \right) \right) } \ge \frac{1}{2} \exp{ \left(- {22} N \epsilon^2 \right) }, \]
	since $	\mathcal{KL} \left(\mathbb{P}^{1},\mathbb{P}^{2} \right) \le 22 N \epsilon^2$, as observed in the proof of Theorem~\ref{thm:hardness3_sc}. Finally, by employing the definition of $N$, we have:
	\[ \mathbb{P}^{1}\left( \mathcal{E} \right)  \ge \eta \,\,\ \vee \,\,\, \mathbb{P}^{2}\left( \mathcal{E}^{C} \right)  \ge \eta.\]
	As a result, by setting $2\gamma = \epsilon$, the statement of the lemma holds. 
\end{proof}

%% file: appendix_lb_sl.tex
\section{ Sample complexity with known prior}\label{appendix:sample_comlexity_known_mu}
%
% \bac{aggioranre nomi algoritmi}
%\bac{aggiornare con altri parametri}
In this section, we discuss the \emph{Bayesian persuasion PAC-learning problem} when the prior distribution $\mu$ is known to the sender. To tackle the problem, we propose Algorithm~\ref{alg:main_algorithm_sc_known}. The main difference with respect to Algorithm~\ref{alg:main_algorithm_sc} is that, in this case, we do not need to employ the~\texttt{Build-Search-Space} procedure, as the prior is already known to the sender.
This allows us to compute a $\gamma$-optimal signaling scheme in only $\mathcal{O}(\nicefrac{1}{\gamma})$ rounds, instead of $\mathcal{O}(\nicefrac{1}{\gamma^2})$ rounds as in the case with an unknown prior.
% We also notice that, differently from Algorithm~\ref{alg:main_algorithm_sc_known}, the parameter $\epsilon$ is set equal to 
%Furthermore, the definition of $\epsilon$ is different with respect to Algorithm~\ref{alg:main_algorithm_sc}.
%However, we still apply the approach described in Section~\ref{sec:appendix_sample_complexity} to compute a value of $\epsilon$ with limited bit-complexity.
%In particular, we perform a binary search between $0$ and $1$ looking for a value of $\epsilon$ between $\nicefrac{\epsilon_1}{2}$ and $\epsilon_1$, with $\epsilon_1 \coloneqq \nicefrac{\gamma}{10nd}$.

\begin{algorithm}[H]
	\caption{\texttt{PAC-Persuasion-w/o-Clue-Known}}\label{alg:main_algorithm_sc_known}
	\begin{algorithmic}[1]
%		\State \bac{$\epsilon??$}
		% \State \bac{settare parametri}
		\Require $\eta \in (0,1), \gamma \in (0,1), \mu \in \Delta_\Theta$
		\State $\epsilon_1 \gets \nicefrac{\gamma}{10nd}$
		\State $\epsilon \gets \texttt{Compute-Epsilon}(\epsilon_1),$ $\widehat{\mu} \gets \mu$
		\State $\widetilde{\Theta} \gets \{\theta \in \Theta \mid \widehat{\mu}_{\theta} > 2\epsilon\}$
		\State $\mathcal{X}_\epsilon \gets \left\{ x \in \mathcal{X} \mid \sum_{\theta  \in \widetilde{\Theta}} \widehat{\mu}_{\theta}x_\theta \ge 2\epsilon \right\}$
		\State $ \mathcal{R}_\epsilon \gets \texttt{Find-Polytopes}(\mathcal{X}_\epsilon, \eta)$ %\Comment{\textcolor{gray}{Phase 2}} 
	% \State \bac{aggiorna nomi}
		\State $\phi \gets \texttt{Compute-Signaling} (\mathcal{R}_\epsilon,\mathcal{X}_\epsilon,  \mu ) $  		%\Comment{\textcolor{gray}{Phase 3}} algoritmo 13 al posto di \mathcal{Y} dovrebbe andarci \mathcal{R}_\epsilon
		\State \textbf{return} $\phi $
	\end{algorithmic}
\end{algorithm}

In this case, the following theorem holds.

\begin{restatable}{theorem}{MainSCKnown}\label{th:main_th_Sc_known} 
	With probability at least $1-\eta$ and in Algorithm~\ref{alg:estimate_prior} computes a $\gamma$-optimal signaling scheme in $\widetilde{\mathcal{O}}\left( \nicefrac{n^3}{\gamma} \log^2\left(\nicefrac{1}{\eta}\right) \left(d^8B +d\binom{d+n}{d}\right) \right)$ rounds.
\end{restatable}
%\bollo{$n^3$ giusto?}
\begin{proof}
	Since $\widehat{\mu} = \mu$, the clean event  $\mathcal{E}_1$ holds with probability one. Consequently, thanks to Lemma~\ref{lem:final_partition}, with probability at least $1-\eta$, Algorithm~\ref{alg:main_algorithm_sc_known} correctly partitions the search space $\mathcal{X}_\epsilon$ in at most:
	\begin{equation*}
		\widetilde{\mathcal{O}}\left( \frac{n^2}{\epsilon} \log^2\left(\nicefrac{1}{\eta}\right) \left(d^7L +\binom{d+n}{d}\right) \right)
	\end{equation*}
	rounds, with $L \coloneqq B+B_\epsilon+B_{\widehat{\mu}}$. According to Algorithm~\ref{alg:main_algorithm_sc_known}, we have $\epsilon \le \epsilon_1 \coloneqq \nicefrac{\gamma}{(10nd)}$.
	%\bac{$\epsilon \ge \epsilon_1/2 = \nicefrac{\gamma}{(20nd)}$ se vogliamo ottenere quel bound ?}
	%\bollo{Perchè ci serve anche il lower bound? $\epsilon \le \nicefrac{\gamma}{(10nd)} \implies 10 \epsilon n d \le \gamma$}
	As a result, $L=\mathcal{O}\left(B +\log(nd)+\log(\nicefrac{1}{\gamma})\right)$, since $B_{\widehat{\mu}} \le B$ and $B_\epsilon = \mathcal{O}(\log(\nicefrac{1}{\epsilon_1})) =   \mathcal{O}(\log(nd)+\log(\nicefrac{1}{\gamma}))$.%\footnote{While in general the bit-complexity of $\epsilon$ can be arbitrary, it is possible to perform a binary search between $0$ and $1$ to find some $\epsilon'$ between $\epsilon/2$ and $\epsilon$ with bit complexity bounded by $\mathcal{O}(\log(\nicefrac{1}{\epsilon}))$.}
	% Given the definitions of $\mathcal{X}_\epsilon \coloneqq \Delta_{d} \cap \left\{ x \in \mathbb{R}^{d} \mid \sum_{\theta  \in \Theta'} \widehat{\mu}_{\theta}x_\theta \ge 2\epsilon \right\}$ and $\Theta' = \{\theta \in \Theta \mid \widehat{\mu}_{\theta} > 2\epsilon\}$, it obviously holds that the event $\mathcal{E}_1$ happens with probability one, since $\widehat{\mu} = \mu$.
	%Consequently, thanks to Lemma~\ref{lem:final_partition}, with probability at least $1-\zeta$ Algorithm~\ref{alg:main_algorithm_sc_known} correctly partition the search space $\mathcal{X}_\epsilon$ in at most:
	%\begin{equation*}
	%	N= \widetilde{\mathcal{O}}\left( \frac{n^2}{\epsilon} \log^2\left(\frac{1}{\zeta}\right) \left(d^7L +\binom{d+n}{d}\right) \right)
	%\end{equation*}
	%	rounds, with $L=B+B_\epsilon+B_{\widehat{\mu}}$.
	%	In the following we consider $\epsilon = \nicefrac{\gamma}{10d}$. 
	%	Thus, it holds that $L=\mathcal{O}(B +\log(d)\log(\nicefrac{1}{\gamma}))$,
	%	since $B_{\widehat{\mu}} \le B$ and $B_\epsilon = \mathcal{O}(\log(\nicefrac{1}{\epsilon})) = \mathcal{\log(\nicefrac{10d}{\gamma})}$\footnote{While in general the bit-complexity of $\epsilon$ can be arbitrary, it is possible to perform a binary search between $0$ and $1$ to find some $\epsilon'$ between $\epsilon/2$ and $\epsilon$ with bit complexity bounded by $\mathcal{O}(\log(\nicefrac{1}{\epsilon}))$.}.
	
	Furthermore, under the event $\mathcal{E}_2$, %which corresponds to the event in which the search space has been correctly partitioned,
	Algorithm~\ref{alg:main_algorithm_sc_known} computes a $10\epsilon n d$-optimal solution, as guaranteed by Lemma~\ref{lem:find_signaling}, with $\widehat{\mu} = \mu$.
	
	Thus, with probability at least $1-\eta$,  Algorithm~\ref{alg:main_algorithm_sc_known} computes a $\gamma$-optimal solution in:
	\begin{equation*}
		\widetilde{\mathcal{O}}\left( \frac{n^3}{\gamma} \log^2\left(\nicefrac{1}{\eta}\right) \left(d^8B +d\binom{d+n}{d}\right) \right)
	\end{equation*}
	rounds, concluding the proof.
\end{proof}

We notice that, differently from the case with an unknown prior, it is possible to achieve a $\mathcal{O}\left(\nicefrac{\log(1/\eta)}{\gamma}\right)$ upper bound with respect to the input parameters $\gamma, \eta > 0$. Finally, we show that such a dependence is tight, as shown in the following theorem.
\begin{restatable}{theorem}{hardnessecondsc}\label{thm:hardness2_sc}
	Given $\gamma,\eta > 0$ no algorithm is guaranteed to return an $\gamma$-optimal signaling scheme with probability of at least $1 - \eta$ employing $ \Omega \left( \nicefrac{\log(1/\eta)}{\gamma}\right)   $ rounds, even when the prior distribution is known to the sender.
\end{restatable}
\begin{proof}	
	We consider two instances characterized by two states of nature and three receiver's actions. The two instances share the same prior distribution, defined as $\mu_{\theta_1}=4\gamma$ and $\mu_{\theta_2}=1-4\gamma$. In both the instances the sender's utility is given by $u^\text{s}_\theta(a_{1})=0$,  $u^\text{s}_\theta(a_{2})=1/2$ and $u^\text{s}_\theta(a_{3})=1$ for all $\theta \in \Theta$. Furthermore, we assume that the receiver's utility in the two instances are given by:
	\begin{alignat*}{2}
		&\circled{1 }
		\begin{cases} u_{\theta_1}(a_1) =1,\, u_{\theta_2}(a_1) =1/2  \\
			u_{\theta_1}(a_2) =1/2,\, u_{\theta_2}(a_2) =1 \\
			u_{\theta_1}(a_3) =1,\, u_{\theta_2}(a_3) =0  \\
		\end{cases}
		&&  
		\hspace{15mm}\circled{2 }
		\begin{cases} u_{\theta_1}(a_1) =1,\, u_{\theta_2}(a_1) =1/2  \\
			u_{\theta_1}(a_2) =1/2,\, u_{\theta_2}(a_2) =1 \\
			u_{\theta_1}(a_3) =1/2,\, u_{\theta_2}(a_3) =0  \\
		\end{cases}
	\end{alignat*}
	We observe that the only case in which the sender receives different feedback in the two instances is when they induce the posterior distribution $\xi_1:=(1,0)$. This is because, when the sender induces $\xi_1$ in the first instance, the receiver plays the action $a_3 \in \mathcal{A}$, breaking ties in favor of the sender, while in the second instance, the receiver plays the action $a_1 \in \mathcal{A}$. Such a posterior can be induced, in both the two instances, with a probability of at most $\gamma$ to be consistent with the prior.

	We also observe that in the first instance the optimal sender's signaling scheme $\gamma$ is such that $\gamma(\xi_1) =4 \gamma$ and $\gamma(\xi_2) = 1 - 4\gamma$, where we let $\xi_2 := (0,1)$. Furthermore, the sender's expected utility in $\gamma$ is equal to $({1+4\gamma})/{2}$. In the second instance, the optimal sender's signaling scheme $\gamma$ is such that $\gamma(\xi_2) = 1 - 8\gamma$ and $\gamma(\xi_3) = 8\gamma$, with $\xi_3 := (1/2,1/2)$. It is easy to verify that such a signaling scheme achieves an expected utility of $1/{2}$.
	
	%equal to $({1+\gamma})/{2}$, meanwhile the optimal expected sender's utility in the second instance is equal to $({1+\gamma})/{2}$. 

	%To achieve this utility we define $\xi_1=(1,0)$ and $\xi_2=(0,1)$, the optimal signaling schem $\gamma$ in the secnd instance is such that $\gamma(\xi_1)=\gamma$  and $\gamma(\xi_2)=1- \gamma$. 
	
	In the following, we let $\mathbb{P}^1$ and $\mathbb{P}^2$ be the probability measures induced by the interconnection of a given algorithm executed in the first and in the second instances, respectively.
	%Equivalently, we define $\mathbb{E}^1$ and $\mathbb{E}^2$ the expecations induced by $\mathbb{P}^1$ and $\mathbb{P}^2$ respectiely.
	Furthermore, we introduce the event $E_N$, under which, during the first $N$ rounds, the sender never observes the action $a_3 \in \mathcal{A}$. It is easy to verify that such an event holds with a probability of at least $\mathbb{P}^1(E_N) \ge (1-4\gamma)^N$ in the first instance. This is because, at each round, the action $a_3$ can be observed with a probability of at most $4\gamma$. In contrast, since in the second instance the receiver never plays the action $a_3$, it holds $\mathbb{P}^2(E_N)=1 \ge (1-4\gamma)^N$.
	
	Then, by letting $\gamma_1^N$ ($\gamma_2^N$) be the signaling scheme returned after $N$ rounds in the first (second) instance, we have: 
	\begin{align}
		\mathbb{P}^1\left( \gamma_1^{N}(\xi_1) \ge 2\gamma \right) & + \mathbb{P}^2\left( \gamma_2^{N}(\xi_1) \le 2\gamma \right) \ge \mathbb{P}^1\left( \gamma_1^{N}(\xi_1) \ge 2\gamma , E_N \right) + \mathbb{P}^2\left( \gamma_2^{N}(\xi_1) \le 2\gamma , E_N \right) \nonumber \\
		& \ge \mathbb{P}^1\left( \gamma_1^{N}(\xi_1) \ge 2\gamma \, | E_N \right)  \mathbb{P}^1\left(E_N \right) + \mathbb{P}^2\left( \gamma_2^{N}(\xi_1) \le 2\gamma \, | E_N \right)   \mathbb{P}^2\left(E_N \right) \nonumber \\
		& \ge \left(\mathbb{P}^1\left( \gamma_1^{N}(\xi_1) \ge 2\gamma \, | E_N \right) + \mathbb{P}^2\left( \gamma_2^{N}(\xi_1) \le 2\gamma \,| E_N \right)  
		\right) \mathbb{P}^1\left(E_N \right) \nonumber\\
		& = \mathbb{P}^1\left(E_N \right) \nonumber \\
		& \ge (1-4\gamma)^N \ge 2 \eta. \label{eq:geo_gebra}
	\end{align}
	The above result holds observing that, under the event $E_N$, the behaviour of any algorithm working in the first instance coincides with the behaviour of the same algorithm working in the second one, as they receive the same feedback. Furthermore, Inequality~\ref{eq:geo_gebra} holds when $N$ is such that:
	\[
	N  \le \frac{\log(1/ 2\eta)}{10 \gamma} \le \frac{\log(2\eta)}{\log(1 - 4 \gamma)},
	\]
	if $\gamma \le 1/5$. % \bollo{Stiamo definendo $N  \coloneqq \nicefrac{\log(1/ 2\eta)}{10 \gamma} = \mathcal{O}(\nicefrac{\log(1/ \eta)}{\gamma})$?}
	
	Finally, we observe that if $\gamma_1^N(\xi_1) \le 2\gamma$, then the sender's expected utility in the first instance is of most $1/2 + \gamma$. This is because, for any signaling scheme $\gamma_1^N$, we have:
	\begin{equation*}
		u^s(\gamma_1^N) \le  \gamma_1^N(\xi_1) u^s(\xi_1) + \frac{1}{2}\left(  1- \gamma_1^N(\xi_1)  \right) = \frac{1}{2} + \gamma .
	\end{equation*}
	Thus, if $\gamma_1^N(\xi_1) \le 2\gamma$ the the signaling scheme $\gamma_1^N$ is at most $\gamma$-optimal.
	
	Equivalently, if the sender's final signaling scheme $\gamma_2^N$ in the second instance is such that $\gamma_2^N(\xi_1) \ge 2\gamma$, then the sender's utility in the second instance is of at most $1/2 -\gamma$. Thus, if $\gamma_2^N(\xi_1) \ge 2\gamma$ the the signaling scheme $\gamma_2^N$ is at most $\gamma$-optimal. 
	Then, we either have:
	\begin{equation*}
		\mathbb{P}^1\left( \gamma_1^{N}(\xi_1) \le 2\gamma \right) \ge \eta \,\,\,\ \textnormal{or} \,\,\,\ \mathbb{P}^2\left( \gamma_2^{N}(\xi_1) \ge 2\gamma \right) \ge \eta,		
	\end{equation*}
	showing that if the number of rounds $ N\le {\log(1/2\eta)}/{(10 \gamma)}$, there exists an instance such that no algorithm is guaranteed to return a $\gamma$-optimal signaling scheme with probability greater than or equal to $1-\eta$.
\end{proof}

% Algorithm~\ref{alg:main_algorithm_sc_known} takes in input a probability parameter $\zeta \in (0,1)$, another parameter $\gamma \in (0,1)$, and the true prior $\mu$, and with probability at least $1-\eta$ it computes a $\gamma$-optimal signaling scheme $\phi$.
%As a first step, it approximates $\gamma$ with the closest smaller rational number representable with at most $B_\mu +B_u$ bits (Line~\ref{line:approx_gamma} Algorithm~\ref{alg:main_algorithm_sc_known}).
%This passage reduces the asymptotic number of samples required by the $\texttt{Find-Partition}$ procedure.
%Subsequently, it computes the search space $\mathcal{X}_\epsilon$ and partitions it by means of the $\texttt{Find-Partition}$ procedure.
%The algorithm computes the search space without executing the $\texttt{Find-Search-Space}$ procedure, since it does not need to estimate the prior $\mu$.
%Thus, it partitions the search space by means of the $\texttt{Find-Partition}$ procedure.
% Finally, it employs the $\texttt{Find-Signaling-Scheme}$ procedure to compute a $\gamma$-optimal signaling scheme.

%% file: checklist.tex
\section*{NeurIPS Paper Checklist}

\begin{enumerate}
	
	\item {\bf Claims}
	\item[] Question: Do the main claims made in the abstract and introduction accurately reflect the paper's contributions and scope?
	\item[] Answer: \answerYes{} %, \answerNo{}, or \answerNA{}.
	\item[] Justification: The abstract and the introduction state all the main contributions of this work. 
	\item[] Guidelines:
	\begin{itemize}
		\item The answer NA means that the abstract and introduction do not include the claims made in the paper.
		\item The abstract and/or introduction should clearly state the claims made, including the contributions made in the paper and important assumptions and limitations. A No or NA answer to this question will not be perceived well by the reviewers. 
		\item The claims made should match theoretical and experimental results, and reflect how much the results can be expected to generalize to other settings. 
		\item It is fine to include aspirational goals as motivation as long as it is clear that these goals are not attained by the paper. 
	\end{itemize}
	
	\item {\bf Limitations}
	\item[] Question: Does the paper discuss the limitations of the work performed by the authors?
	\item[] Answer: \answerYes{} % Replace by \answerYes{}, \answerNo{}, or \answerNA{}.
	\item[] Justification: All the assumptions are stated in the paper.
	\item[] Guidelines:
	\begin{itemize}
		\item The answer NA means that the paper has no limitation while the answer No means that the paper has limitations, but those are not discussed in the paper. 
		\item The authors are encouraged to create a separate "Limitations" section in their paper.
		\item The paper should point out any strong assumptions and how robust the results are to violations of these assumptions (e.g., independence assumptions, noiseless settings, model well-specification, asymptotic approximations only holding locally). The authors should reflect on how these assumptions might be violated in practice and what the implications would be.
		\item The authors should reflect on the scope of the claims made, e.g., if the approach was only tested on a few datasets or with a few runs. In general, empirical results often depend on implicit assumptions, which should be articulated.
		\item The authors should reflect on the factors that influence the performance of the approach. For example, a facial recognition algorithm may perform poorly when image resolution is low or images are taken in low lighting. Or a speech-to-text system might not be used reliably to provide closed captions for online lectures because it fails to handle technical jargon.
		\item The authors should discuss the computational efficiency of the proposed algorithms and how they scale with dataset size.
		\item If applicable, the authors should discuss possible limitations of their approach to address problems of privacy and fairness.
		\item While the authors might fear that complete honesty about limitations might be used by reviewers as grounds for rejection, a worse outcome might be that reviewers discover limitations that aren't acknowledged in the paper. The authors should use their best judgment and recognize that individual actions in favor of transparency play an important role in developing norms that preserve the integrity of the community. Reviewers will be specifically instructed to not penalize honesty concerning limitations.
	\end{itemize}
	
	\item {\bf Theory Assumptions and Proofs}
	\item[] Question: For each theoretical result, does the paper provide the full set of assumptions and a complete (and correct) proof?
	\item[] Answer: \answerYes{} % Replace by \answerYes{}, \answerNo{}, or \answerNA{}.
	\item[] Justification: The assumptions needed are reported in the statements of the theorems and lemmas, while all the proofs are reported in the appendices. % \justificationTODO{}
	\item[] Guidelines:
	\begin{itemize}
		\item The answer NA means that the paper does not include theoretical results. 
		\item All the theorems, formulas, and proofs in the paper should be numbered and cross-referenced.
		\item All assumptions should be clearly stated or referenced in the statement of any theorems.
		\item The proofs can either appear in the main paper or the supplemental material, but if they appear in the supplemental material, the authors are encouraged to provide a short proof sketch to provide intuition. 
		\item Inversely, any informal proof provided in the core of the paper should be complemented by formal proofs provided in appendix or supplemental material.
		\item Theorems and Lemmas that the proof relies upon should be properly referenced. 
	\end{itemize}
	
	\item {\bf Experimental Result Reproducibility}
	\item[] Question: Does the paper fully disclose all the information needed to reproduce the main experimental results of the paper to the extent that it affects the main claims and/or conclusions of the paper (regardless of whether the code and data are provided or not)?
	\item[] Answer: \answerNA{} %\answerTODO{} % Replace by \answerYes{}, \answerNo{}, or \answerNA{}.
	\item[] Justification: The paper does not include experiments.%\justificationTODO{}
	\item[] Guidelines:
	\begin{itemize}
		\item The answer NA means that the paper does not include experiments.
		\item If the paper includes experiments, a No answer to this question will not be perceived well by the reviewers: Making the paper reproducible is important, regardless of whether the code and data are provided or not.
		\item If the contribution is a dataset and/or model, the authors should describe the steps taken to make their results reproducible or verifiable. 
		\item Depending on the contribution, reproducibility can be accomplished in various ways. For example, if the contribution is a novel architecture, describing the architecture fully might suffice, or if the contribution is a specific model and empirical evaluation, it may be necessary to either make it possible for others to replicate the model with the same dataset, or provide access to the model. In general. releasing code and data is often one good way to accomplish this, but reproducibility can also be provided via detailed instructions for how to replicate the results, access to a hosted model (e.g., in the case of a large language model), releasing of a model checkpoint, or other means that are appropriate to the research performed.
		\item While NeurIPS does not require releasing code, the conference does require all submissions to provide some reasonable avenue for reproducibility, which may depend on the nature of the contribution. For example
		\begin{enumerate}
			\item If the contribution is primarily a new algorithm, the paper should make it clear how to reproduce that algorithm.
			\item If the contribution is primarily a new model architecture, the paper should describe the architecture clearly and fully.
			\item If the contribution is a new model (e.g., a large language model), then there should either be a way to access this model for reproducing the results or a way to reproduce the model (e.g., with an open-source dataset or instructions for how to construct the dataset).
			\item We recognize that reproducibility may be tricky in some cases, in which case authors are welcome to describe the particular way they provide for reproducibility. In the case of closed-source models, it may be that access to the model is limited in some way (e.g., to registered users), but it should be possible for other researchers to have some path to reproducing or verifying the results.
		\end{enumerate}
	\end{itemize}

	\item {\bf Open access to data and code}
	\item[] Question: Does the paper provide open access to the data and code, with sufficient instructions to faithfully reproduce the main experimental results, as described in supplemental material?
	\item[] Answer:\answerNA{} % \answerTODO{} % Replace by \answerYes{}, \answerNo{}, or \answerNA{}.
	\item[] Justification: The paper does not include experiments. %\answerNA{} %\justificationTODO{}
	\item[] Guidelines:
	\begin{itemize}
		\item The answer NA means that paper does not include experiments requiring code.
		\item Please see the NeurIPS code and data submission guidelines (\url{https://nips.cc/public/guides/CodeSubmissionPolicy}) for more details.
		\item While we encourage the release of code and data, we understand that this might not be possible, so “No” is an acceptable answer. Papers cannot be rejected simply for not including code, unless this is central to the contribution (e.g., for a new open-source benchmark).
		\item The instructions should contain the exact command and environment needed to run to reproduce the results. See the NeurIPS code and data submission guidelines (\url{https://nips.cc/public/guides/CodeSubmissionPolicy}) for more details.
		\item The authors should provide instructions on data access and preparation, including how to access the raw data, preprocessed data, intermediate data, and generated data, etc.
		\item The authors should provide scripts to reproduce all experimental results for the new proposed method and baselines. If only a subset of experiments are reproducible, they should state which ones are omitted from the script and why.
		\item At submission time, to preserve anonymity, the authors should release anonymized versions (if applicable).
		\item Providing as much information as possible in supplemental material (appended to the paper) is recommended, but including URLs to data and code is permitted.
	\end{itemize}

	\item {\bf Experimental Setting/Details}
	\item[] Question: Does the paper specify all the training and test details (e.g., data splits, hyperparameters, how they were chosen, type of optimizer, etc.) necessary to understand the results?
	\item[] Answer: \answerNA{} %\answerTODO{} % Replace by \answerYes{}, \answerNo{}, or \answerNA{}.
	\item[] Justification: The paper does not include experiments. % \answerNA{} % \justificationTODO{}
	\item[] Guidelines:
	\begin{itemize}
		\item The answer NA means that the paper does not include experiments.
		\item The experimental setting should be presented in the core of the paper to a level of detail that is necessary to appreciate the results and make sense of them.
		\item The full details can be provided either with the code, in appendix, or as supplemental material.
	\end{itemize}
	
	\item {\bf Experiment Statistical Significance}
	\item[] Question: Does the paper report error bars suitably and correctly defined or other appropriate information about the statistical significance of the experiments?
	\item[] Answer: \answerNA{} % \answerTODO{} % Replace by \answerYes{}, \answerNo{}, or \answerNA{}.
	\item[] Justification: The paper does not include experiments. % \answerNA{} % \justificationTODO{}
	\item[] Guidelines:
	\begin{itemize}
		\item The answer NA means that the paper does not include experiments.
		\item The authors should answer "Yes" if the results are accompanied by error bars, confidence intervals, or statistical significance tests, at least for the experiments that support the main claims of the paper.
		\item The factors of variability that the error bars are capturing should be clearly stated (for example, train/test split, initialization, random drawing of some parameter, or overall run with given experimental conditions).
		\item The method for calculating the error bars should be explained (closed form formula, call to a library function, bootstrap, etc.)
		\item The assumptions made should be given (e.g., Normally distributed errors).
		\item It should be clear whether the error bar is the standard deviation or the standard error of the mean.
		\item It is OK to report 1-sigma error bars, but one should state it. The authors should preferably report a 2-sigma error bar than state that they have a 96\% CI, if the hypothesis of Normality of errors is not verified.
		\item For asymmetric distributions, the authors should be careful not to show in tables or figures symmetric error bars that would yield results that are out of range (e.g. negative error rates).
		\item If error bars are reported in tables or plots, The authors should explain in the text how they were calculated and reference the corresponding figures or tables in the text.
	\end{itemize}
	
	\item {\bf Experiments Compute Resources}
	\item[] Question: For each experiment, does the paper provide sufficient information on the computer resources (type of compute workers, memory, time of execution) needed to reproduce the experiments?
	\item[] Answer: \answerNA{} %\answerTODO{} % Replace by \answerYes{}, \answerNo{}, or \answerNA{}.
	\item[] Justification: The paper does not include experiments. % \answerNA{}%\justificationTODO{}
	\item[] Guidelines:
	\begin{itemize}
		\item The answer NA means that the paper does not include experiments.
		\item The paper should indicate the type of compute workers CPU or GPU, internal cluster, or cloud provider, including relevant memory and storage.
		\item The paper should provide the amount of compute required for each of the individual experimental runs as well as estimate the total compute. 
		\item The paper should disclose whether the full research project required more compute than the experiments reported in the paper (e.g., preliminary or failed experiments that didn't make it into the paper). 
	\end{itemize}
	
	\item {\bf Code Of Ethics}
	\item[] Question: Does the research conducted in the paper conform, in every respect, with the NeurIPS Code of Ethics \url{https://neurips.cc/public/EthicsGuidelines}?
	\item[] Answer: \answerYes{} % \answerTODO{} % Replace by \answerYes{}, \answerNo{}, or \answerNA{}.
	\item[] Justification: The paper conforms with the NeurIPS Code of Ethics. % \justificationTODO{}
	\item[] Guidelines:
	\begin{itemize}
		\item The answer NA means that the authors have not reviewed the NeurIPS Code of Ethics.
		\item If the authors answer No, they should explain the special circumstances that require a deviation from the Code of Ethics.
		\item The authors should make sure to preserve anonymity (e.g., if there is a special consideration due to laws or regulations in their jurisdiction).
	\end{itemize}

	\item {\bf Broader Impacts}
	\item[] Question: Does the paper discuss both potential positive societal impacts and negative societal impacts of the work performed?
	\item[] Answer: \answerNA{} % \answerTODO{} % Replace by \answerYes{}, \answerNo{}, or \answerNA{}.
	\item[] Justification: There is no societal impact of the work performed, since the work is mainly theoretical.
	 % \justificationTODO{}
	\item[] Guidelines:
	\begin{itemize}
		\item The answer NA means that there is no societal impact of the work performed.
		\item If the authors answer NA or No, they should explain why their work has no societal impact or why the paper does not address societal impact.
		\item Examples of negative societal impacts include potential malicious or unintended uses (e.g., disinformation, generating fake profiles, surveillance), fairness considerations (e.g., deployment of technologies that could make decisions that unfairly impact specific groups), privacy considerations, and security considerations.
		\item The conference expects that many papers will be foundational research and not tied to particular applications, let alone deployments. However, if there is a direct path to any negative applications, the authors should point it out. For example, it is legitimate to point out that an improvement in the quality of generative models could be used to generate deepfakes for disinformation. On the other hand, it is not needed to point out that a generic algorithm for optimizing neural networks could enable people to train models that generate Deepfakes faster.
		\item The authors should consider possible harms that could arise when the technology is being used as intended and functioning correctly, harms that could arise when the technology is being used as intended but gives incorrect results, and harms following from (intentional or unintentional) misuse of the technology.
		\item If there are negative societal impacts, the authors could also discuss possible mitigation strategies (e.g., gated release of models, providing defenses in addition to attacks, mechanisms for monitoring misuse, mechanisms to monitor how a system learns from feedback over time, improving the efficiency and accessibility of ML).
	\end{itemize}
	
	\item {\bf Safeguards}
	\item[] Question: Does the paper describe safeguards that have been put in place for responsible release of data or models that have a high risk for misuse (e.g., pretrained language models, image generators, or scraped datasets)?
	\item[] Answer:  \answerNA{} % \answerTODO{} % Replace by \answerYes{}, \answerNo{}, or \answerNA{}.
	\item[] Justification: The paper poses no such risks.% \justificationTODO{}
	\item[] Guidelines:
	\begin{itemize}
		\item The answer NA means that the paper poses no such risks.
		\item Released models that have a high risk for misuse or dual-use should be released with necessary safeguards to allow for controlled use of the model, for example by requiring that users adhere to usage guidelines or restrictions to access the model or implementing safety filters. 
		\item Datasets that have been scraped from the Internet could pose safety risks. The authors should describe how they avoided releasing unsafe images.
		\item We recognize that providing effective safeguards is challenging, and many papers do not require this, but we encourage authors to take this into account and make a best faith effort.
	\end{itemize}
	
	\item {\bf Licenses for existing assets}
	\item[] Question: Are the creators or original owners of assets (e.g., code, data, models), used in the paper, properly credited and are the license and terms of use explicitly mentioned and properly respected?
	\item[] Answer: \answerNA{} % Replace by \answerYes{}, \answerNo{}, or \answerNA{}.
	\item[] Justification: The paper does not release new assets.%\justificationTODO{}
	\item[] Guidelines:
	\begin{itemize}
		\item The answer NA means that the paper does not use existing assets.
		\item The authors should cite the original paper that produced the code package or dataset.
		\item The authors should state which version of the asset is used and, if possible, include a URL.
		\item The name of the license (e.g., CC-BY 4.0) should be included for each asset.
		\item For scraped data from a particular source (e.g., website), the copyright and terms of service of that source should be provided.
		\item If assets are released, the license, copyright information, and terms of use in the package should be provided. For popular datasets, \url{paperswithcode.com/datasets} has curated licenses for some datasets. Their licensing guide can help determine the license of a dataset.
		\item For existing datasets that are re-packaged, both the original license and the license of the derived asset (if it has changed) should be provided.
		\item If this information is not available online, the authors are encouraged to reach out to the asset's creators.
	\end{itemize}
	
	\item {\bf New Assets}
	\item[] Question: Are new assets introduced in the paper well documented and is the documentation provided alongside the assets?
	\item[] Answer: \answerNA{} % Replace by \answerYes{}, \answerNo{}, or \answerNA{}.
	\item[] Justification: The paper does not release new assets. %\justificationTODO{}
	\item[] Guidelines:
	\begin{itemize}
		\item The answer NA means that the paper does not release new assets.
		\item Researchers should communicate the details of the dataset/code/model as part of their submissions via structured templates. This includes details about training, license, limitations, etc. 
		\item The paper should discuss whether and how consent was obtained from people whose asset is used.
		\item At submission time, remember to anonymize your assets (if applicable). You can either create an anonymized URL or include an anonymized zip file.
	\end{itemize}
	
	\item {\bf Crowdsourcing and Research with Human Subjects}
	\item[] Question: For crowdsourcing experiments and research with human subjects, does the paper include the full text of instructions given to participants and screenshots, if applicable, as well as details about compensation (if any)? 
	\item[] Answer: \answerNA{} % Replace by \answerYes{}, \answerNo{}, or \answerNA{}.
	\item[] Justification:  The paper does not involve crowdsourcing nor research with human subjects. %\justificationTODO{}
	\item[] Guidelines:
	\begin{itemize}
		\item The answer NA means that the paper does not involve crowdsourcing nor research with human subjects.
		\item Including this information in the supplemental material is fine, but if the main contribution of the paper involves human subjects, then as much detail as possible should be included in the main paper. 
		\item According to the NeurIPS Code of Ethics, workers involved in data collection, curation, or other labor should be paid at least the minimum wage in the country of the data collector. 
	\end{itemize}
	
	\item {\bf Institutional Review Board (IRB) Approvals or Equivalent for Research with Human Subjects}
	\item[] Question: Does the paper describe potential risks incurred by study participants, whether such risks were disclosed to the subjects, and whether Institutional Review Board (IRB) approvals (or an equivalent approval/review based on the requirements of your country or institution) were obtained?
	\item[] Answer: \answerNA{} % Replace by \answerYes{}, \answerNo{}, or \answerNA{}.
	\item[] Justification: The paper does not involve crowdsourcing nor research with human subjects. % \justificationTODO{}
	\item[] Guidelines:
	\begin{itemize}
		\item The answer NA means that the paper does not involve crowdsourcing nor research with human subjects.
		\item Depending on the country in which research is conducted, IRB approval (or equivalent) may be required for any human subjects research. If you obtained IRB approval, you should clearly state this in the paper. 
		\item We recognize that the procedures for this may vary significantly between institutions and locations, and we expect authors to adhere to the NeurIPS Code of Ethics and the guidelines for their institution. 
		\item For initial submissions, do not include any information that would break anonymity (if applicable), such as the institution conducting the review.
	\end{itemize}
	
\end{enumerate}